\documentclass[runningheads]{llncs}

\usepackage[T1]{fontenc}
\usepackage{microtype}

\usepackage{enumitem}
\usepackage{hyperref}
\usepackage{graphicx}
\usepackage{wrapfig}
\usepackage{lipsum}

\usepackage{xcolor}

\makeatletter\let\proof\@undefined\let\endproof\@undefined\makeatother

\usepackage{amsmath,amssymb,amsthm}
\usepackage{mathrsfs}
\usepackage{physics}

\usepackage{numprint}
\npthousandsep{\,}

\usepackage{algorithm}
\usepackage[noend]{algpseudocode}

\usepackage{booktabs}

\usepackage{tikz}
\usetikzlibrary{shapes,positioning,arrows,calc,automata,matrix}

\usepackage{authblk}

\newtheorem*{theorem*}{Theorem}
\newtheorem*{definition*}{Definition}
\newtheorem*{prop*}{Proposition}

\newcommand{\todo}[1]{\textcolor{red}{ #1}}

\newcommand{\para}[1]{\smallskip\noindent\textbf{#1}~}
\newcommand{\parag}[1]{\vspace{6pt}\noindent\textbf{#1}~}
\renewcommand{\para}[1]{\subsubsection{#1}}

\newcommand{\bangle}[1]{\langle #1 \rangle}

\newcommand{\half}{\frac{1}{2}}

\newcommand{\diff}{\mathop{}\!d}

\usepackage{amsmath}

\DeclareMathOperator*{\argmin}{arg\,min}
\usepackage{bm}
\renewcommand{\vec}[1]{\bm{#1}}
\newcommand{\normal}{\mathcal N}
\newcommand{\uniform}{\mathbb U}
\newcommand{\discrete}{\text{Discr}}

\newcommand{\reals}{\mathbb R}
\newcommand{\rationals}{\mathbb Q}
\newcommand{\naturals}{\mathbb N}

\newcommand{\expectation}{\mathbb E}

\newcommand{\probability}{\mathbb P}

\newcommand{\sand}{\,\land\,}
\newcommand{\snot}{\lnot\,}
\newcommand{\until}{\mathbf U}
\newcommand{\eventually}{\mathbf F}
\newcommand{\always}{\mathbf G}
\newcommand{\true}{tt}

\newcommand{\kr}{k}
\newcommand{\featf}{\Phi}

\newcommand{\rhon}{\hat{\rho}}

\advance\textwidth2mm 
\advance\hoffset-1mm
\advance\textheight2mm 
\advance\voffset-1mm

\let\llncssubparagraph\subparagraph
\let\subparagraph\paragraph
\usepackage{titlesec}
\let\subparagraph\llncssubparagraph

\titlespacing*{\section}{0pt}{2.5ex plus 1ex minus 1ex}{2ex plus 0.5ex minus 0.5ex}
\titlespacing*{\subsection}{0pt}{2.25ex plus 0.5ex minus 1ex}{1.5ex plus 0.5ex minus 0.5ex}
\titlespacing*{\subsubsection}{0pt}{1ex plus 0.5ex minus 0ex}{1ex}

\title{
Learning Model Checking and the Kernel Trick for Signal Temporal Logic 
on Stochastic Processes}

\author{
Luca Bortolussi\inst{1,2} \and
Giuseppe Maria Gallo\inst{1} \and
Jan K\v ret\'insk\'y\inst{3} \and 
Laura Nenzi\inst{1,4}
}
\institute{
Department of Mathematics and Geoscience, University of Trieste, Italy \and
Modelling and Simulation Group, Saarland University, Germany \and
Technical University of Munich, Germany \and
University of Technology, Vienna, Austria
}

\date{}
\renewcommand\date[1]{}

\begin{document}


\maketitle

\vspace*{-1em}

\begin{abstract}
We introduce a similarity function on formulae of signal temporal logic (STL).
It comes in the form of a \emph{kernel function}, well known in machine learning as a conceptually and computationally efficient tool.
The corresponding \emph{kernel trick} allows us to circumvent the complicated process of feature extraction, i.e. the (typically manual) effort to identify the decisive properties of formulae so that learning can be applied.
We demonstrate this consequence and its advantages on the task of \emph{predicting (quantitative) satisfaction} of STL formulae on stochastic processes:
Using our kernel and the kernel trick, we learn
(i)~computationally efficiently
(ii)~a practically precise predictor of satisfaction,
(iii)~avoiding the difficult task of finding a way to explicitly turn formulae into vectors of numbers in a sensible way.
We back the high precision we have achieved in the experiments by a theoretically sound PAC guarantee, ensuring our procedure efficiently delivers a close-to-optimal predictor.

\end{abstract}

\vspace*{-1em}

\medskip

\section{Introduction}
\label{sec:intro}

\medskip

\emph{Is it possible to predict the probability that a system satisfies a property \emph{without knowing or executing} the system, solely based on previous experience with the system behaviour w.r.t.\ some \emph{other} properties?
More precisely, let $\probability_M[\varphi]$ denote 
the probability that a (linear-time) property $\varphi$ holds on a run of a stochastic process $M$.
Is it possible to predict $\probability_M[\varphi]$ knowing only $\probability_M[\psi_i]$ for properties $\psi_1,\ldots,\psi_k$, which were randomly chosen (a-priori, not knowing $\varphi$) and thus do not necessarily have any logical relationship, e.g.\ implication, to $\varphi$?
}

While this question cannot be in general answered with complete reliability, we show that in the setting of signal temporal logic, under very mild assumptions, it can be  answered with high accuracy and low computational costs.
\smallskip

\para{Probabilistic verification and its limits.}
Stochastic processes form a natural way of capturing systems whose future behaviour is determined at each moment by a unique (but possibly unknown) probability measure over the successor states.
The vast range of applications includes not only engineered systems such as software with probabilistic instructions or cyber-physical systems with failures, but also naturally occurring systems such as biological systems.
In all these cases, predictions of the system behaviour may be required even in cases the system is not (fully) known or is too large.
For example, consider a safety-critical cyber-physical system with a third party component, or a complex signalling pathway to be understood and medically exploited.

\emph{Probabilistic model checking}, e.g.~\cite{baier2008principles}, provides a wide repertoire of analysis techniques, in particular to determine the probability $\probability_M[\varphi]$ that the system $M$ satisfies the logical formula $\varphi$.
However, there are two caveats.
Firstly, despite recent advances \cite{handbook} the scalability is still quite limited, compared to e.g. hardware or software verification.
Moreover, this is still the case even if we only require \emph{approximate} answers, i.e., for a given precision $\varepsilon$, to determine $v$ such that $\probability_M[\varphi]\in [v-\varepsilon,v+\varepsilon]$.
Secondly, knowledge of the model $M$ is required to perform the analysis. 

\emph{Statistical model checking} \cite{younes02} fights these two issues at an often acceptable cost of relaxing the guarantee to \emph{probably approximately} correct (PAC), requiring that the approximate answer of the analysis may be incorrect with probability at most $\delta$.
This allows for a statistical evaluation:
Instead of analyzing the model, we evaluate the satisfaction of the given formula on a number of observed runs of the system, and derive a statistical prediction, which is valid only with some confidence.
Nevertheless, although $M$ may be unknown, it is still necessary to execute the system in order to obtain its runs.

\emph{``Learning'' model checking} is a new paradigm we propose, in order to fill in a hole in the model-checking landscape where very little access to the system is possible.
We are given a set of input-output pairs for model checking, i.e., a~collection $\{(\psi_i,p_i)\}_i$ of formulae and their satisfaction values on a given model $M$, where $p_i$ can be the probability $\probability_M[\psi_i]$ of satisfying $\psi_i$, or its robustness (in case of real-valued logics), or any other quantity.
From the data, we learn a predictor for the model checking problem: a classifier for Boolean satisfaction, or a regressor for quantitative domains of $p_i$.
Note that apart from the results on the a-priori given formulae, no knowledge of the system is required; also, no runs are generated and none have to be known.
As an example consequence, a user can investigate properties of a system even before buying it, solely based on producer's guarantees on the standardized formulae $\psi_i$.

\emph{Advantages of our approach} can be highlighted as follows, not intending to replace standard model checking in standard situations but focusing on the case of extremely limited (i) information and (ii) online resources.
\emph{Probabilistic} model checking re-analyzes the system for every new property on the input; 
\emph{statistical} model checking can generate runs and then, for every new property, analyzes these runs; 
\emph{learning} model checking performs one analysis with complexity dependent only on the size of the data set (a-priori formulae) and then, for every new formula on input, only evaluates a simple function (whose size is again independent of the system and the property, and depends only on the data set size).
Consequently, it has the least access to information and the least computational demands.
While lack of any guarantees is typical for machine-learning techniques and, in this context with the lowest resources required, expectable, yet we provide PAC guarantees.

\parag{Technique and our approach.}
To this end, we show how to efficiently learn on the space of temporal formulae via the so-called \emph{kernel trick}, e.g.~\cite{shawe2004kernel}.
This in turn requires to introduce a mapping of formulae to vectors (in a Hilbert space) that preserves the information on the formulae.
\emph{How to transform a formula into a vector of numbers (of always the same length)?}
While this is not clear at all for finite vectors, we 
take the dual perspective on formulae, namely as functionals mapping trajectories to values. 
This point of view provides us with a large bag of functional analysis tools \cite{brezis2010functional} and allows us to define the needed semantic similarity of two formulae (the inner product on the Hilbert space).

\parag{Application examples.}
Having discussed the possibility of learning model checking, the main potential of our kernel (and generally introducing kernels for any further temporal logics) is that it opens the door to \emph{efficient learning on formulae} via kernel-based machine-learning techniques~\cite{murphy2012machine,rasmussen:williams:2006}.
Let us sketch a few further applications that immediately suggest themselves: \vspace*{-0.5em}
\begin{description}
    \item[Game-based synthesis] 
    Synthesis with temporal-logic specifications can often be solved via games on graphs \cite{DBLP:conf/cav/MeyerSL18,DBLP:journals/corr/abs-1904-07736}.
    However, exploration of the game graph and finding a winning strategy is done by graph algorithms ignoring the logical information. 
    For instance, choosing between $a$ and $\neg a$ is tried out blindly even for specifications that require us to visit $a$s.
    Approaches such as \cite{DBLP:conf/atva/KretinskyMM19} demonstrate how to tackle this but hit the barrier of inefficient learning of formulae.
    Our kernel will allow for learning reasonable choices from previously solved games.
    \item[Translating, sanitizing and simplifying specifications]
    A formal specification given by engineers might be somewhat different from their actual intention.
    Using the kernel, we can, for instance, find the closest simple formula to their inadequate translation from English to logic, which is then likely to match better. 
    (Moreover, the translation would be easier to automate by natural language processing since learning from previous cases is easy once the kernel gives us an efficient representation for formulae learning.)
    \item[Requirement mining] A topic which received a lot of attention recently is that of identifying 
    specifications from observed data, i.e. to tightly characterize a set of observed behaviours or anomalies  \cite{bartocci2018specification}. Typical methods are using either formulae templates  \cite{BBS14} or methods based e.g. on decision trees \cite{bombara2016} or genetic algorithms \cite{NenziSBB18}.  Our kernel opens a different strategy to tackle this problem: lifting the search problem from the discrete combinatorial space of syntactic structures of formulae to a continuous space in which distances preserve semantic similarity (using e.g. kernel PCA \cite{murphy2012machine} to build finite-dimensional embeddings of formulae into $\mathbb{R}^m$).
\end{description}


\vspace*{-0.5em}

\parag{Our main contributions} are the following: \vspace*{-0.5em}
\begin{itemize}
    \item From the technical perspective, we define a kernel function for temporal formulae (of signal temporal logic, see below) and design an efficient way to learn it. 
    This includes several non-standard design choices, improving the quality of the predictor (see Conclusions).
    \item Thereby we open the door to various learning-based approaches for analysis and synthesis and further applications, in particular also to what we call the \emph{learning} model checking.
    \item We demonstrate the efficiency practically on the predicting the expected satisfaction of formulae on stochastic systems.
    We complement the experimental results with a theoretical analysis and provide a PAC bound.
\end{itemize}

\subsection{Related Work}

\parag{Signal temporal logic (STL)}~\cite{Maler2004} is gaining momentum as a requirement specification language for complex systems and, in particular, cyber-physical systems ~\cite{bartocci2018specification}.
STL has been applied in several flavours, from runtime-monitoring~\cite{bartocci2018specification}, falsification problems \cite{FainekosH019} to control synthesis~\cite{HaghighiMBB19}, and recently also within learning algorithms, trying to find a maximally discriminating formula between sets of trajectories~\cite{bombara2016,BBS14}. In these applications, a central role is played by the real-valued quantitative semantics
\cite{donze2013efficient}, measuring robustness of satisfaction.  
Most of the applications of STL have been directed to deterministic (hybrid) systems, with less emphasis on non-deterministic or stochastic ones~\cite{BartocciBNS15}. 

\parag{Metrics and distances}
form another area in which formal methods are  providing interesting tools, in particular logic-based distances between models, like bisimulation metrics for Markov models~\cite{BacciBLM16,DBLP:conf/concur/BacciBLMTB19,DBLP:conf/aaai/AmortilaBPP19}, which are typically based on a branching logic. In fact, extending these ideas to linear time logic is hard \cite{jan2016linear}, and typically requires statistical approximations. 
Finally, another relevant problem is how to measure the distance between two logic formulae, thus giving a metric structure to the formula space, a task relevant for learning which received little attention for STL, with the notable exception of \cite{madsen2018metrics}.

\parag{Kernels}
make it possible to work in a \emph{feature space} of a higher dimension without increasing the computational cost. Feature space, as  used in machine learning \cite{rasmussen:williams:2006,comaniciu2002mean}, refers to an $n$-dimensional real space that is the co-domain of a mapping from the original space of data. The idea is to map the original space in a new one that is easier to work with. 
The so-called \emph{kernel trick}, e.g.~\cite{shawe2004kernel} allows us to efficiently perform approximation and learning tasks over the feature space without explicitly constructing it. 
We provide the necessary background information in Section~\ref{sec:background:kernel}.

\parag{\emph{Overview of the paper:}}
Section \ref{sec:background} recalls STL and the classic kernel trick.
Section~\ref{sec:overview} provides an overview of our technique and results. 
Section~\ref{sec:kernelSTL} then discusses all the technical development in detail.
In Section~\ref{sec:exp}, we experimentally evaluate the accuracy of our learning method. 
In Section~\ref{sec:conc}, we conclude with future work.
For space reasons, some technical proofs, further details on the experiments, and additional quantitative evidence for our respective conclusions had to be moved to Appendix.
While this material is not needed to understand our method, constructions, results and their analysis, we believe the extra evidence may be interesting for the questioning readers.

\section{Background}
\label{sec:background}

Let $\reals,\reals_{\geq 0},\rationals,\naturals$ denote the sets of  non-negative real, rational, and (positive) natural numbers, respectively.
For vectors $\vec x,\vec y\in\reals^n$ (with $n\in\naturals)$, we write $\vec x=(x_1,\ldots,x_n)$ to access the components of the vectors, in contrast to sequences of vectors $\vec {x_1},\vec {x_2},\ldots\in\reals^n$.
Further, we write $\bangle{\vec x,\vec y}=\sum_{i=1}^n x_i y_i$ for the scalar product of vectors. 
  
\subsection{Signal Temporal Logic} 
\label{sec:STL}

\parag{Signal Temporal Logic (STL)}~\cite{Maler2004} is a linear-time temporal logic suitable to monitor properties of  trajectories. 
A \emph{trajectory} is a function $\xi: I\to D$ with a \textit{time domain} $I\subseteq\reals_{\geq 0}$, 
and a \textit{state space} $D\subseteq\reals^n$ for some $n\in\naturals$.
We define the \textit{trajectory space}  $\mathcal T$ as the set of all possible continuous functions\footnote{The whole framework can be easily relaxed to piecewise continuous c\`adl\`ag trajectories endowed with the Skorokhod topology and metric \cite{billingsley2008probability}.} over $D$. 
An \emph{atomic predicate} of STL is a continuous computable predicate\footnote{Results are easily generalizable to predicates defined by piecewise continuous c\`adl\`ag functions.} on $\vec x\in\reals^n$ of the form of $f(x_1, ..., x_n) \geq 0$, typically linear, i.e. $\sum_{i=1}^n q_ix_i\geq 0$ for $q_1,\ldots,q_n\in\rationals$.
  
\parag{Syntax.}  
The set $\mathcal P$ of STL formulae is given by the following syntax: \vspace*{-0.75em}
   	\[\varphi:=\true\mid\pi\mid\snot\varphi\mid \varphi_1\sand\varphi_2\mid\varphi_1\until_{[a, b]}\varphi_2\]\vspace*{-1.75em}
   	
   	\noindent
  	where $\true$ is the Boolean \textit{true} constant, $\pi$ ranges over atomic predicates, {\it negation} $\snot$ and {\it conjunction} $\sand$  are the standard Boolean connectives and $\until_{[a, b]}$ is the \textit{until} operator, with $a, b\in\rationals$ and $a<b$.  As customary, we can derive the {\it disjunction} operator $\vee$ by De Morgan's law and the   {\it eventually} (a.k.a.\ future) operator $\eventually_{[t_{1},t_{2}]}$ and  the {\it always} (a.k.a.\ globally) operator $\always_{[t_{1},t_{2}]}$ operators from  the until operator. 
  	
\parag{Semantics.}  	
STL can be given not only the classic Boolean notion of 
\emph{satisfaction}, denoted by $s(\varphi,\xi,t) = 1$ if  $\xi$ at time $t$ satisfies $\varphi$, and $0$ otherwise, but also
a quantitative one, denoted by $\rho(\varphi, \xi, t)$.
This measures the quantitative level of satisfaction of a formula for a given trajectory, evaluating how ``robust'' is the satisfaction of $\varphi$ with respect to perturbations in the signal~\cite{donze2013efficient}. The quantitative semantics is defined recursively as follows: \vspace*{-0.75em}
  	\begin{align*}
  	&\rho(\pi,\xi,t) &=& f_\pi(\xi(t)) \qquad \text{for } \pi(x_1,...,x_n)=\big(f_\pi(x_1,...,x_n)\geq 0\big)\\
  	&\rho(\snot\varphi,\xi,t) &=& -\rho(\varphi,\xi,t)\\
  	&\rho(\varphi_1\sand\varphi_2,\xi,t) &=& \min\big(\rho(\varphi_1,\xi,t), \rho(\varphi_2,\xi,t)\big)\\
  	&\rho(\varphi_1\until_{[a, b]}\varphi_2,\xi,t) &=& \max_{t'\in[a+t,b+t]}\big(\min\big(\rho(\varphi_2,\xi,t'),\min_{t''\in[t,t']}\rho(\varphi_1,\xi,t'')\big)\big)
  	\end{align*}
\vspace*{-1.75em}

 \parag{Soundness and Completeness} 
 Robustness is compatible with satisfaction in that it complies with the following soundness property: if $\rho(\varphi, \xi, t) > 0$ then $s(\varphi,\xi,t) = 1$; and if $\rho(\varphi, \xi, t) < 0$ then $s(\varphi,\xi,t) = 0$. 
 If the robustness is $0$, both satisfaction and the opposite may happen, but either way only non-robustly: there are arbitrarily small perturbations of the signal so that the satisfaction changes.
In fact, it complies also with a completeness property that $\rho$ measures how robust the satisfaction of a trajectory is with respect to perturbations, see~\cite{donze2013efficient} for more detail.
  	

\parag{Stochastic process} in this context is a probability space $\mathcal M = (\mathcal T, \mathcal A, \mu)$, where $\mathcal T$ is a trajectory space and $\mu$ is a probability measure on a $\sigma$-algebra $\mathcal A$ over $\mathcal T$.
Note that the definition is essentially equivalent to the standard definition of a~stochastic process as a collection $\{D_t\}_{t\in I}$ of random variables, where $D_t(\xi)\in D$ is the signal $\xi(t)$ at time $t$ on $\xi$ \cite{billingsley2008probability}.
The only difference is that we require, for simplicity\footnote{Again, this assumption can be relaxed since continuous functions are dense in the Skorokhod space of c\`adl\`ag functions.}, the signal be continuous.

\parag{Expected robustness and satisfaction probability.}
Given a stochastic process $\mathcal M = (\mathcal T, \mathcal A, \mu)$, we define the  \textit{expected robustness} $R_{\mathcal M}:\mathcal P\times I\to\reals$ as 
\[R_{\mathcal M}(\varphi, t) := \expectation_{\mathcal M}[\rho(\varphi, \xi, t)] = \int_{\xi\in\mathcal T}\rho(\varphi,\xi,t)d\mu(\xi)\,.\] 
The qualitative counterpart of the expected robustness is the \textit{satisfaction probability} $ S(\varphi)$, i.e. the probability that a trajectory generated by the stochastic process $\mathcal M$ satisfies the formula $\varphi$, i.e.\ $S_{\mathcal M}(\varphi,t):= \expectation_{\mathcal M}[s(\varphi, \xi, t)] = \int_{\xi\in\mathcal T}s(\varphi,\xi,t)d\mu(\xi)$.%
\footnote{As argued above, this is essentially equivalent to 
integrating the indicator function of robustness being positive
since a formula has robustness exactly zero only with probability zero as we sample all values from continuous distributions.} 
Finally, when $t=0$ we often drop the parameter $t$ from all these functions.


\subsection{Kernel Crash Course}\label{sec:background:kernel}
\label{sec:kernel}
We recall the necessary background for readers less familiar with machine learning.

\parag{Learning linear models.}
Linear predictors take the form of a vector of weights, intuitively giving positive and negative importance to features.
A predictor given by a vector $\vec w=(w_1,\ldots,w_d)$ evaluates a data point $\vec x=(x_1,\ldots,x_d)$ to $w_1x_1+\cdots+w_dx_d=\bangle{\vec w,\vec x}$.
To use it as a classifier, we can, for instance, take the sign of the result and output yes iff it is positive;
to use it as a regressor, we can simply output the value.
During learning, we are trying to separate, respectively approximate, the training data $\vec {x_1},\ldots \vec{x_N}$ with a linear predictor, which corresponds to solving an optimization problem of the form
\[\argmin_{\vec{w}\in\reals^d} f(\bangle{\vec w,\vec{x_1}},\ldots,\bangle{\vec w,\vec {x_N}},\bangle{\vec w,\vec w})\]
where the possible, additional last term comes from regularization (preference of simpler weights, with lots of zeros in $\vec w$).

\parag{Need for a feature map $\Phi:\mathit{Input}\to\reals^n$.} 
In order to learn, the input object first needs to be transformed to a vector of numbers.
For instance, consider learning the logical exclusive-or function (summation in $\mathbb Z_2$) $y=x_1\oplus x_2$.
Seeing true as 1 and false as 0 already transforms the input into elements of $\reals^2$.
However, observe that there is no linear function separating sets of points $\{(0,0),(1,1)\}$ (where xor returns true) and $\{(0,1),(1,0)\}$ (where xor returns false).
In order to facilitate learning by linear classifiers, richer feature space may be needed than what comes directly with the data.
In our example, we can design a feature map to a higher-dimensional space using $\Phi:(x_1,x_2)\mapsto(x_1,x_2,x_1\cdot x_2)$.
Then e.g.\ $x_3\leq \frac{x_1+x_2-1}2$ holds in the new space iff $x_1\oplus x_2$ and we can learn this linear classifier.


\begin{wrapfigure}{r}{0.5\textwidth}
\vspace*{-0.2cm}
\includegraphics[scale=0.16]{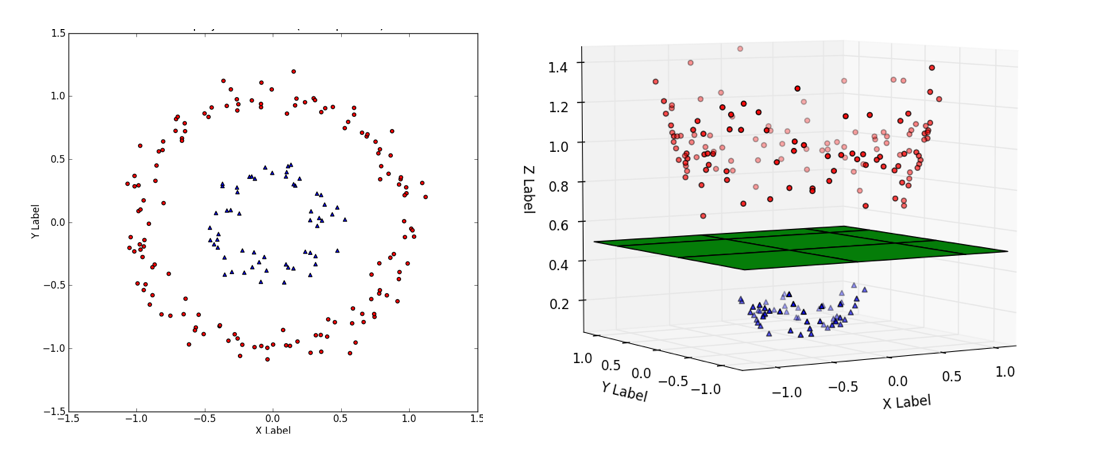}
\vspace*{-1cm}
\caption{An example illustrating the need for feature maps in linear classification \cite{kernel-trick-pic}.}
\label{fig:higher-dim}
\vspace*{-0.8cm}
\end{wrapfigure}
 
Another example can be seen in Fig.~\ref{fig:higher-dim}.
The inner circle around zero cannot be linearly separated from the outer ring.
However, considering $x_3:=x_1^2+x_2^2$ as an additional feature turns them into easily separable lower and higher parts of a paraboloid.

In both examples, a feature map $\Phi$ mapping the input to a space with higher dimension ($\reals^3$), was used.
Nevertheless, two issues arise:
\begin{enumerate}
    \item What should be the features? Where do we get good candidates?
    \item How to make learning efficient if there are too many features?
\end{enumerate}
On the one hand, identifying the right features is hard, so we want to consider as many as possible.
On the other hand, their number increases the dimension and thus decreases the efficiency both computationally and w.r.t.~the number of samples required.
 
\parag{Kernel trick.}
Fortunately, there is a way to consider a huge amount of features, but with efficiency independent of their number (and dependent only on the amount of training data)!
This is called the kernel trick.
It relies on two properties of linear  classifiers:
\begin{itemize}
    \item The optimization problem above, after the feature map is applied, takes the form \vspace*{-0.5em}
    \[\argmin_{\vec{w}\in\reals^n} f\big(\bangle{\vec w,\featf(\vec {x_1})},\ldots,\bangle{\vec w,\featf(\vec {x_N})},\bangle{\vec w,\vec w}\big)\] \vspace*{-0.5em}
    
    \item Representer theorem: The optimum of the above can be written in the form \vspace*{-0.5em} \[\vec w^*=\sum_{i=1}^N \alpha_i\featf(\vec {x_i})\] \vspace*{-0.5em}
    
    \noindent
    Intuitively, anything orthogonal to training data cannot improve precision of the classification on the training data, and only increases $||\vec w||$, which we try to minimize (regularization).
\end{itemize}
Consequently, plugging the latter form into the former optimization problem yields an optimization problem of the form:
\[\argmin_{\vec{\alpha}\in\reals^N}
g\big(\vec\alpha, \bangle {\featf(\vec {x_i}),\featf(\vec {x_j})}_{1\leq i,j\leq N}\big)
\]
In other words, optimizing weights $\vec \alpha$ of expressions where data only appear in the form $\bangle {\featf(\vec x_i),\featf(\vec x_j)}$.
Therefore, we can take all features in $\featf(\vec x_i)$ into account \emph{if}, at the same time, we can efficiently evaluate the kernel function
\[\kr:(\vec x,\vec y)\mapsto \bangle {\featf(\vec x),\featf(\vec y)}\]
i.e.\ \emph{without} explicitly constructing $\featf(\vec x)$ and $\featf(\vec y)$.
Then we can efficiently learn the predictor on the rich set of features.
Finally, when the predictor is applied to a new point $\vec x$, we only need to evaluate the expression
 \[\bangle{\vec w,\featf(\vec x)}=\sum_{i=1}^N \alpha_i \bangle{\featf(\vec {x_i}),\featf(\vec x)}=\sum_{i=1}^N \alpha_i \kr(\vec{x_i},\vec x)\]

\section{Overview of Our Approach and Results}\label{sec:overview}
In this section, we describe what our tasks are if we want to apply the kernel trick in the setting of temporal formulae, what our solution ideas are, and where in the paper they are fully worked out.

\begin{enumerate}
     \item
\emph{Design the kernel function:} define a similarity measure for STL formulae and prove it takes the form $\bangle{\featf(\cdot),\featf(\cdot)}$ 
 \begin{enumerate}
     \item
     \emph{Design an embedding of formulae into a Hilbert space (vector space with possibly infinite dimension) (Thm.~\ref{thm:hilbert} in 
     App.~\ref{app:STLkernel} proves this is well defined):}
     Although learning can be applied also to data with complex structure such as graphs, the underlying techniques typically work on vectors.
     How do we turn a formula into a vector?
     
     Instead of looking at the syntax of the formula, we can look at its semantics.
     Similarly to Boolean satisfaction, where a formula can be identified with its language, i.e., the set $\mathcal T\to 2\cong 2^{\mathcal T}$ of trajectories that satisfy it, we can regard an STL formula $\varphi$ as a map $\rho(\varphi,\cdot):\mathcal T\to\reals\cong \reals^{\mathcal T}$ of trajectories to their robustness. 
     Observe that this is a real function, i.e., an \emph{infinite-dimensional} vector of reals.
     Although explicit computations with such objects are problematic, kernels circumvent the issue.
     In summary, we have the implicit features given by the map:
     \[\varphi\stackrel \featf\mapsto \rho(\varphi,\cdot)\]
     
     \item
     \emph{Design similarity on the feature representation (
     in Sec.~\ref{par:kernel}):}
     Vectors' similarity is typically captured by their scalar product $\bangle{ \vec x,\vec y}=\sum_{i} x_i y_i$ since it gets larger whenever the two vectors ``agree'' on a component.
     In complete analogy, we can define for infinite-dimensional vectors (i.e.\ functions) $f,g$ their ``scalar product'' $\langle f,g\rangle=\int f(x)g(x)\diff x$.
     Hence we want the kernel to be defined as
     \begin{equation*}
     \kr(\varphi,\psi)=
     \langle\rho(\varphi,\cdot),\rho(\psi,\cdot)\rangle=
     \int_{\xi\in\mathcal T} \rho(\varphi,\xi)\rho(\psi,\xi) \diff\xi   
     \end{equation*}

     \item 
     \emph{Design a measure on trajectories (Sec.~\ref{par:trajectory space.}):}
     Compared to finite-dimensional vectors, where in the scalar product each component is taken with equal weight, integrating over uncountably many trajectories requires us to put a finite measure on them, according to which we integrate.
     Since, as a side effect, it necessarily expresses their importance, we define a probability measure $\mu_0$ preferring ``simple'' trajectories, where the signals do not change too dramatically (the so-called total variation is low).
     This finally yields the definition of the kernel as\footnote{On the conceptual level; technically, additional normalization and Gaussian transformation are performed to ensure usual desirable properties, 
     see Cor.~\ref{cor:norm_exp_kernel} 
     in Sec.~\ref{par:kernel}.}
     \begin{center}
     \vspace*{6pt}
     \noindent\framebox[0.7\textwidth]{\parbox{0.6\textwidth}{
        \begin{equation}
            \kr(\varphi,\psi)= \int_{\xi\in\mathcal T} \rho(\varphi,\xi)\rho(\psi,\xi) \diff\mu_0(\xi)\qquad\label{eq:kernel}
        \end{equation}
     }}
     \vspace*{6pt}
     \end{center}
\end{enumerate}     
     
     \item 
     \emph{Learn the kernel (Sec.~\ref{subsec:setting}):}
\begin{enumerate}
    \item 
    \emph{Get training data $\vec {x_i}$:}
    The formulae for training should be chosen according to the same distribution as they are coming in the final task of prediction.
    Since that distribution is unknown, we assume at least a general preference of simple formulae and thus design a probability distribution $\mathcal F_0$, preferring formulae with simple syntax trees (see Section \ref{subsec:setting}).
    We also show that several hundred formulae are sufficient for practically precise predictions.
    
    \item
    \emph{Compute the ``correlation'' of the data $\bangle{\phi(\vec {x_i}),\phi(\vec {x_j})}$ by kernel $\kr(\vec {x_i},\vec {x_j})$:}
    Now we evaluate (\ref{eq:kernel}) for all the data pairs.
    Since this involves an integral over all trajectories, we simply approximate it by Monte Carlo: We choose a number of trajectories according to $\mu_0$ and sum the values for those.
    In our case, 10\,000 provide a very precise approximation.
    
    \item 
    \emph{Optimize the weights $\vec\alpha$ (using values from (b)
    ):} 
    Thus we get the most precise linear classifier given the data, but penalizing too ``complicated'' ones since they tend to overfit and not generalize well (so-called regularization).
    Recall that the dimension of $\vec\alpha$ is the size of the training data set, not the infinity of the Hilbert space.
\end{enumerate}     

     \item 
     \emph{Evaluate the predictive power of the kernel} and thus implicitly the kernel function design:
     \begin{itemize}
         \item We evaluate the accuracy of predictions of robustness for single trajectories (Sec.~\ref{sec:exp:single}), the expected robustness on a stochastic system 
         and the corresponding Boolean notion of satisfaction probability (Sec.~\ref{sec:exp:sat}).
         Moreover, we show that there is no need to derive kernel for each stochastic process separately depending on their probability spaces, but the one derived from the generic $\mu_0$ is sufficient and, surprisingly, even more accurate (Sec.~\ref{sec:exp:processes}).
         
         \item
         Besides the experimental evaluation, we provide a PAC bound on our methods in terms of Rademacher complexity \cite{foundationsML2018} (Sec.~\ref{par:pac}).
     \end{itemize}
 \end{enumerate}
\section{A Kernel for Signal Temporal Logic}
\label{sec:kernelSTL}
In this section, we sketch the technical details of the construction of the STL kernel, of the correctness proof, and of PAC learning bounds. More details on the definition, including proofs, are provided in Appendix~\ref{app:STLkernel}.


  
\subsection{Definition of  STL Kernel}
\label{par:kernel}
  Let us fix a formula $\varphi\in\mathcal{P}$ in the STL formulae space and consider the robustness $\rho(\varphi, {}\cdot{}, {}\cdot{}):\mathcal{T}\times I\to\reals$,
  seen as a real-valued function on the domain $\mathcal{T}\times I$, where $I\subset\reals$ is a bounded interval, and  $\mathcal{T}$ is the trajectory space of continuous functions. The STL kernel is defined as follows. 
    
  \begin{definition}\label{def:STLkernel}
    Fixing a probability measure $\mu_0$ on $\mathcal{T}$, we define the STL-kernel \vspace*{-0.5em}
  \begin{eqnarray*}
  	\kr'(\varphi, \psi)  & = \int_{\xi\in\mathcal T}\int_{t\in I}\rho(\varphi, \xi, t)\rho(\psi, \xi, t)dtd\mu_0\,
  	\end{eqnarray*}
 \end{definition}
 \vspace*{-0.5em}

\noindent The integral is well defined as it corresponds to a scalar product in a suitable Hilbert space of functions.  Formally proving this, and leveraging foundational results on kernel functions \cite{foundationsML2018}, in Appendix~\ref{app:STLkernel} we prove the following: \vspace*{-0.3em}

\begin{theorem}\label{th:proper_kernel}
The function $\kr'$ is a proper kernel function.
  \end{theorem}

In the previous definition, we can fix time to $t=0$ and remove the integration w.r.t. time. This simplified version of the kernel is called \emph{untimed}, to distinguish it from the \emph{timed} one introduced above. 
    
 \smallskip

In the rest of the paper, we mostly work with two derived kernels,  $\kr_0$ and $\kr$:

\begin{equation}
\small
\kr_0(\varphi, \psi) = \frac{\kr'(\varphi, \psi)}{\sqrt{\kr'(\varphi, \varphi)\kr'(\psi, \psi)}}\;\;\;\;\;\;\;\;\;
    \kr(x, y) = \exp\left(-\frac{1-2\kr_0(x,y)}{\sigma^2}\right).
\end{equation}
The normalized kernel $\kr_0$ rescales $\kr'$ to guarantee that $\kr(\varphi, \varphi) \geq \kr(\varphi, \psi)\,,\ \forall \varphi, \psi\in\mathcal P$. The Gaussian kernel $\kr$, additionally, allows us to introduce a soft threshold $\sigma^2$ to fine tune the identification of significant similar formulae in order to improve learning.
 The following proposition is straightforward in virtue of the closure properties of kernel functions \cite{foundationsML2018}:
\begin{corollary}\label{cor:norm_exp_kernel}
The functions  $\kr_0$ and $\kr$ are proper kernel functions.
  \end{corollary}

\subsection{The Base Measure $\mu_0$}
\label{par:trajectory space.}

  In order to make our kernel meaningful and not too expensive to compute, we endow the trajectory space $\mathcal T$ with a probability distribution such that more complex trajectories are less probable.  We use the total variation~\cite{pallara2000functions} of a~trajectory\footnote{The total variation of function $f$ defined on  $[a, b]$ is  $V_a^b(f) = \sup_{P\in\mathbf P}\sum_{i=0}^{n_{P}-1}|f(x_{i+1})-f(x_i)|$, where $\mathbf P=\{ P=\{x_{0},\dots ,x_{n_{P}}\}\mid P{\text{ is a partition of }}[a,b]\} $.} and the number of changes in its monotonicity  as indicators of its ``complexity''.

  

  
  Because later we use the probability measure $\mu_0$ for Monte Carlo approximation of the kernel $\kr$, it is advantageous to define $\mu_0$ algorithmically, by providing a \emph{sampling algorithm}.
  The algorithm samples from continuous piece-wise linear functions, a dense subset of $\mathcal{T}$, and is described in detail in Appendix~\ref{app:mu}.  Essentially, we simulate the value of a trajectory at discrete steps $\Delta$, for a total of $N$ steps (equal to 100 in the experiments) by first sampling its total variation distance from a squared Gaussian distribution, and then splitting such total variation in the single steps, changing sign of the derivative at each step with small probability $q$. We then interpolate linearly  between consecutive points of the discretization and make the trajectory continuous piece-wise linear.
 
  In Section \ref{sec:exp:processes}, we show that using this simple measure still allows us to make predictions with remarkable accuracy even for other stochastic processes on $\mathcal{T}$.

\subsection{Normalized Robustness} 
Consider the predicates $x_1 - 10 \geq 0$ and $x_1- 10^7 \geq 0$.
Given that we train and evaluate on $\mu_0$, whose trajectories typically take values in the interval $[-3,3]$ (see also Fig.~\ref{fig:trajectories_base} in App.~\ref{app:mu}), both predicates are essentially equivalent for satisfiability. However, their robustness on the same trajectory differs by orders of magnitude.
This very same effect, on a smaller scale, happens also when comparing $x_1\geq 10$ with $x_1\geq 20$. In order to ameliorate this issue and make the learning less sensitive to outliers, we also consider 
 a {\it normalized robustness}, where we rescale the value of the secondary (output) signal to $(-1,1)$ using a sigmoid function. More precisely, given  an atomic predicate $\pi(x_1,...,x_n)=(f_\pi(x_1,...,x_n)\geq 0)$, we define $\rhon(\pi,\xi,t) = \tanh{(f_\pi(x_1,...,x_n))}$. The other operators of the logic follow the same rules of the standard robustness described in Section~\ref{sec:STL}. 
 Consequently, both $x_1 - 10 \geq 0$ and $x_1- 10^7 \geq 0$ are mapped to very similar robustness for typical trajectories w.r.t. $\mu_0$, thus reducing the impact of outliers.

\subsection{PAC Bounds for the STL Kernel}\label{par:pac}

Probably Approximately Correct (PAC) bounds \cite{foundationsML2018} for learning  provide a bound on  the generalization error on unseen data (known as risk) in terms of the training loss plus additional terms which shrink to zero as the number of samples grows. These additional terms typically depend also on some measure of the complexity of the class of models we consider for learning (the so-called hypothesis space), which ought to be finite. The bound holds with  probability $1-\delta$, where $\delta>0$ can be set arbitrarily small at the price of the bound getting looser.

In the following, we will state a PAC bound for  learning with STL kernels for classification. A bound for regression, and more details on the the classification bound, can be found in Appendix~\ref{sec:app:pac}.
We first recall the definition of the risk $L$ and the empirical risk $\hat{L}$ for classification. The former is an average of the zero-one loss over the data generating distribution $p_{data}$, while the latter averages over a finite sample $D$ of size $m$ of $p_{data}$.  Formally,
\[L(h) = \mathbb{E}_{\varphi \sim p_{data}}\left[ \mathbb{I}\big(h(\varphi) \neq y(\varphi)\big) \right] \;\;\;\text{and}\;\;\; \hat{L}_D(h) = \frac{1}{m}\sum_{i=1}^m \mathbb{I}\big(h(\varphi_i) \neq y(\varphi_i)\big), \]
where $y(\varphi)$ is the actual class (truth value) associated with $\varphi$, in contrast to the predicted class $h(\varphi)$, and $\mathbb{I}$ is the indicator function.

The major issue with PAC bounds for kernels is that we need to constrain in some way the model complexity. This is achieved by requesting the functions that can be learned  have a bounded norm.  We recall that the norm $\|h\|_{\mathbb{H}}$ of a function $h$ obtainable by kernel methods, i.e. $h(\varphi) = \sum_{i=1}^{N}\alpha_i k(\varphi_i,\varphi)$, is   $\|h\|_{\mathbb{H}} = \boldsymbol{\alpha}^T K \boldsymbol{\alpha}$, where $K$ is the Gram matrix (kernel evaluated between all pairs of input points, $K_{ij}=\kr(\varphi_i,\varphi_j)$). 
The following theorem, stating the bounds, can be proved by combining  bounds on the Rademacher complexity for kernels with Rademacher complexity based PAC bounds, as we show in Appendix~\ref{sec:app:pac}.

\begin{theorem}[PAC bounds for Kernel Learning in Formula Space]
\label{th:PAC}
Let $k$ be a kernel (e.g.  normalized, exponential) for STL formulae  $\mathcal{P}$, and fix $\Lambda > 0$. 
Let $y:\mathcal{P}\rightarrow \{-1,1\}$ be a target function to learn as a classification task. Then for any $\delta >0$ and hypothesis function $h$ with $\|h\|_{\mathbb{H}}\leq \Lambda$, with probability at least $1-\delta$ it holds that
\begin{equation}
\label{PACboundClassificationSTL}
 L(h) \leq \hat{L}_D(h) + \frac{\Lambda}{\sqrt{m}} + 3\sqrt{\frac{\log\frac{2}{\delta}}{2m}}.   
\end{equation}
\end{theorem}

The previous theorem gives us a way to control the learning error, provided we restrict the full hypothesis space. Choosing a value of $\Lambda$ equal to 40 (the typical value we found in experiments) and  confidence 95\%, the bound predicts around 650\,000 samples to obtain an accuracy bounded by the accuracy on the training set plus 0.05. This theoretical a-priori bound is much larger than the training set sizes in the order of hundreds, for which we observe good performance in practice.

\section{Experiments}
\label{sec:exp}
We test the performance of the STL kernel in predicting (a) robustness and satisfaction on single trajectories, and (b) expected robustness and satisfaction probability estimated statistically from  $K$ trajectories. 
Besides, we test the kernel on trajectories sampled according to the a-priori base measure $\mu_0$ and according to the respective stochastic models to check the generalization power of the generic $\mu_0$-based kernel. Here we report the main results; for additional details as well as plots and tables for further ways of measuring the error, we refer the interested reader to Appendix \ref{ap:exp}.

Computation of the STL robustness and of the kernel were implemented in Python exploiting PyTorch \cite{paszke2017automatic} for parallel computation on GPUs. All the experiments were run on a AMD Ryzen 5000 with 16 GB of RAM and on a consumer NVidia GTX 1660Ti with 6 GB of DDR6 RAM. We run each experiment 1000 times for single trajectories and 500 for expected robustness and satisfaction probability were we use 5000 trajectories for each run.
Where not  indicated differently, each result is the mean over all experiments. Computational time is fast: the whole process of sampling from $\mu_0$, computing the kernel, doing regression for training, test set of size 1000 and validation set of size 200, takes about 10 seconds on GPU.  
We use the following acronyms: RE = relative error, AE= absolute error,
MRE = mean relative error, MAE = mean absolute error, MSE = mean square error.

\subsection{Setting}
\label{subsec:setting}
To compute the kernel itself, we sampled 10\,000 trajectories from $\mu_0$, using the sampling method described in Section~\ref{par:trajectory space.}.  As regression algorithm (for optimizing $\vec\alpha$ of Sections \ref{sec:background:kernel} and \ref{sec:overview}) we use the  \emph{Kernel Ridge Regression} (KRR)~\cite{murphy2012machine}. KRR was as good as, or superior, to other regression techniques (a comparison can be found in Appendix \ref{app:setting}).

\parag{Training and test set} are composed of M formulae sampled randomly according to the measure $\mathcal F_0$ given by a syntax-tree random recursive growing scheme (reported in detail in Appendix \ref{app:setting}), where  
the root is always an operator node and each node is an atomic predicate with
probability $p_{\mathit{leaf}}$ (fixed in this experiments  to $0.5$), or, otherwise, another operator node (sampling the type using a uniform distribution). In these experiments, we fixed $M = 1000$.

\parag{Hyperparameters.} 
 We vary several hyperparameters, testing their impact on errors and accuracy. 
Here we briefly summarize the results. \\
\noindent - The impact of \emph{formula complexity}: We vary the parameter $p_{\mathit{leaf}}$ in the formula generating algorithm in the range $[0.2,0.3,0.4,0.5]$ (average formula size around $[100,25,10,6]$ nodes in the syntax tree), but only a slight increase in the median relative error is observed for more complex formulae: $[0.045,0.037,0.031,0.028]$. \\
\noindent  - The addition of \emph{time bounds} in the formulae has essentially no impact on the performance in terms of errors. \\
\noindent - There is a very small improvement (<10\%) using \emph{integrating signals w.r.t. time} (timed kernel) vs using only robustness at time zero (untimed kernel), but at the cost of a 5-fold increase in computational training time. \\
\noindent  - \emph{Exponential kernel $\kr$} gives a 3-fold improvement in accuracy w.r.t. normalized kernel $\kr_0$.\\
  \vspace*{-0.4cm}
\begin{wrapfigure}{r}{0.45\textwidth}
\centering
  \vspace*{-1.2cm}
  	\includegraphics[scale=0.41]{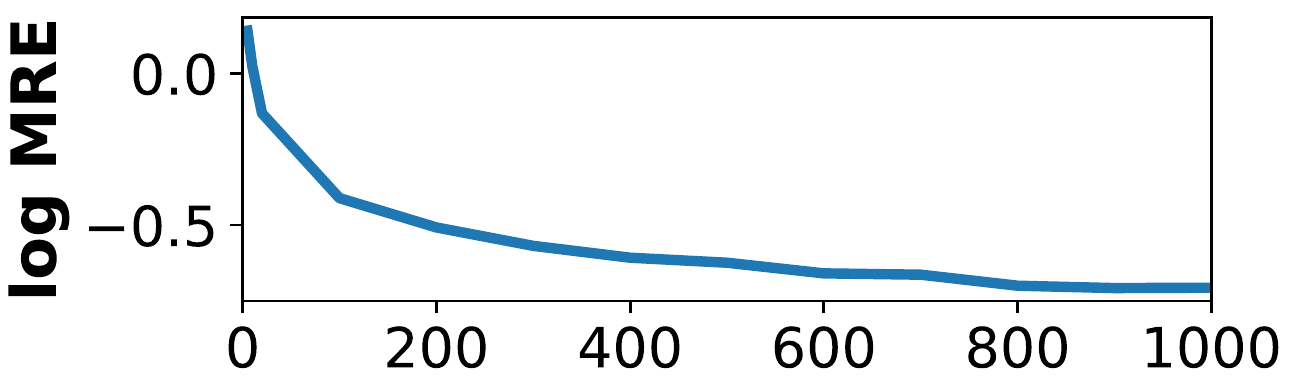}
 	  \vspace*{-0.5cm}
  \caption{\label{fig:training-size}
  MRE of  predicted average robustness vs the size of the training set.}
 	  \vspace*{-0.7cm}
  \end{wrapfigure}  
\noindent - \emph{Size of training set}:   The error in estimating robustness decreases as we increase the amount of training formulae, see Fig.~\ref{fig:training-size}.  However, already for a few hundred formulae, the predictions are quite accurate.\\
\noindent - \emph{Dimensionality of signals}: Error tends to increase linearly with dimensionality. 
For 1000 formulae in the training set, from dimension 1 to 5, MRE is [0.187, 0.248, 0.359,
 0.396, 0.488] and MAE is [0.0537, 0.0735,
 0.0886, 0.098, 0.112].


\begin{figure}[!b]
    \centering
     \includegraphics[width=.47\linewidth]{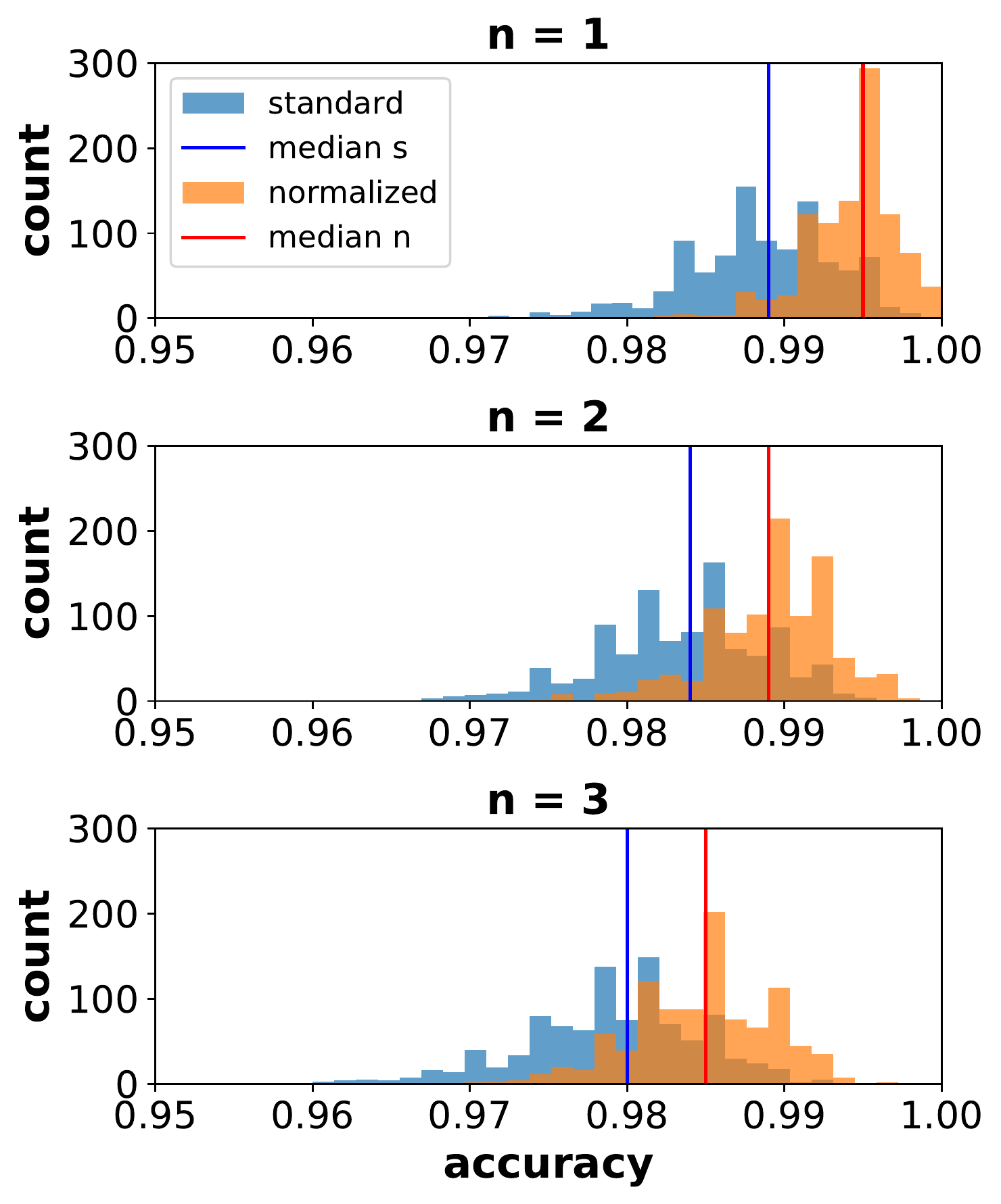}
         \includegraphics[width=.44\linewidth]{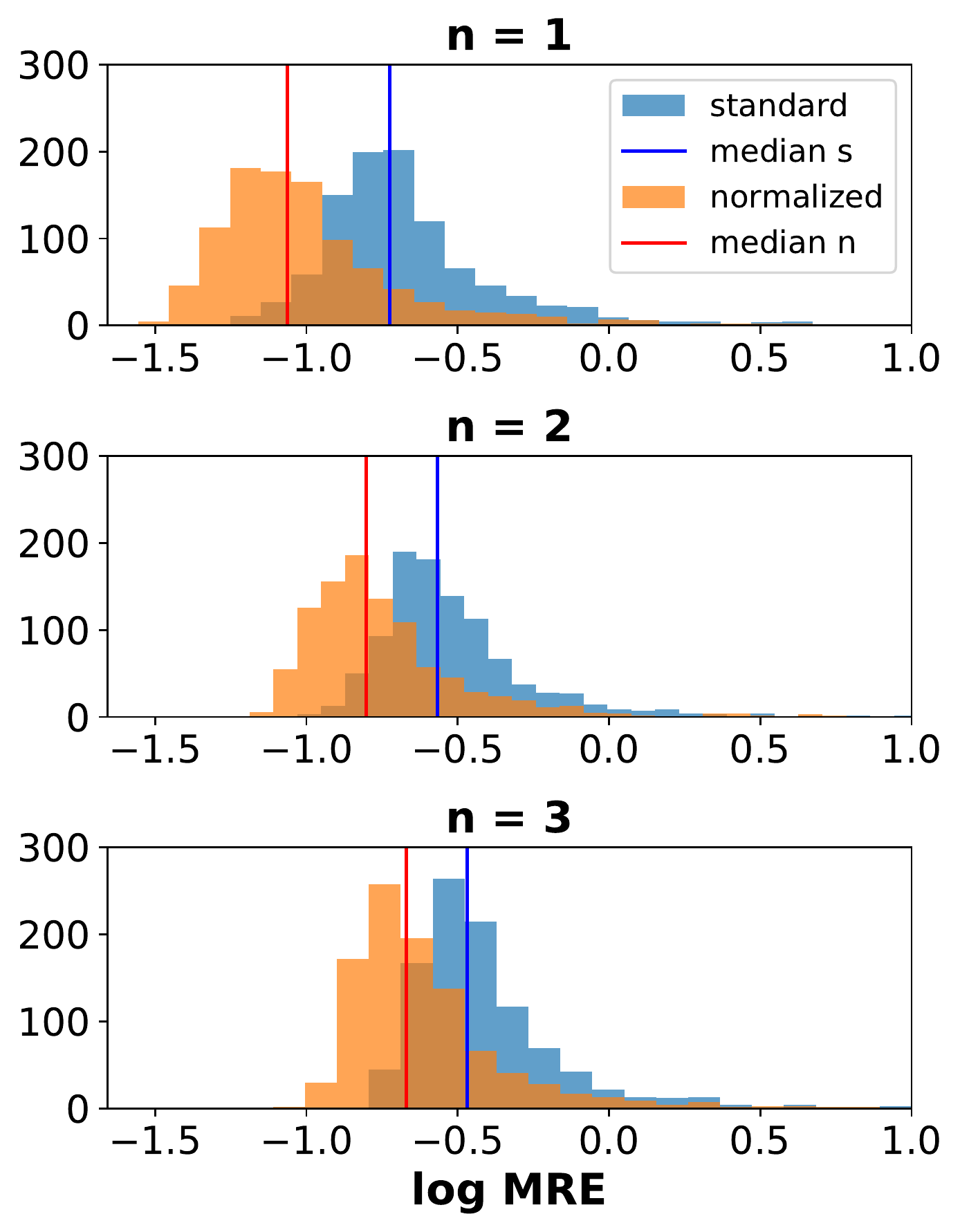}
         \vspace*{-5mm}
    \caption{Accuracy of satisfiability prediction (left)  and  $log_{10}$ of the MRE (right)   over all $1000$ experiments for standard and normalized robustness for samples from $\mu_0$ with dimensionality of signals $n=1,2,3$.
    (Note the logarithmic scale, with log value of -1 corresponding to 0.1 of the standard non-logarithmic scale.)
    }
    \label{fig:acc_mre}
\end{figure}

\subsection{Robustness and Satisfaction on Single Trajectories}\label{sec:exp:single}
In this experiment, we predict the Boolean satisfiability of a formula using as a~discriminator the sign of the robustness.  We generate the training and test set of formulae using $\mathcal F_0$, and
the function sampling trajectories from $\mu_0$  with dimension $n=1,2,3$. We evaluate the standard robustness $\rho$ and the normalized one $\rhon$ of each trajectory for each formula in the training and test sets. We then predict $\rho$ and $\rhon$ for the test set and check if the sign of the predicted robustness agrees with that of the true one, which is a proxy for satisfiability, as discussed previously. Accuracy and distribution of the $\log_{10}$ MRE  over all experiments are reported in Fig. \ref{fig:acc_mre}. 
Results are good for both  but the  normalized robustness performs always better. Accuracy is always greater than 0.96 and gets slightly worse when increasing the dimension.
We report the mean of quantiles of $\rho$ and $\rhon$ for RE and AE for n=3 (the toughest case) in Table~\ref{tab:quantile_all} (top two rows).  Errors for the normalized one are also always lower and slightly worsen when increasing the dimension. 

\newcolumntype{g}{>{\columncolor{Gray}}c}
\newcommand{\ccell}[1]{\begin{tabular}[t]{@{}c@{}}#1\end{tabular}}
   \begin{table}[t]
  	\begin{center}
  		\hspace*{0cm}
  		\caption{Mean of quantiles for RE and AE over all experiments for prediction of the standard and normalized robustness ($\rho$, $\rhon$),  expected robustness ($R$, $\hat{R}$),  
  		the satisfaction probability  (S) with trajectories sampled from $\mu_0$ and signals with dimensionality n=3, and of the normalized expected robustness on trajectories sampled from \emph{Immigration} (1 dim), \emph{Isomerization} (2 dim), and \emph{Transcription} (3 dim)}
  		\label{tab:quantile_all}
\begin{tabular}{lllllll|lllll}
\toprule
  		{} & \multicolumn{6}{c}{relative error (RE) } & \multicolumn{5}{c}{absolute error (AE)} \\
  		\midrule
{} &   5perc  &  1quart &   median &   3quart &   95perc &   99perc &  1quart &   median &   3quart &   99perc \\
\midrule
$\rho\phantom{aa}$ & 0.0035$\phantom{a}$ & 0.018 & 0.045 & 0.141 & 0.870& 4.28  & 0.016 & 0.039 & 0.105 & 0.689 \\
$\rhon$ & 0.0008 &0.001& 0.006 & 0.019 & 0.564 & 2.86 & 0.004 & 0.012 & 0.039 & 0.286 &\\
\midrule
$R$ & 0.0045 & 0.021 & 0.044 & 0.103 & 0.548& 2.41&  0.013 & 0.029 & 0.070& 0.527 & \\
$\hat{R}$ & 0.0006 & 0.003 & 0.007 & 0.020 & 0.133&0.55&   0.001 & 0.003 & 0.007 & 0.065\\
\midrule
$S$ & 0.0005 & 0.003 & 0.008 & 0.030 & 0.586 &  81.8 & 0.001 & 0.003 & 0.007 & 0.072\\
\midrule
{$\hat{R}$} imm & 0.0053&   0.0067 & 0.016 & 0.049 & 0.360 & 1.83& 0.0037 & 0.008 & 0.019 & 0.151\\

{$\hat{R}$} iso & 0.0030&0.0092 & 0.026 & 0.091 & 0.569 & 2.74& 0.0081 & 0.021 & 0.057 & 0.460\\
{$\hat{R}$} trancr & 0.0072 &  0.0229 & 0.071 & 0.240 & 1.490 & 7.55& 0.018 & 0.049 & 0.12 & 0.680\\
\bottomrule
\end{tabular}
  	\end{center}
  	\vspace*{-0.6cm}
  \end{table}  
\begin{figure}[t!]
    \centering
    \includegraphics[width=.45\linewidth]{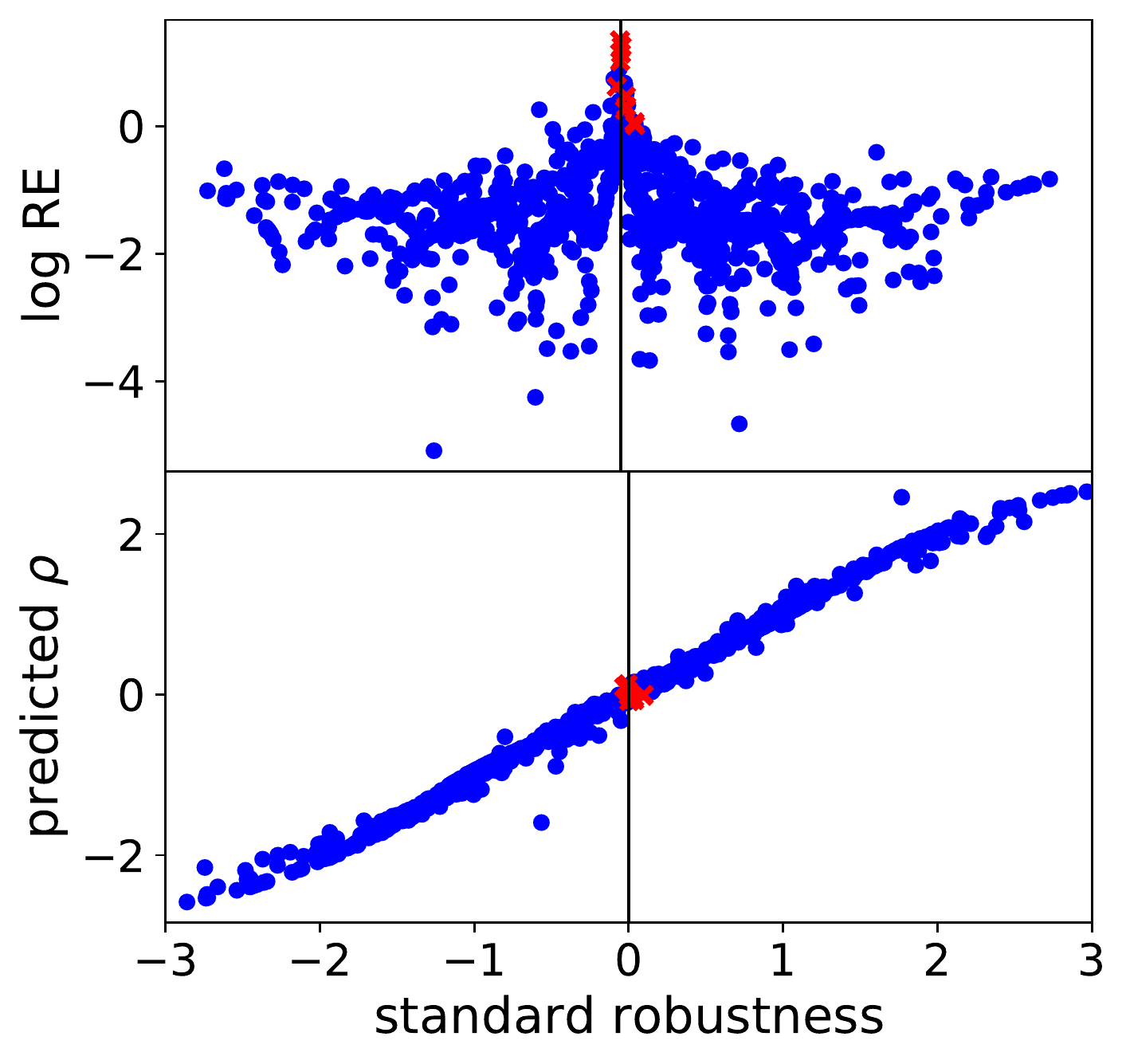}~~~
        \includegraphics[width=.46\linewidth]{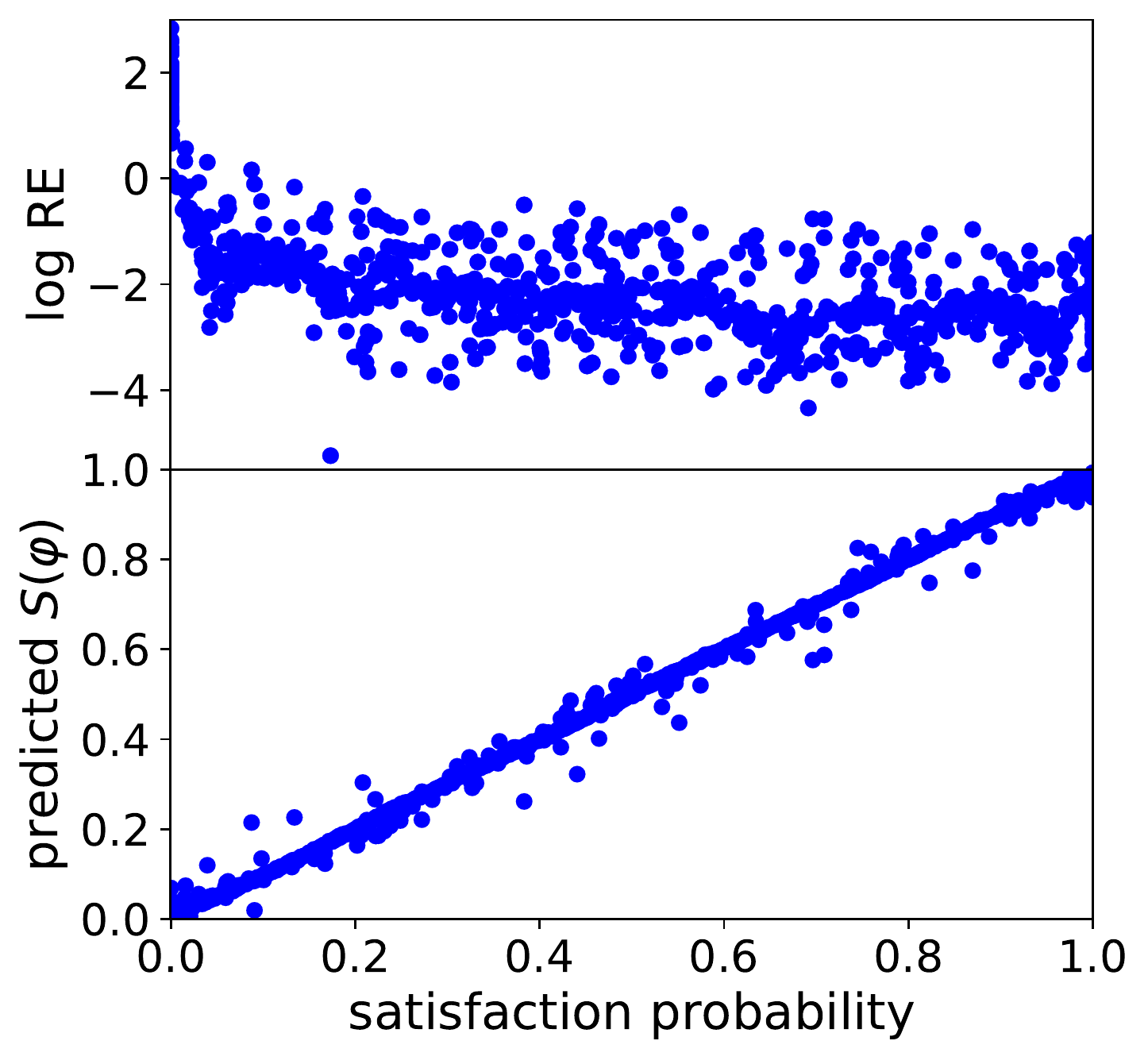}
    \vspace{-0.5em}
\caption{ (left) True  standard robustness vs predicted values and RE on  single trajectories sampled from $\mu_0$. The misclassified formulae are the red crosses. (right) Satisfaction probability vs predicted values and RE (again  for a single experiment).}
    \label{fig:single_exp_scatter}
    \vspace{-0.5cm}
\end{figure}

  In Fig.~\ref{fig:single_exp_scatter} (left), we plot the true standard robustness for random test formulae in contrast to their predicted values and the corresponding log RE. Here we can clearly observe that the misclassified formulae (red crosses) tend to have a robustness close to zero, where even tiny absolute errors unavoidably produce large relative errors and frequent misclassification.

We test our method also on three specifications of the ARCH-COMP 2020 \cite{ARCH20}, to show that it works well even on real formulae. We obtain still good results, with an accuracy equal to 1, median AE = $0.0229$, and median  RE = $0.0316$ in the worst case (Appendix~\ref{app:subsec_single}).  


\subsection{Expected Robustness and Satisfaction Probability}\label{sec:exp:sat}

In these experiments, we approximate the expected robustness and the satisfaction probability using a fixed set of 5000 trajectories sampled according to $\mu_0$, evaluating it for each formula in the training and test sets, and predicting the expected standard $R(\varphi)$ and normalized robustness $\hat{R}(\varphi)$ and the satisfaction probability $S(\phi)$ for the test set. 

For the robustness,  mean of quantiles of RE and AE show good results as can be seen in Table~\ref{tab:quantile_all}, rows 3--4. Values of MSE, MAE and MRE are smaller than those achieved on single trajectories with medians for n=3 equal to 0.0015, 0.0637, and 0.2 for the standard robustness and 0.000212, 0.00688, and 0.0478 for the normalized one. Normalized robustness continues to outperform the standard one. 

For the satisfaction probability, values of MSE and MAE errors are very low, with a median for n=3 equal to 0.000247 for MSE and 0.0759 for MAE. MRE instead is higher and equal to 3.21.  The reason can be seen in Fig.~\ref{fig:single_exp_scatter}~(right), where we plot the satisfaction probability vs the relative error for a random experiment.
We can see that all large relative errors are concentrated on formulae with satisfaction probability close to zero, for which even a small absolute deviation can cause large errors. 
Indeed the $95$th percentile of RE is still pretty low, namely $0.586$ (cf.~Table~\ref{tab:quantile_all}, row 5), while we observe the  $99$th percentile of RE blowing up to 81.8 (at points of near zero true probability). This heavy tailed behaviour suggests to rely on median for a proper descriptor of typical errors, which is 0.008 (hence the typical relative error is less than 1\%).


\subsection{Kernel Regression on  Other Stochastic Processes}\label{sec:exp:processes}
 The last aspect that we investigate is whether the definition of our kernel w.r.t. the fixed measure $\mu_0$  can be used for making predictions also for other stochastic processes, i.e. without redefining and recomputing the kernel every time that we change the distribution of interest on the trajectory space. 

  \begin{wrapfigure}{R}{0.57\textwidth}
  \vspace*{-0.9cm}
  	\includegraphics[scale=0.55]{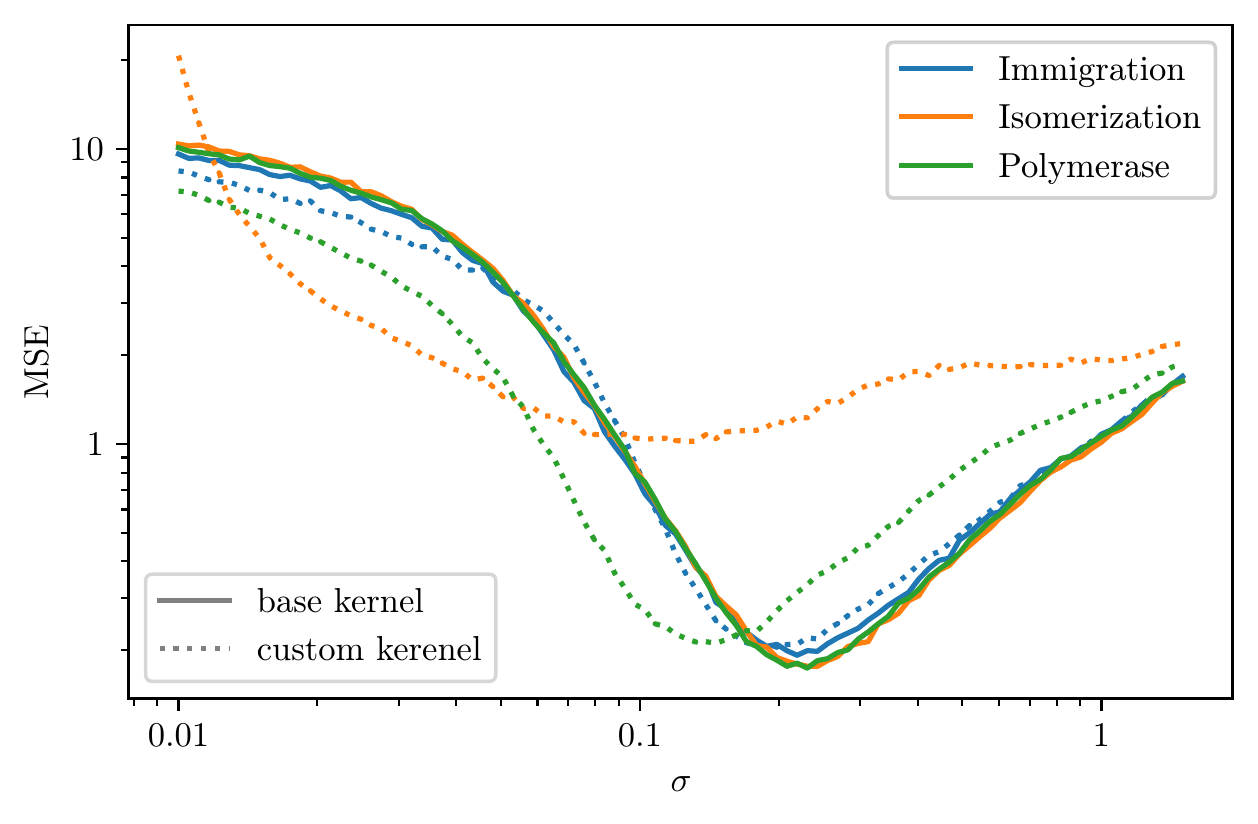}
  	  \vspace*{-0.9cm}
  \caption{\label{fig:compare_sigma_exp_rob_other_stoch}
  Expected robustness prediction using  the kernel evaluated according to the {\it base kernel}, and a {\it custom kernel}.    We depict  MSE as a function of the bandwidth $\sigma$ of the Gaussian kernel (with both axes in logarithmic scale). 
  	}
  	  \vspace*{-0.9cm}
  \end{wrapfigure}
  
\noindent{\bf Standardization.} To use the same kernel of $\mu_0$ we need to standardize the trajectories so that they have the same scale as our base measure. Standardization, by subtracting to each variable its sample mean and dividing by its sample standard deviation, will result in a similar range of values as that of trajectories sampled from $\mu_0$, thus removing distortions due to the presence of different scales and allowing us to reason on the trajectories using thresholds like those generated by the STL sampling algorithm.

\noindent{\bf Performance of base and custom kernel.} 
We 
consider three different stochastic models: \emph{Immigration} (1 dim), \emph{Isomerization} (2 dim)  and \emph{Polymerise} (2 dim),   simulated using the Python library StochPy \cite{maarleveld2013stochpy} (see also Appendix~\ref{app:exp:models} and Fig.~\ref{fig:stoch_trajectories}). 
We compare the performance using the kernel evaluated according to the base measure $\mu_0$ (base kernel), and a custom kernel 
computed replacing $\mu_0$ with the measure on trajectories given by the stochastic model itself. Results show that the base kernel is still the best performing one, see Fig.~\ref{fig:compare_sigma_exp_rob_other_stoch}.
This can be explained by the fact that the measure $\mu_0$ is broad in terms of coverage of the trajectory space, so even if two formulae are very similar, there will be, with a high probability, a set of trajectories for which the robustnesses of the two formulae are very different. This allows us to better distinguish among STL formulae, compared to models that tend to focus the probability mass on narrower regions of $\mathcal{T}$ as, for example, the \emph{Isomerization} model, which is the model with the most homogeneous trajectory space and has indeed the worst performance.

\noindent{\bf Expected Robustness} 
Setting is the same as for the corresponding experiment on $\mu_0$. Instead of the Polymerase model, we consider here a \emph{Transcription} model \cite{maarleveld2013stochpy} (Appendix~\ref{app:exp:models}), to have also a 3-dimensional model.
Results of quantile for RE and AE for the normalized robustness are reported in Table~\ref{tab:quantile_all}, bottom three rows.
The results on the different models are remarkably promising, with the Transcription model (median RE 7\%) performing a bit worse than Immigration and Isomeration  (1.6\% and 2.6\% median RE). 

Similar experiments have been done also on single trajectories, where we obtain similar results as for the Expected Robustness (Appendix~\ref{app:subsec_stoch_single}).

\vspace*{-0.5em}
\section{Conclusions}\label{sec:conc} 
\vspace*{-0.5em}

To enable any learning over formulae, their features must be defined. 
We circumvented the typically manual and dubious process by adopting a more canonic, infinite-dimensional feature space, relying on the quantitative semantics of STL.
To effectively work with such a space, we defined a kernel for STL.
To further overcome artefacts of the quantitative semantics, we proposed several normalizations of the kernel.
Interestingly, we can use \emph{exactly} the same kernel with a fixed base measure over trajectories across different stochastic models, not requiring any access to the model.
We evaluated the approach on realistic biological models from the stochpy library 
as well as on realistic formulae from Arch-Comp
and concluded a good accuracy already with a few hundred training formulae.



Yet smaller training sets are possible through a wiser choice of the training formulae: one can incrementally pick  formulae significantly different (now that we have a~similarity measure on formulae) from those already added.
Such active learning results in a better coverage of the formula space, allowing for a more parsimonious training set.
Besides estimating robustness of concrete formulae, one can lift the technique to computing STL-based distances between stochastic models, given by differences of robustness over \emph{all} formulae, similarly to \cite{jan2016linear}.
To this end, it suffices to resort to a dual kernel construction, and build non-linear embeddings of formulae into finite-dimensional real spaces using the kernel-PCA techniques \cite{murphy2012machine}. 
Our STL kernel, however, can be used for many other tasks, some of which we sketched in Introduction. 
Finally, to further improve its properties, another direction for future work is to refine the quantitative semantics so that equivalent formulae have the same robustness, e.g. using ideas like in \cite{madsen2018metrics}.


\bibliographystyle{alpha}
\bibliography{biblio}

\clearpage
\appendix

\begin{center}
\textbf{\Large Appendix}
\end{center}

\section{Sampling algorithm for the base measure $\mu_0$}\label{app:mu}

In this section, we detail a pseudocode of the  sampling algorithm over piecewise linear functions that we use for Monte Carlo approximation. In doing so, we sample from a dense subset of $C(I), I= [a,b]$. 
  The sampling algorithm is the following:
  \begin{enumerate}
  	\item Set a discretization step $\Delta$; define $N = \frac{b-a}{\Delta}$ and $t_i = a + i\Delta$;
  	\item Sample a starting point $\xi_0\sim\normal(m', \sigma')$ and set $\xi(t_0) = \xi_0$;
  	\item Sample $K\sim(\normal(m'', \sigma''))^2$, that will be the total variation of $\xi$;
  	\item Sample $N-1$ points $y_1,...,y_{N-1}\sim\uniform([0, K])$ and set $y_0=0$ and $y_n=K$;
  	\item Order $y_1, ..., y_{N-1}$ and rename them such that $y_1 \leq y_2 \leq...\leq y_{N-1}$;
  	\item Sample $s_0\sim\discrete(-1, 1)$;
  	\item Set iteratively $\xi(t_{i+1}) = \xi(t_i) + s_{i+1}(y_{i+1}-y_i)$ with $s_{i+1} = s_is$,\\
  	$P(s=-1) = q$ and $P(s=1) = 1-q$, for $i = 1, 2, ..., N$.
  \end{enumerate}

 \begin{wrapfigure}{R}{0.40\textwidth}
\vspace*{-0.5cm}
\includegraphics[scale=0.4]{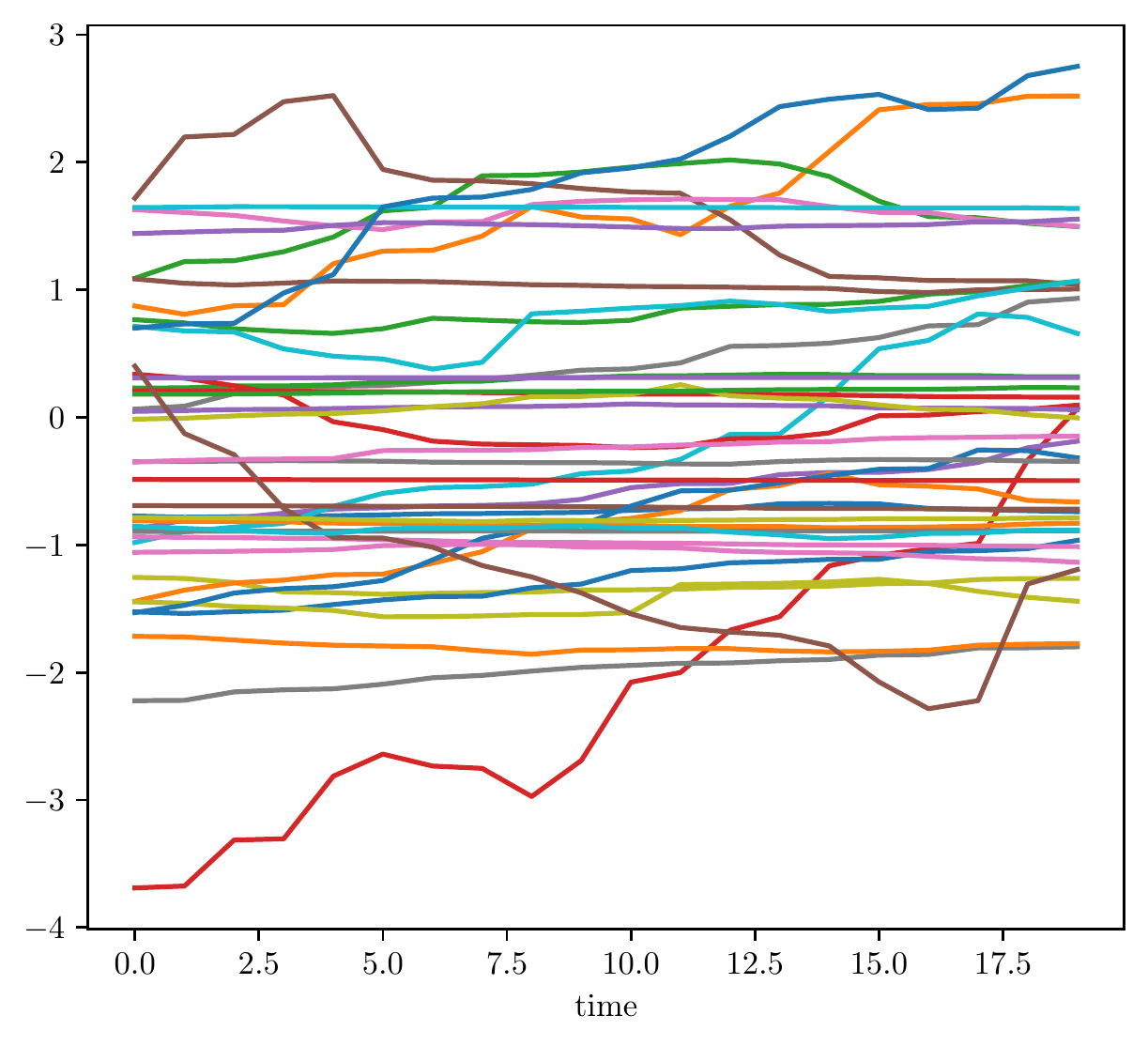}
\vspace*{-0.5cm}
  	\caption{\label{fig:trajectories_base}Trajectories randomly sampled from $\mu_0$}
\vspace*{-0.7cm}
\end{wrapfigure}

  Finally, we  linearly interpolate between consecutive points of the discretization and make the trajectory continuous, i.e., $\xi\in C(I)$.
  For our implementation, we fixed the above parameters as follows: $a = 0$, $b = 100$,  $h = 1$,
$m' = m'' = 0"$, $\sigma' = 1$,
 $\sigma'' = 1$,
 $q = 0.1$.

 Note that this algorithm, conditioned on the total variation $K$, samples trajectories with sup norm bounded by $K$. As $K$ is sampled from a squared Gaussian distribution, it is not guaranteed to be finite, however the probability of having a large sup norm  will decay exponentially with $K$, which  preserves $L_2$ integrability and the validity of our kernel definition (see the proof of theorem \ref{th:L2}).

\section{Technical details on the STL Kernel}\label{app:STLkernel}
    We consider the map $h: \mathcal{P}\to C(\mathcal{T}\times I)$ defined by $h(\varphi)(\xi, t) = \rho(\varphi, \xi, t)$, and $h_0: \mathcal{P}\to C(\mathcal{T})$ defined by $h_0(\varphi)(\xi) = \rho(\varphi, \xi, 0)$, where $C(X)$ denotes the set of the continuous functions on the topological space $X$. It can be proved that
  the functions in $h(\mathcal P)$ are square integrable in $\mathcal{T}\times I$ and those in $h_0(\mathcal P)$ are square integrable in $\mathcal{T}$,  hence we can use the scalar product in $L^2$ as a kernel for $\mathcal{P}$.   
  
  Before formally describing the theorem and is proof, let us recall the definition of $L^2$ and its inner product.
  \begin{definition}
  	Given a measure space $(\Omega, \mu)$, we call Lebesgue space $L^2$ the space defined by
  	$$L^2(\Omega) = \{f\in\Omega: \|f\|_{L^2} < \infty\}\,,$$
  	where $\|\cdot\|_{L^2} $ is a norm defined by
  	$$\|f\|_{L^2} =\left(\int_\Omega |f|^2 d\mu\right)^{\half}\,.$$
  	We define the function $\bangle{\cdot,\cdot}_{L^2}:L^2(\Omega)\times L^2(\Omega)\to\reals$ as
  	$$\bangle{f,g}_{L^2} = \int_\Omega fg\,.$$
  	It can be easily proved that $\bangle{\cdot, \cdot}_{L^2}$ is an inner product.\footnote{An inner product $\bangle{\mathbf{x_1},\mathbf{x_2}}$ maps  vectors $\mathbf{x_1}$ and $\mathbf{x_2}$ of a vector space into a real number. It is bi-linear, symmetric and positive definite (ie positive when $\mathbf{x_1}\neq\mathbf{x_2}$) and generalises the notion of scalar product to more general spaces, like Hilbert ones.}
  \end{definition}
  The following is a standard result.
  \begin{proposition}[e.g.~\cite{kothe1983topological}]
  	$L^2(\Omega)$ with the inner product $\bangle{\cdot, \cdot}_{L^2}$ is a Hilbert space.
  \end{proposition}
 
\begin{theorem}\label{thm:hilbert}
Given the STL formulae space  $\mathcal{P}$, the trajectory space $\mathcal{T}$, a bounded interval $I\subset\reals$, let  $h: \mathcal{P}\to C(\mathcal{T}\times I)$ and $h_0: \mathcal{P}\to C(\mathcal{T})$ defined as above, then:
  \begin{equation}\label{subset L}
  h(\mathcal P)\subseteq L^2(\mathcal{T}\times I)\;\;\text{and}\;\; h_0(\mathcal P)\subseteq L^2(\mathcal{T})
  \end{equation} 
  \label{th:L2}
  \vspace*{-1.5em}
\end{theorem}
\begin{proof}
  We prove the result for $h$, as the proof of $h_0$ is similar. In order to satisfy (\ref{subset L}), we  make the hypothesis that
  $\mathcal T$ is a bounded (in the $\sup$ norm) subset of $C(I)$, with $I$ a bounded interval, which means that exists $B\in\reals$ such that $\|\xi\|_{\infty} = \sup_{t\in I} \xi(t) \leq B$ 
  for all $\xi\in\mathcal T$. Moreover, the measure on $\mathcal T$ is a distribution, and so it is a finite measure. Hence
  \begin{align*}
  \int_{\xi\in\mathcal T}\int_{t\in I}|h(\varphi)(\xi, t)|^2dtd\mu = &\int_{\xi\in\mathcal T}\int_{t\in I}|\rho(\varphi, \xi, t)|^2dtd\mu \\\leq &\int_{\xi\in\mathcal T}\int_{t\in I}M(\varphi)^2dtd\mu \\\leq & M(\varphi))^2|I|
  \end{align*}
  for each $\varphi\in\mathcal P$, where $M(\varphi)<\infty $ is the maximum absolute value of all atomic propositions of $\varphi$ w.r.t. $\mathcal{T}$, which is finite due to the boundedness of $\mathcal{T}$. This implies $h(\mathcal P)\subseteq L^2(\mathcal T\times I)$.
\end{proof} 

Note that the requirement of $\mathcal{T}$ bounded is not particularly stringent in practice (we can always enforce a very large bound), and it  can be further relaxed requiring the measure $\mu$ assigns to functions with large norm an exponentially decreasing probability (w.r.t the increasing norm).

We can now use the scalar product in $L^2$ as a kernel for $\mathcal{P}$. In such a way, we will obtain a kernel that returns a high positive value for formulae that agree on high-probability trajectories and high negative values for formulae that, on average, disagree. We report here the definition in the paper, decorating it also with the newly introduced notation. 
\begin{definition*}[\bf \ref{def:STLkernel}]
    Fixing a probability measure $\mu_0$ on $\mathcal{T}$, we can then define the STL-kernel as:
  \begin{eqnarray*}
  	\kr'(\varphi, \psi) = \bangle{h(\varphi), h(\psi)}_{L^2(\mathcal T\times I)} & = \int_{\xi\in\mathcal T}\int_{t\in I}h(\varphi)(\xi, t)h(\psi)(\xi, t)dtd\mu_0 \\
  	& = \int_{\xi\in\mathcal T}\int_{t\in I}\rho(\varphi, \xi, t)\rho(\psi, \xi, t)dtd\mu_0\,
  	\end{eqnarray*}
\end{definition*}
 
In the previous definition, we can replace $h$ by $h_0$ and take the scalar product in $L^2(\mathcal T)$. We call this latter version of the kernel the \textit{untimed} kernel, and the former the \textit{timed} one. As we will see in the next section, the difference in accuracies achievable between the two kernels is marginal, but the untimed one is sensibly faster to compute.

We are now ready to prove Theorem~\ref{th:proper_kernel} of the main paper, reported below.

\begin{theorem*}[\bf \ref{th:proper_kernel}]
The function $\kr'$ is a proper kernel function.
\end{theorem*}

The key to prove the theorem is the following proposition, showing that  the function $\kr'$ satisfies the finitely positive semi-definite property being defined via a scalar product.

\begin{proposition}\label{kernel matrices characterization}
The function $\kr'$ satisfies the positive semi-definite property.
  \end{proposition}
  \begin{proof}
  	In order to prove that it has the positive semi-definite property, we need to show that for any tuple of points $\varphi_1,\ldots,\varphi_n$, the Gram matrix $K$ defined by $K_{ij} = \bangle{h(\varphi_i), h(\varphi_j)}$ for $i, j = 1, ..., n$ is positive semi-definite.
  	\begin{align*}
  	v^T Kv &= \sum_{i, j=1}^n v_i v_j K_{ij}= \sum_{i, j=1}^n v_i v_j \bangle{h(\varphi_i), h(\varphi_j)}\\
  	&=  \left\langle\sum_{i=1}^n v_ih(\varphi_i),\sum_{j=1}^n v_jh(\varphi_j)\right\rangle \\
  	&=\left\|\sum_{i=1}^n v_ih(\varphi_i)\right\|^2\geq 0\,,
  	\end{align*}
  	which implies that $K$ is positive semi-definite.
  \end{proof}

The previous proposition, in virtue of the following well known theorem characterizing kernels \cite{foundationsML2018}, guarantees that we have  defined a proper kernel. For completeness, we also report a proof of this last theorem. 

\begin{theorem}[Characterization of kernels, e.g.~\cite{foundationsML2018}]\label{kernels characterization}
  	A function $\kr:X\times X\to\reals$ which is either continuous or has a finite domain, can be written as
  	$$\kr(x,z)=\bangle{\phi(x),\phi(z)}_{\mathbb{H}}\,,$$
  	where $\phi$ is a feature map into a Hilbert space $\mathbb{H}$, if and only if it satisfies the finitely positive semi-definite property.
  \end{theorem}
  
  \begin{proof}
  	Firstly, let us observe that if $\kr(x,z)=\bangle{\phi(x),\phi(z)}$, then it satisfies the finitely positive semi-definite property for the Proposition \ref{kernel matrices characterization}. The difficult part to prove is the other implication.
  	
  	Let us suppose that $\kr$ satisfies the finitely positive semi-definite property. We will construct the Hilbert space $F_\kr$ as a function space. We recall that $F_\kr$ is a Hilbert space if it is a vector space with an inner product that induces a norm that makes the space complete.
  	
  	Let us consider the function space
  	$$\mathcal F=\left\{\sum_{i=1}^n\alpha_i\kr(x_i,\cdot)\,:\,n\in\naturals,\ \alpha_i\in\reals,\ x_i\in X,\ i=1,...,n\right\}\,.$$
  	The sum in this space is defined as
  	$$(f+g)(x)=f(x)+g(x)\,,$$
  	which is clearly a close operation. The multiplication by a scalar is a close operation too. Hence, $\mathcal F$ is a vector space.
  	
  	We define the inner product in $\mathcal F$ as follows. Let $f, g\in\mathcal F$ be defined by
  	$$f(x) = \sum_{i=1}^n\alpha_i\kr(x_i,x)\,, g(x) = \sum_{i=1}^m\beta_i\kr(z_i,x)\,,$$
  	so the inner product is defined as
  	$$\bangle{f,g}:=\sum_{i=1}^n\sum_{j=1}^m\alpha_i\beta_j\kr(x_i,z_j) = \sum_{i=1}^n\alpha_i g(x_i)=\sum_{j=1}^m\beta_i f(z_j)\,,$$
  	where the last two equations follows from the definition of $f$ and $g$.
  	This map is clearly symmetric and bilinear. So, in order to be an inner product, it suffices to prove
  	$$\bangle{f,f}\geq 0\text{ for all } f\in\mathcal F\,,$$
  	and that
  	$$\bangle{f,f} = 0 \iff f \equiv 0\,.$$
  	If we define the vector $\alpha = (\alpha_1,...,\alpha_n)$ we obtain
  	$$\bangle{f,f}=\sum_{i=1}^n\sum_{j=1}^n\alpha_i\alpha_j\kr(x_i,x_j)=\alpha^T K\alpha\geq 0\,,$$
  	where $K$ is the kernel matrix constructed over $x_1,...,x_n$ and the last equality holds because $\kr$ satisfies the finite positive semi-definite property.
  	
  	It is worth to notice that this inner product satisfies the property
  	$$\bangle{f,\kr(x,\cdot)} = \sum_{i=1}^n\alpha_i \kr(x_i,x)=f(x)\,.$$
  	This property is called \emph{reproducing property} of the kernel.
  	
  	From this property it follows also that, if $\bangle{f,f}=0$ then
  	$$f(x) = \bangle{f,\kr(x,\cdot)}\leq \|f\|\kr(x,x) = 0\,,$$
  	applying the Cauchy-Schwarz inequality and the definition of the norm. The other side of the implication, i.e. $$f\equiv 0\implies\bangle{f,f} =0\,,$$ follows directly from the definition of the inner product.
  	
  	It remains to show the completeness property. Actually, we will not show that $\mathcal F$ is complete, but we will use $\mathcal F$ to construct the space $F_\kr$ of the enunciate. Let us fix $x$ and consider a Cauchy sequence $\{f_n\}_{n=1}^\infty$. Using the reproducing property we obtain
  	$$(f_n(x)-f_m(x))^2=\bangle{f_n-f_m,\kr(x,\cdot)}^2\leq\|f_n-f_m\|^2\kr(x,x)\,.$$
  	where we applied the Cauchy-Schwarz inequality. So, for the completeness of $\reals$, $f_n(x)$ has a limit, that we call $g(x)$. Hence we define $g$ as the punctual limit of $f_n$ and we define $F_\kr$ as the space obtained by the union of $\mathcal F$ and the limit of all the Cauchy sequence in $\mathcal F$, i.e.
  	$$F_\kr = \overline{\mathcal F}\,,$$
  	which is the closure of $\mathcal F$. Moreover, the inner product in $\mathcal F$ extends naturally in an inner product in $F_\kr$ which satisfies all the desired properties.
  	
  	In order to complete the proof we have to define a map $\phi:X\to F_\kr$ such that
  	$$\bangle{\phi(x),\phi(z)}=\kr(x,z)\,.$$
  	The map $\phi$ that we are looking for is $\phi(x) = \kr(x,\cdot)$. In fact
  	$$\bangle{\phi(x),\phi(z)}=\bangle{\kr(x,\cdot),\kr(z,\cdot)}=\kr(x,z)\,.$$
  \end{proof}

\section{PAC bounds}\label{sec:app:pac}
Probably Approximate Correct (PAC) bounds provide learning guarantees giving probabilistic bounds on the error committed, i.e. providing conditions under which the error is small with high probability. These conditions typically depend on the number of samples and some index measuring the complexity of the class of functions forming the hypothesis space. The following treatment is taken from \cite{foundationsML2018}.

\subsection{Rademacher Complexity}
One way to measure the complexity of a class of functions $\mathcal{G}$ is the Rademacher Complexity, which for a fixed dataset $D=\{z_1,\ldots,z_m\}$ of size $m$ is defined as 
\begin{equation}
    \label{eqn:rademacher}
    R_D(\mathcal{G}) = \mathbb{E}_{\boldsymbol{\sigma}}\left[\sup_{g\in\mathcal{G}} \frac{1}{m}\sum_{i=1}^m \sigma_i g(z_i)\right].    
\end{equation}
Note that in our scenario each $z_i$ is a pair consisting of a STL formula $\varphi$ and the quantity $y_i$ which we want to predict (typically robustness, expected robustness, satisfaction probability or Boolean satisfaction), while the function $g$ encodes the loss  associated with a predictor $h\in\mathcal{H}$, where $\mathcal{H}$ is the hypothesis space of possible models for the function  mapping each $\phi$ to its associated output $y$. Typically, the loss  is either the square loss (for regression), hence $g((\varphi,y)) = (y-h(\varphi))^2$ or the 0-1 loss (for classification), hence $g((\varphi,y)) = \mathbb{1}(y \neq h(\varphi))$. Additionally, $\boldsymbol{\sigma}$ in the equation above is distributed according to the Rademacher distribution, giving probability 0.5 to each element in $\{-1,1\}$.

Rademacher complexity can also be defined w.r.t. a data generating distribution $p_{data}$, which turns each dataset $D$ into a random variable, hence for a fixed sample size $m$:
\[ R_m(\mathcal{G}) = \mathbb{E}_{D \sim p_{data} }[ R_D(\mathcal{H})]. \]

It is possible to prove that,  with probability $1-\delta/2$ with respect to $p_{data}$, $R_m(\mathcal{G}) \leq R_D(\mathcal{G}) + \sqrt{\frac{\log \frac{2}{\delta}}{2m}}$, where $D \sim p_{data}$. 

In most of the scenarios, rather than discussing the Rademacher complexity of the loss functions $\mathcal{}{G}$, it is preferrable to express the complexity of the hypothesis space $\mathcal{H}$. In the binary classification case, for the 0-1 loss, the corresponding Rademacher complexities are essentially equivalent modulo a constant: \[ R_D(\mathcal{G}) = \frac{1}{2}R_D(\mathcal{H}). \] 

In case of kernel based methods, given a positive definite kernel $k$, we can consider the corresponding Reproducing Kernel Hilbert Space  (RKHS) $\mathbb{H}$ of functions, which however can have infinite dimension. To obtain a finite bound for Rademacher complexity, we need to restrict $\mathbb{H}$ by picking functions with a bounded norm $\Lambda$: 
\[ \mathcal{H}_{\Lambda} = \{ h\in \mathbb{H}~|~\|h\|_{\mathbb{H}}\leq \Lambda \},\] 
where $\|h\|_{\mathbb{H}} = \sqrt{<h,h>_{\mathbb{H}}}$ is the norm in the Hilbert space defined by the scalar product associated with $k$. 

In such a scenario, one can prove a bound on the Rademacher complexity of such space $\mathcal{H}_{\Lambda}$ as:
\begin{equation}
    R_D(\mathcal{H}_{\Lambda})\leq \sqrt{\frac{\Lambda^2r^2}{m}}
\end{equation}
where $D$ is a sample of size $m$ and $r^2$ is a uniform upper bound on the kernel evaluated in each point $\varphi$ of the formula space, i.e. $k(\varphi,\varphi)\leq r^2$ for each $\varphi$. Note that in our case, for the normalized and exponential kernel, we have $r=1$.

In order to effectively enforce the constraint defining $\mathcal{H}_{\Lambda}$, we need to evaluate the scalar product in $\mathbb{H}$. Following \cite{foundationsML2018}, the computation is easy if we restrict to the pre-Hilbert space $\mathbb{H}_0$, which is dense in $\mathbb{H}$ and consists of functions of the form 
\[ h(\varphi) = \sum_{i=1}^n \alpha_i k(\varphi,\varphi_i) \]
for $\alpha_i in \mathbb{R}$. For two such functions $h$ and $h'$, it holds that
\[ <h,h'>_{\mathbb{H}} = \sum_{i,j} \alpha_i\alpha'_j k(\varphi_i,\varphi'_j), \] which for a single function $h$ can be rewritten as the  quadratic form
\[ <h,h>_{\mathbb{H}} = \boldsymbol{\alpha}^T K \boldsymbol{\alpha}, \]
where $K$ is the Gram matrix for input points $\varphi_1,\ldots,\varphi_n$ and $\boldsymbol{\alpha}$ is the vector of all coefficients $\alpha_i$.

\subsection{PAC bounds with Rademacher Complexity}

PAC bounds for the zero-one loss classification problem are stated in terms of the risk $L$ and the empirical risk $\hat{L}$, namely the expected value of the loss for a given hypothesis averaged over the data generating distribution $p_{data}$ (the true risk) or the empirical distribution induced by a finite sample $D$ of size $m$ (the empirical risk). Formally,
\[L(h) = \mathbb{E}_{\varphi \sim p_{data}}\left[ h(\varphi) \neq y(\varphi) \right] \;\;\;\text{and}\;\;\; \hat{L}_D(h) = \frac{1}{m}\sum_{i=1}^m (h(\varphi_i) \neq y(\varphi_i)), \]
where $y(\varphi)$ is the true value of $\varphi$. 

The PAC bound for hypothesis space $\mathcal{H}$ states that, for any $\delta>0$, with probability at least $1-\delta$ over a sample $D$ of size $m$ drawn according to $p_{data}$,  for any $h\in \mathcal{H}$:
\begin{equation}
\label{PACboundClassification}
 L(h) \leq \hat{L}_D(h) + R_D(\mathcal{H}) + 3\sqrt{\frac{\log\frac{2}{\delta}}{2m}}.   
\end{equation}

In case of a regression task, if we consider the square loss and the corresponding risk $L^2$ and empirical risk $\hat{L}^2_D$, defined by 
\[L^2(h) = \mathbb{E}_{\varphi \sim p_{data}}\left[ (h(\varphi) \- y(\varphi))^2 \right] \;\;\;\text{and}\;\;\; \hat{L}^2_D(h) = \frac{1}{m}\sum_{i=1}^m (h(\varphi_i) - y(\varphi_i))^2,\]
then we have the following pac bound, holding for any $h\in\mathcal{H}$ and $\$\delta>0$, with probability at least $1-\delta$ on $D\sim p_{data}$:
\begin{equation}
\label{PACboundRegression}
L^2(h) \leq \hat{L}^2_D(h) + 4 M R_D(\mathcal{H}) + 3 M \sqrt{\frac{\log\frac{2}{\delta}}{2m}},  
\end{equation}
where $M$ is a an upper bound independent of $h$ on the difference $|h(\varphi) - y|$. 

Note that both bounds essentially give an upper bound in probability on the generalization error in terms of the training set error and the complexity of the hypothesis class - assuming the learning framework is based on the minimization of the empirical risk.

\subsection{PAC bounds for the STL kernel}

Combining the bounds on the Rademacher complexity for kernels and the PAC bounds for regression and classification, we can easily prove the following

\begin{theorem*}[{\bf \ref{th:PAC}}]
Let $k$ be a kernel (e.g.  normalized, exponential) for STL formulae  $\mathcal{P}$, and let $\mathbb{H}$ be the associated Reproducing Kernel Hilbert space on $\mathcal{P}$ defined by $k$. Fix $\Lambda > 0$ and consider the hypothesis space $\mathcal{H}_{\Lambda} = \{ f\in \mathbb{H}~|~\|f\|_{\mathbb{H}}\leq \Lambda \}$. Let $y:\mathcal{P}\rightarrow \mathbb{R}$ be a target function to learn as a regression task, and assume that there is $M>0$ such that for any $h$, $|h(\varphi)-y(\varphi)|\leq M$. Then for any $\delta >0$ and $h\in \mathcal{H}_{\Lambda}$, with probability at least $1-\delta$ it holds that

\begin{equation}
\label{PACboundRegressionSTL}
    L^2(h) \leq \hat{L}^2_D(h) + 4 M \frac{\Lambda}{\sqrt{m}} + 3 M \sqrt{\frac{\log\frac{2}{\delta}}{2m}}
\end{equation} 

In case of a classification problem $y:\mathcal{P}\rightarrow \{-1,1\}$, the bound becomes: 
\begin{equation}
\label{PACboundClassificationSTL}
 L(h) \leq \hat{L}_D(h) + \frac{\Lambda}{\sqrt{m}} + 3\sqrt{\frac{\log\frac{2}{\delta}}{2m}}.   
\end{equation}
\end{theorem*}

The previous theorem gives us a way to control the learning error, provided we can restrict the full hypothesis space. 
This itself requires to bound the norm in the Hilbert space generated by the kernel $k$. In case of Kernel (ridge) regression, this requires us to minimize w.r.t. coefficients $\alpha$ the  quadratic objective function
\[ (y - \alpha K)(y - \alpha K)^T \]
subject to the quadratic constraint $\alpha K \alpha^T < \Lambda$, which can be solved by introducing a Lagrange multiplier for the constraint and resorting to KKT conditions \cite{murphy2012machine}. As the so obtained objective function will be quadratic, the problem remains convex and admits a unique solution. 

Note that, practically, one typically uses a soft penalty term on the 2-norm of $\alpha$, thus obtaining ridge regression.  This penalty can be added to the objective above, and if the solution of the unconstrained problem for a given dataset has norm smaller than $\Lambda$, then this is also the solution of the constrained problem, due to its convexity. Hence, a practical approach to evaluate the PAC bound is to solve the unconstrained regularized problem, then compute the norm in the Hilbert space of the so obtained solution, and use any $\Lambda$ greater than this norm in the bound. 

Note also that using the bound for regression may not be trivial, given that it depends on a constant $M$ bounding both robustness and the functions in the hypothesis space. While imposing bounds on robustness may not be problematic, finding upper bounds on the values of $h\in\mathcal{H}_{\Lambda}$ is far less trivial.  On the other hand, the bound for classification is more easily computable. In such a case, we run some experiments to estimate the constant $\Lambda$, and obtained a median value of roughly 40, with a range of values from 10 to 1000, and first and second quartiles equal to 25 and 65. With these values, taking the median value as reference, and fixing our confidence at 95\%, the bound predicts at least 650k samples to obtain an accuracy bounded by the accuracy on the training set plus 0.05, which is much larger than training set sizes for which we observe good performances in practice.

\section{Experiments}\label{ap:exp}
As notation we use: MSE= Mean Square Error, MAE= Mean Absolute Error, MRE= Mean Relative Error, AE = Absolute Error, RE= Relative Error.

\subsection{Setting}\label{app:setting}

\parag{Syntax-tree random growing scheme} is designed as follow:
  \begin{enumerate}
    \item We start from root, forced to be an operator node. For each node, with probability $p_{\mathit{leaf}}$ we make it an atomic predicate, otherwise it will be an internal node. 
    \item In each internal (operator) node, we sample its type using a uniform distribution, then recursively sample its child or children. 
    \item We consider atomic predicates of the form $x_i \leq \theta$ or $x_i \geq \theta$. We sample randomly the variable index (dimension of the signals is a fixed parameter), the type of inequality, and sample $\theta$ from a standard Gaussian distribution $\mathcal{N}(0,1)$.
    \item For temporal operators, we sample the right bound of the temporal interval uniformly from $\{1,2,\ldots,t_{\mathit{max}}\}$, and fix the left bound to zero. 
  \end{enumerate}
In the experiments we run in the paper, we fix $p_{\mathit{leaf}} = 0.5$ 
and $t_{\mathit{max}} = 10$, see also the paragraph on hyperparameters below. 

\subsection*{Comparison of different regressor models}
We compare the performance of the regression models: {\it Nadaraya-Watson estimator}, {\it K-Nearest Neighbors regression}, {\it Support Vector Regression} (SVR) and {\it Kernel Ridge Regression} (KRR) \cite{murphy2012machine}
We compare the Mean Square Error (MSE) as a function of the bandwidth $\sigma$  of the Gaussian kernel, for the prediction of the expected robustness and the satisfaction probability with respect the base measure $\mu_0$, using different regressors. The errors are computed training the regressors on $300$ different train sets made up of 400 samples and averaging the error over 300 test sets (one different per train set) made of 100 samples.
 \begin{table}[h!]
  	\begin{center}
  	\caption{MSE for the expected robustness and satisfaction probability using different regressors on different kernels.}
  		\label{tab:regressors comparison}
  		\hspace*{0cm}\begin{tabular}{r|ccccc|ccccc}
  		\hline
  		{} & \multicolumn{5}{c}{expected robustness} & \multicolumn{5}{c}{satisfaction probability} \\
  		{}	& Original &	$\sigma=0.05$ &$\sigma=0.1$ & $\sigma=0.5$ & $\sigma=1$  &  Original& $\sigma=0.05$ &$\sigma=0.1$ & $\sigma=0.5$ & $\sigma=1$ \\
  			\hline
  			NW   & -&0.31 & 0.52 & 1.8  & 3.1 & -&0.0058 &0.00088 & 0.0029 & 0.029   \\
  			KNN  & 0.34 &0.32 & 0.34 & 0.34 & 0.31  & 0.0018&0.0018 & 0.0018  & 0.0016 & 0.001\\
  			SVR  &1.5& 3.1  & 0.51 & 0.29 & 0.77  &0.0018& 0.067  & 0.033   & 0.0051 & 0.0044 \\
  			KRR  &1.4& 2.6  & 0.6  & 0.32 & 0.69 & 0.25&0.16   & 0.08    & 0.0023 & 0.00047\\
  		\end{tabular}
  	\end{center}
  \end{table}
  
  \begin{figure}[t]
    \centering
    \includegraphics[width=0.53\linewidth]{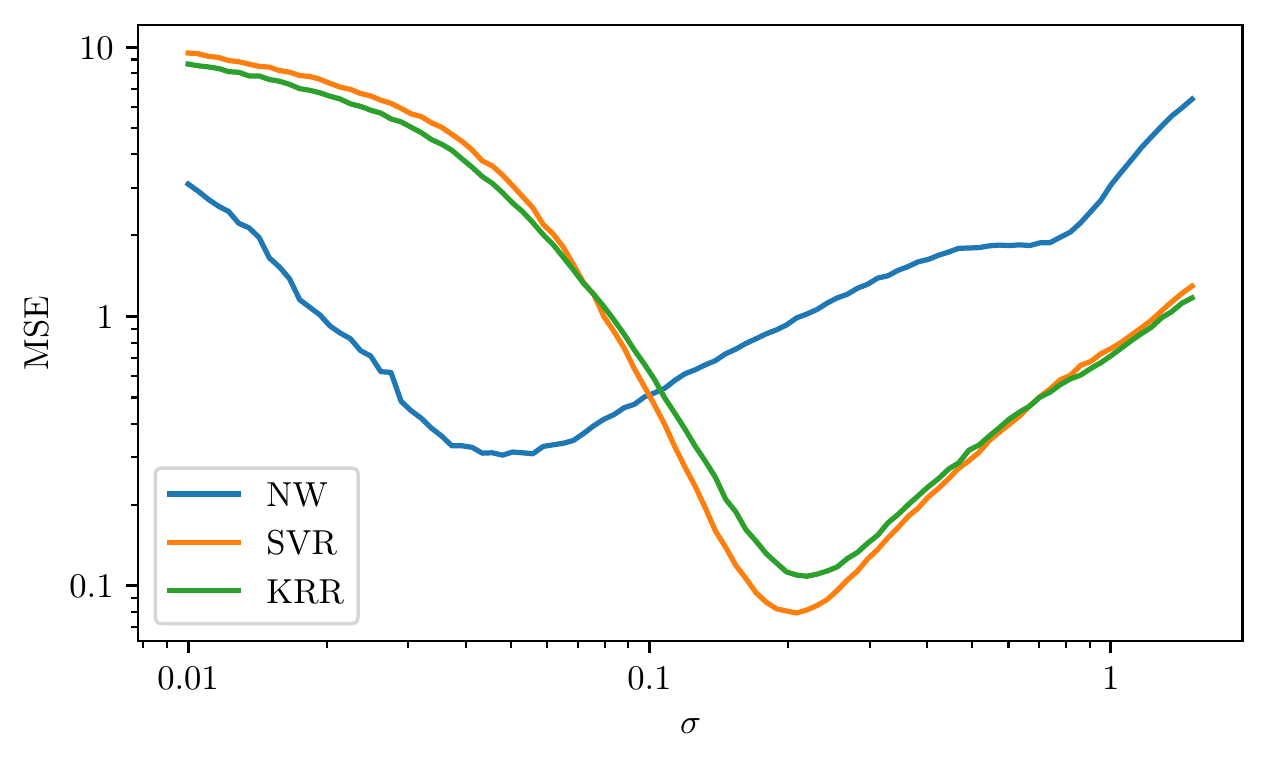}
    \includegraphics[width=0.46\linewidth]{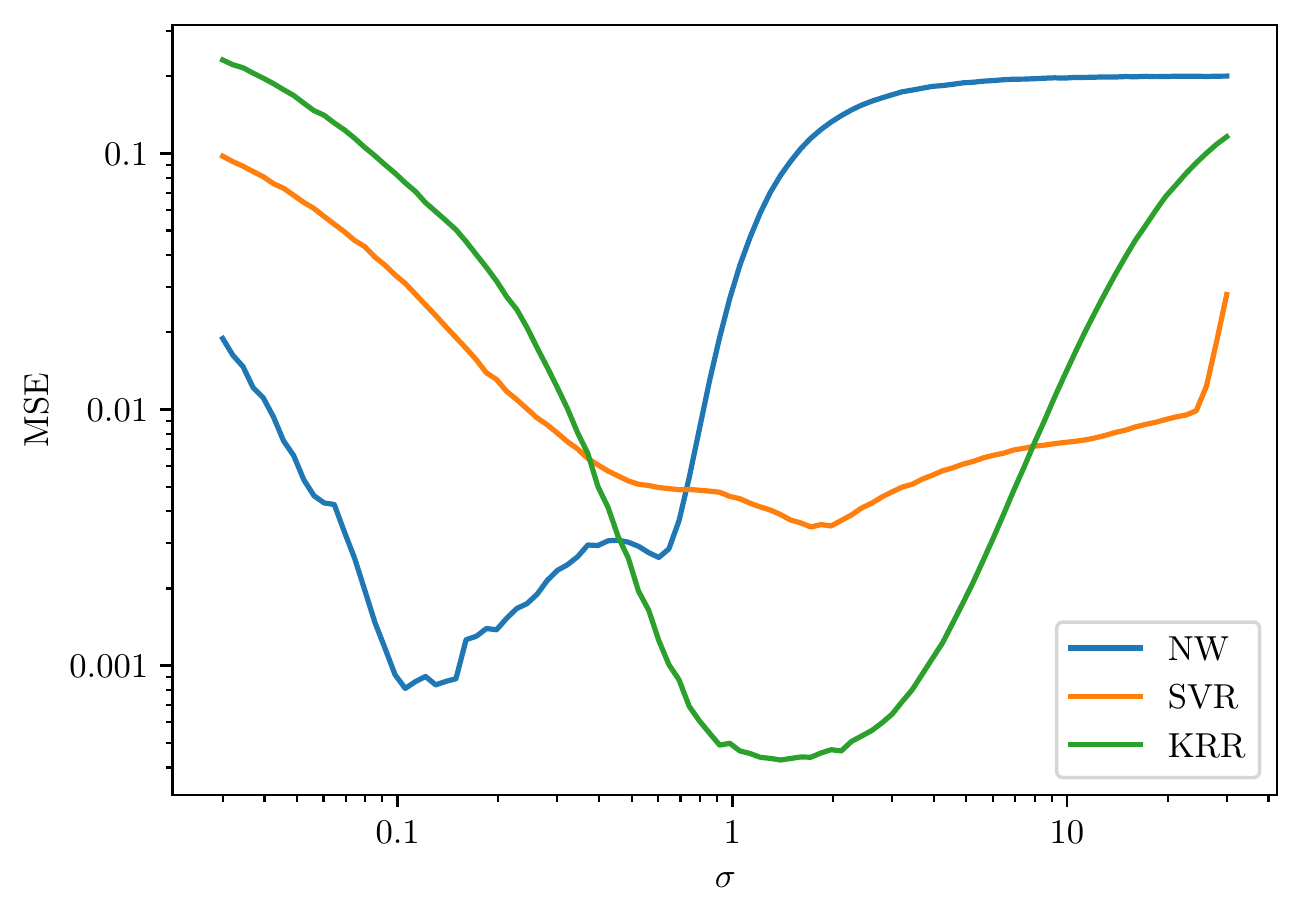}
    \caption{
    MSE as a function of the bandwidth $\sigma$ of the Gaussian kernel, for the prediction of the expected robustness (left) and the satisfaction probability (right), using different regressors. Both the axis are in logarithmic scale.}
    \label{fig:compare_sigma}
\end{figure}
  
  From Table \ref{tab:regressors comparison}, we can see that the best performance for the prediction of the expected robustness is achieved by the SVR, using the Gaussian kernel with $\sigma=0.5$. A more precise estimation of the best $\sigma$ for the Gaussian kernel is given by Figure \ref{fig:compare_sigma}~(left) . That plots confirms that SVR is the better performing regressor and the minimum regression error for the expected robustness is given by the kernel with $\sigma=0.2$. 

\subsection*{Hyperparameters}

We vary some hyperparameters of the model, testing how they impact on errors and accuracy. We test performance of KRR using the exponential kernel (setting its scale parameter $\sigma$ by cross validation when not differently specified) on the expected robustness of a formula w.r.t. the base measure $\mu_0$.

\parag{Time bounds or unbound timed operators.} 
We compare the performance on predicting the expected robustness considering formula with bound or unbound timed operators, i.e temporal operators with time intervals of the form $[0,T]$ for $T<\infty$ or $[0,\infty)$.
Results are displayed  in Fig.~\ref{fig:error_time_unbound}, and they show that the addition of time bounds has no significant impact on the performances in terms of errors. Computational times are comparable with the time bounded version slightly faster.
The mean over 100 experiments of the computational time to train and test $1000$ formulas are $5.30\pm0.023$
and $3.546\pm0.020$ seconds for the bound version and $6.61\pm0.032$ and  $4.612\pm0.026$ seconds for the unbounded version.
 \begin{figure}
    \centering
   \includegraphics[width=1.\linewidth]{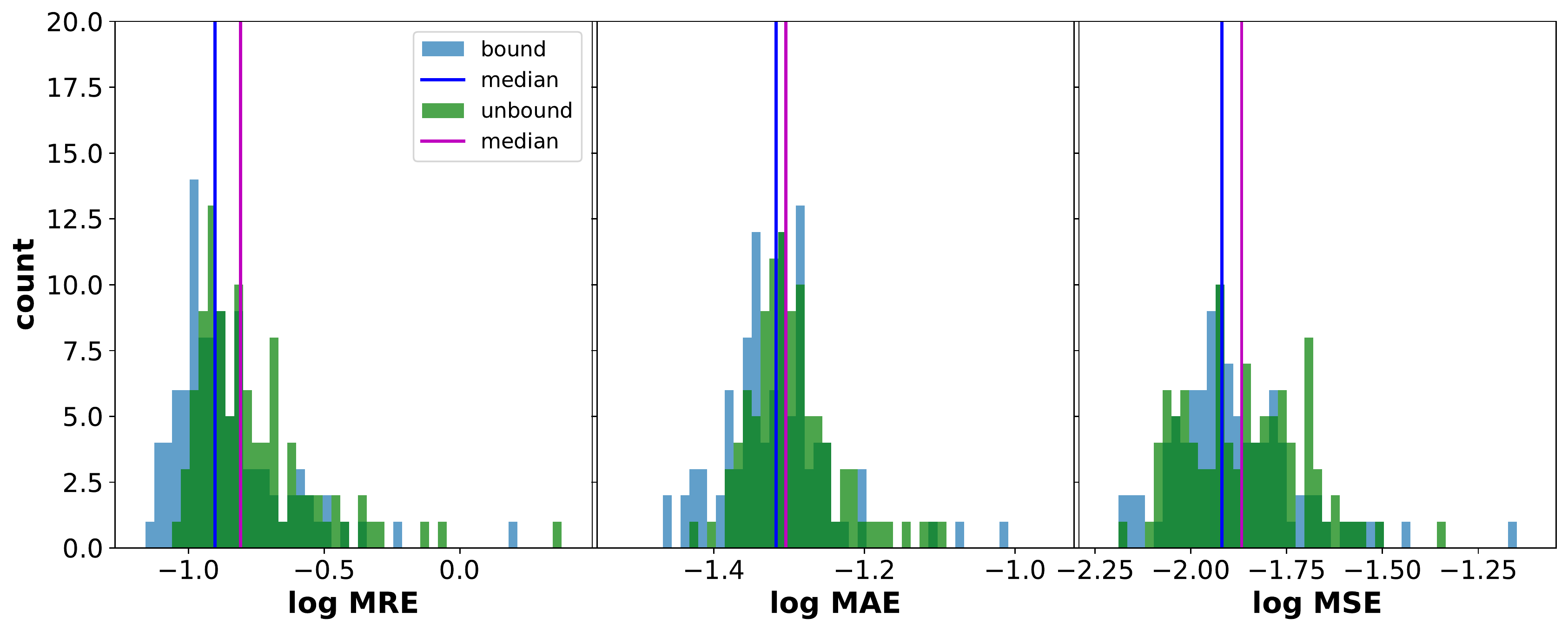}
    \caption{log MRE, MAE and MSE considering formulas with time bound or unbound operators.}
    \label{fig:error_time_unbound}
\end{figure}

\parag{Time Integration} Integrating signals $R(\varphi,t)$ w.r.t. time vs using only robustness at time zero $R(\varphi,0)$ for the definition of the kernel. In Fig.~\ref{fig:error_integration} we plot MRE, MAE and MSE for 100 experiments, we can see that using the integration gives only a very small improvement in performance (<10\%). Instead computational time for $R(\varphi,t)$ are much higher. The mean over 100 experiments of the computational time  to train and test $1000$ formulas are $5.30\pm0.023$
and $3.562\pm0.020$ seconds for $R(\varphi,0)$, and $26.6\pm0.075$ and $20.62\pm0.065$ seconds for $R(\varphi,t)$.
 \begin{figure}
    \centering
   \includegraphics[width=1.\linewidth]{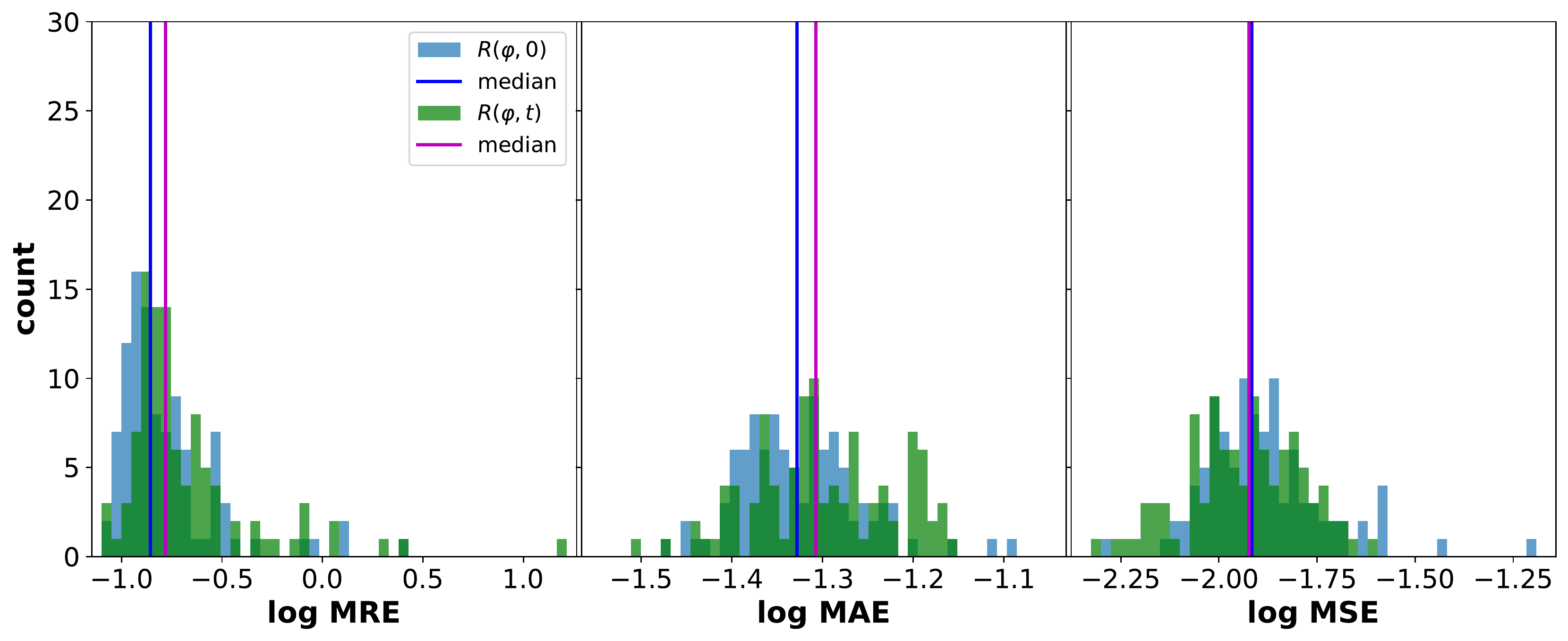}
    \caption{log MRE, MAE and MSE for using only robustness at time zero or integrating the whole signal.}
    \label{fig:error_integration}
\end{figure}

\parag{Size of the Training Set}
We analyse the performance of our kernel for different size of the training set. Results are reported in Fig.~\ref{fig:error_vs_train_size} (left). For size = [20, 200, 400, 1000], we have MRE = [0.739, 0.311, 0.247, 0.196], and MAE=[0.202, 0.087,  0.069, 0.050].
 \begin{figure}
    \centering
   \includegraphics[width=.43\linewidth]{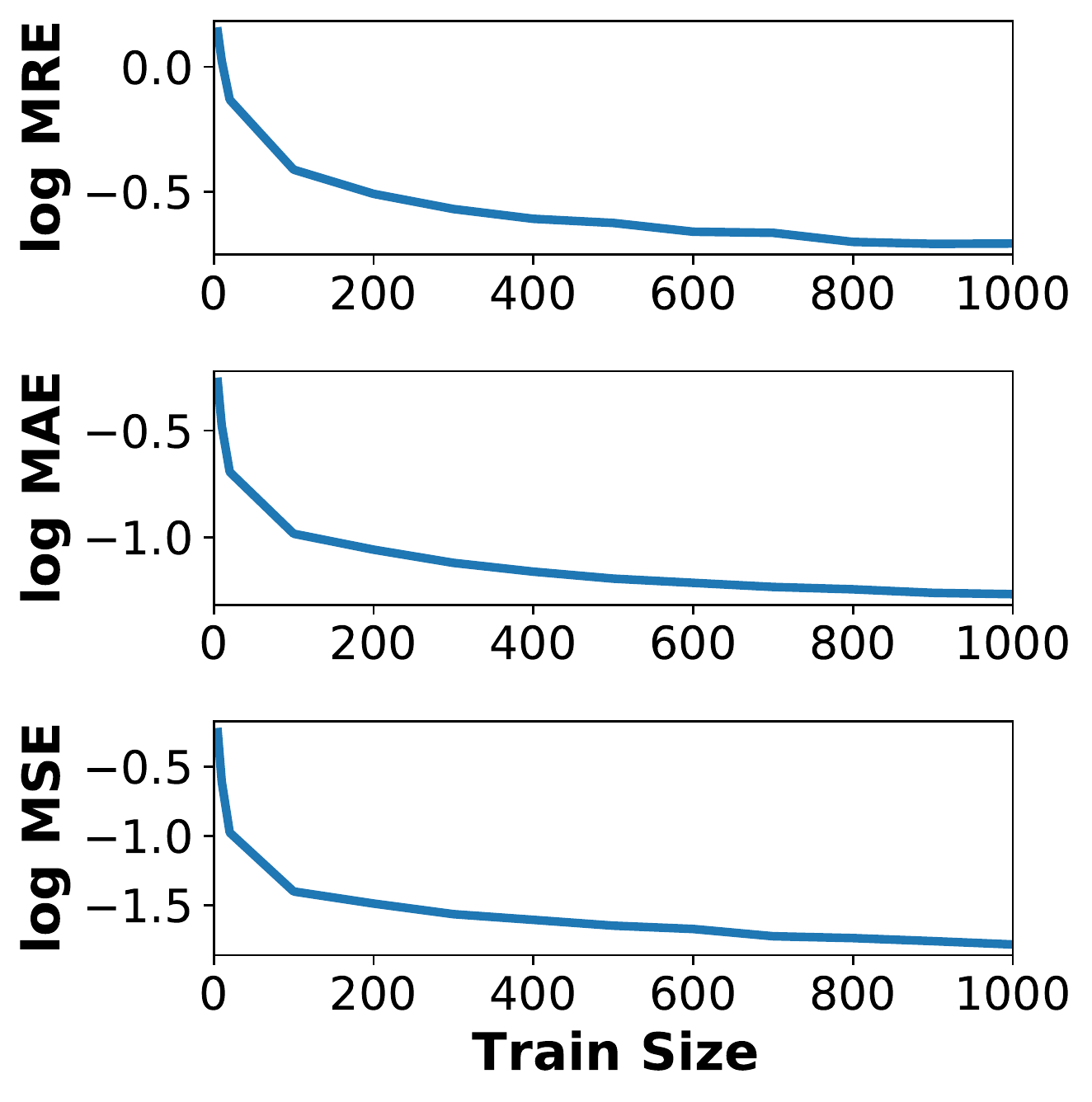}
      \includegraphics[width=.5\linewidth]{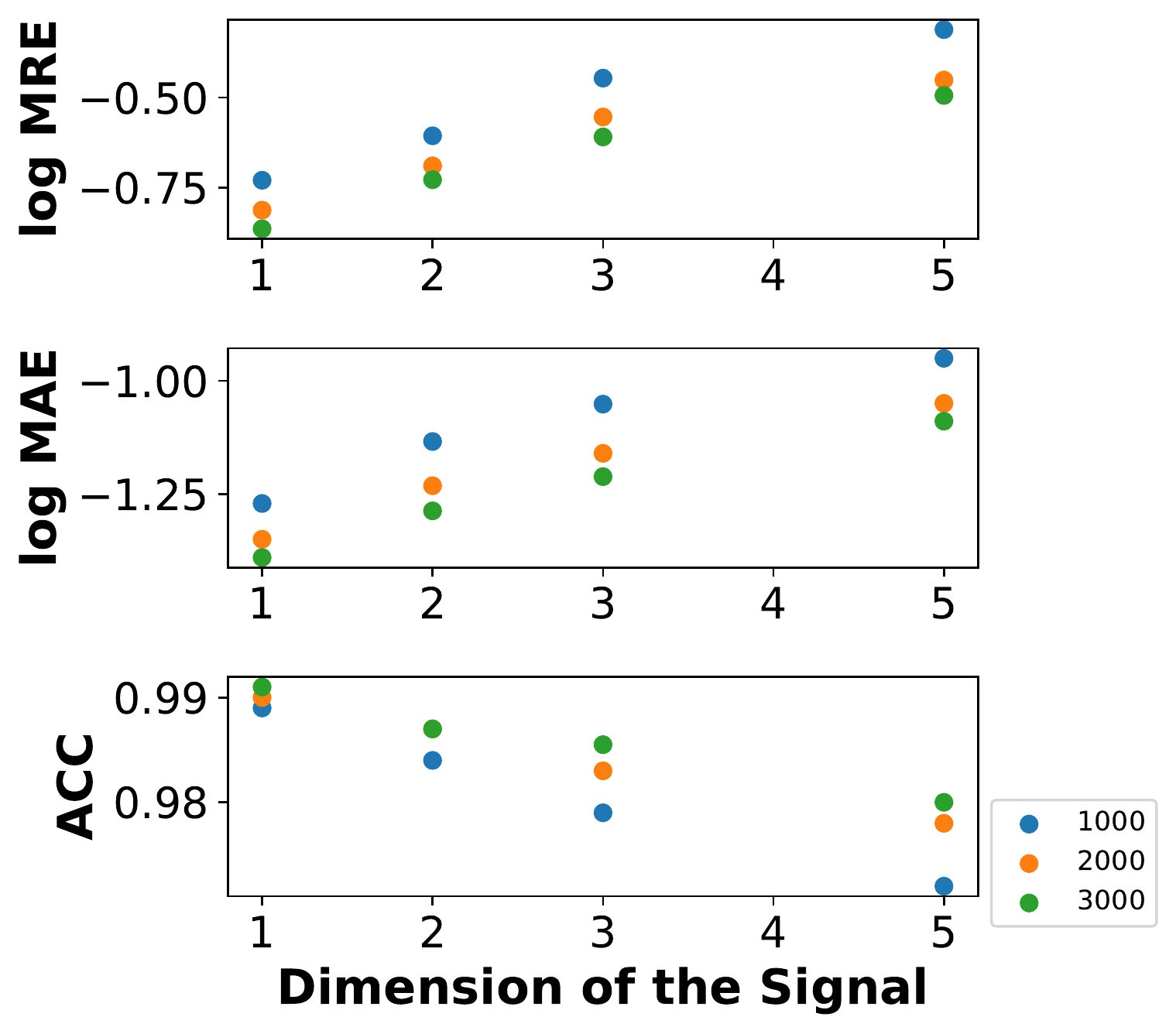}
    \caption{Median of log MRE, MAE and ACC for 100 experiments varying  (left) the size of the train set with trajectory with 1 dimension (right) the dimensionality of the signal for size of the train set n = 1000, 2000, 3000. }
    \label{fig:error_vs_train_size}
\end{figure}

\parag{Dimensionality of Signals}
We explore the error with respect different dimensionality of the signal from 1 to 5 dimension, considering train set with 1000, 2000, 3000 formulae. Results are shown in Fig.~\ref{fig:error_vs_train_size} (right). Error tends to have a linear increase, with median accuracy still over 97 for signal with dimension equal to 5.

\parag{Size of formulas}
We vary the parameter $p_{\mathit{leaf}}$ in the formula generating algorithm in the range $[0.2,0.3,0.4,0.5]$ (average formula size around $[100,25,10,6]$ nodes in the syntax tree). We observe only  a slight increase in the median relative error, see Table \ref{tab:quantile_varying_p} and Fig. \ref{fig:error_vs_train_size_P}. Also, median accuracy in predicting satisfiability with normalized robustness for a training set of size 1000 ranges in $[0.9902,0.9921,0.9931,0.9946]$ for $p_{\mathit{leaf}}$  in $[0.2,0.3,0.4,0.5]$.

 \begin{figure}
    \centering
      \includegraphics[width=.48\linewidth]{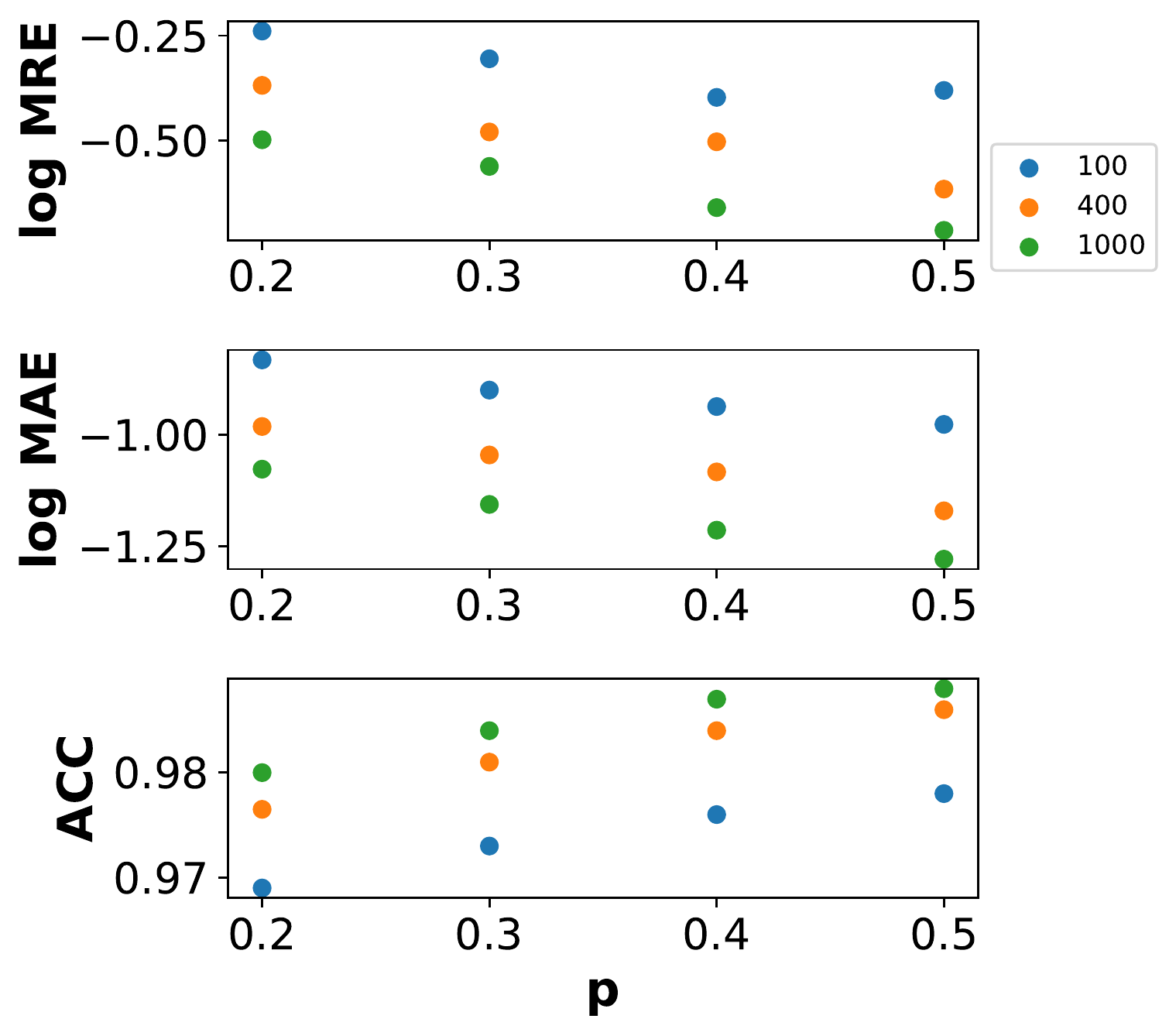}
    \includegraphics[width=.48\linewidth]{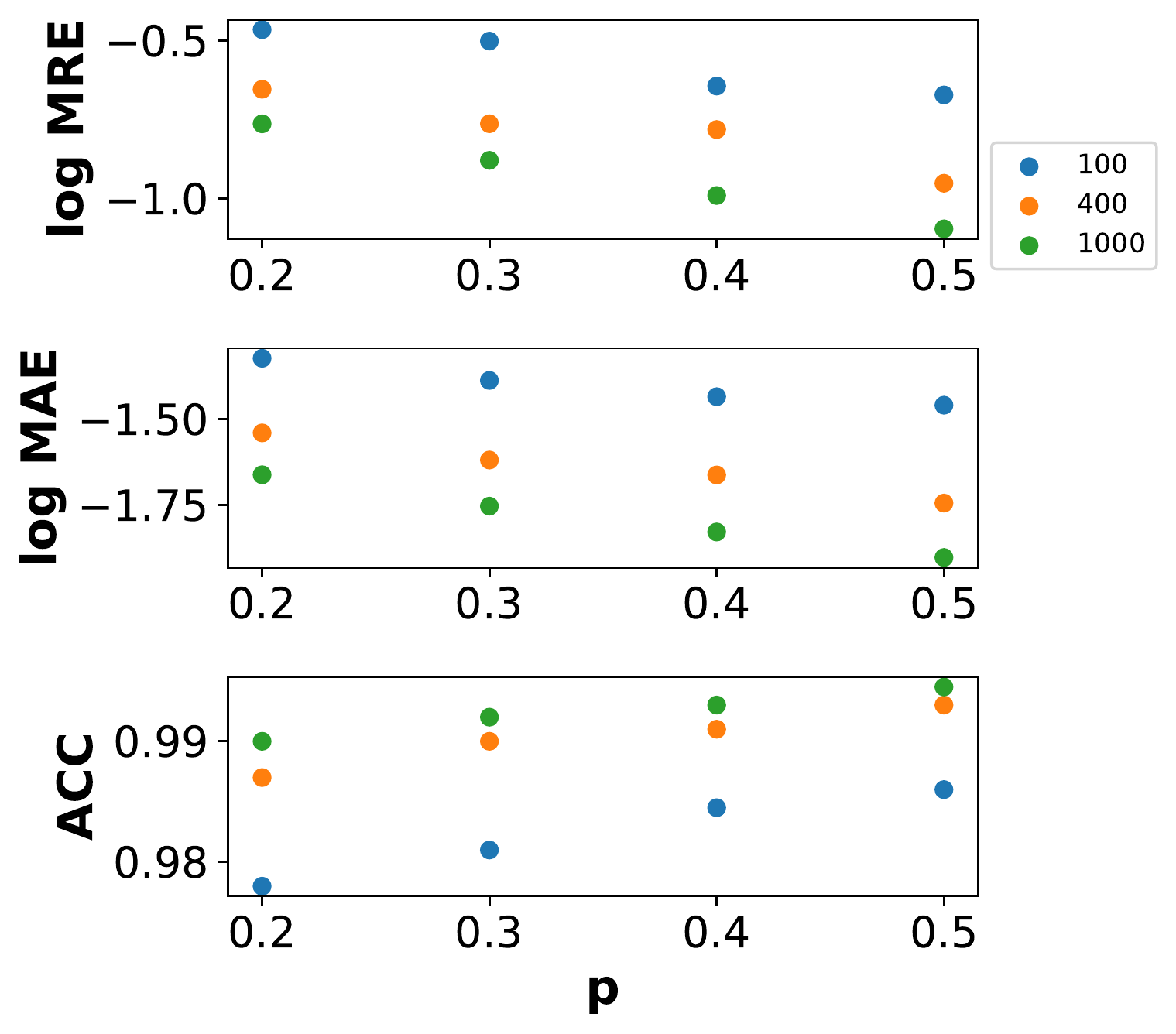} 
    \caption{Median of log MRE, MAE and ACC for 100 experiments varying the probability p for size of the train set n = 100, 400, 1000, for standard (left) and normalized (right) robustness. }
    \label{fig:error_vs_train_size_P}
\end{figure}

   \begin{table}[h!]
  	\begin{center}
  		\hspace*{0cm}
\begin{tabular}{llrrrr|rrrrr}
\toprule
  		{}& & \multicolumn{5}{c}{relative error (RE) } & \multicolumn{4}{c}{absolute error (AE)} \\
  		\midrule
{} &&    1quart &   median &   3quart &      99perc &  1quart &   median &   3quart &   99perc \\
\midrule
{$n=100$}&p=0.2  & 0.0361 & 0.0959 & 0.263 & 6.91&  0.0295 & 0.0802 & 0.200 & 0.986 \\
        &p=0.3 & 0.0303 & 0.0775 & 0.214 & 6.53& 0.0241 & 0.0640& 0.166 & 0.984 \\
        &p=0.4  & 0.0247& 0.0608 & 0.175 &  5.05& 0.0204 & 0.0523 & 0.145 & 0.974 \\
         &p=0.5  & 0.0226 & 0.0542 & 0.153 &4.92& 0.0193 & 0.047 & 0.127 & 0.967 \\
\midrule
{$n=400$}&p=0.2 & 0.0238 & 0.0602 & 0.174 &4.87 & 0.0201 & 0.0525 & 0.130& 0.873 \\
        &p=0.3 & 0.0185 & 0.0442 & 0.132 &4.26  & 0.0155 & 0.038 & 0.101 & 0.841 \\
        &p=0.4& 0.0165 & 0.0377& 0.109 & 3.58& 0.0138 & 0.0339 & 0.0872 & 0.823 \\
         &p=0.5  & 0.0164 & 0.0349 & 0.0902 &  3.04& 0.0123 & 0.0284 & 0.0694 & 0.701 \\
         \midrule
{$n=1000$}&p=0.2& 0.0181 & 0.0444 & 0.127 & 4.030 & 0.01620 & 0.0400& 0.0957 & 0.739 \\
        &p=0.3& 0.0163 & 0.0368 & 0.101 &  3.38& 0.0138& 0.0326 & 0.0775 & 0.683 \\
        &p=0.4 & 0.0147 & 0.0313 & 0.0822 & 2.863 & 0.0122 & 0.0271 & 0.0642 & 0.623 \\
         &p=0.5  & 0.0131 & 0.0276 & 0.0703& 2.45& 0.0100 & 0.0222 & 0.0518 & 0.560 \\
\bottomrule
\end{tabular}
\vspace{0.3cm}
  		\caption{Mean of quantiles for RE and AE over 100 experiments for prediction of  robustness}
  		\label{tab:quantile_varying_p}

  	\end{center}
  \end{table}


\newpage
 \subsection{Satisfiability and Robustness on Single Trajectories} \label{app:subsec_single}
Further results on experiment for prediction of Boolean satisfiability of a formula using as a discriminator the sign of the robustness and base measure $\mu_0$. Note that, for how we design our method and experiments, we never predict a robustness exactly equal to zero, so it could never happen that we classify as true a formula for which the robustness is zero but the trajectory does not satisfy the formula. We plot in Fig. \ref{fig:mae_single}, \ref{fig:mse_single}, the distribution of the $\log_{10}$ of MAE and MSE over 1000 experiments for the standard and normalized robustness respectively. Table~\ref{tab:median_traj_trainsize} reports the median values of accuracy, MSE, MAE, and MRE distribution. Mean values of the quantiles for AE and RE are reported in table~\ref{tab:result_traj_trainsize}. 
Distribution of AE and RE for a randomly picked experiments are shown in Fig.~\ref{fig:absrel_single}. In Fig. \ref{fig:robvsabs}, instead, we plot AE versus robustness for a random run for the standard (left) and normalized (right) robustness.

 \begin{table}[h]
  	\begin{center}
  		\hspace*{0cm}
\begin{tabular}{lcccccccc}
\toprule
    & \multicolumn{2}{c}{MSE} & \multicolumn{2}{c}{MAE} & \multicolumn{2}{c}{MRE} & \multicolumn{2}{c}{ACC} \\
    &   $\rho$ &  $\rhon$ &   $\rho$ &  $\rhon$ & $\rho$ &  $\rhon$ &$\rho$ &  $\rhon$ \\
\midrule
n=1 &  0.0165 &   0.000945 & 0.0537 &   0.0126 &  0.189 &   0.086 &  0.989 &   0.995 \\
n=2 &  0.0219 &   0.00248 & 0.0718 &   0.0236 &  0.271 &   0.157 &  0.984 &   0.989 \\
n=3 &  0.0873 &   0.0331 &0.0873 &   0.0331 &  0.340 &   0.214 &  0.980 &   0.985 \\
\bottomrule
\end{tabular}
\vspace{0.3cm}
\caption{Median of accuracy MSE, MAE, and MRE distribution of 1000 experiments for prediction of the standard and normalised robustness on single trajectories  for dimensionality of signal $n=1,2,3$ sampled according to $\mu_0$.
}  		\label{tab:median_traj_trainsize}
  	\end{center}
  \end{table}
  
 \begin{table}[]
  	\begin{center}
  		\hspace*{0cm}
\begin{tabular}{cccc|cc|cc|cc|cc}
\toprule
{} &{} & \multicolumn{2}{c}{5perc} & \multicolumn{2}{c}{1quart} & \multicolumn{2}{c}{median} & \multicolumn{2}{c}{3quart} & \multicolumn{2}{c}{95perc} \\
{} &{} &   $\rho$ &  $\rhon$ &   $\rho$ &  $\rhon$ &   $\rho$ &  $\rhon$ &   $\rho$ &  $\rhon$ &   $\rho$ &  $\rhon$ \\
\midrule
&n=1 & 0.00266 & 0.000689 & 0.0128 & 0.00333 & 0.0271 & 0.00816 & 0.0689 & 0.0252 & 0.479 & 0.223 \\
RE&n=2 & 0.00298 & 0.000928 & 0.0147 & 0.00477 & 0.0350 & 0.0137 & 0.105 & 0.0507 & 0.700 & 0.412 \\
&n=3 & 0.00349 & 0.00115 & 0.0175 & 0.00618 & 0.0445 & 0.0194 & 0.141 & 0.0748 & 0.870 & 0.564 \\
\midrule
&n=1 & 0.00217 & 0.000487 & 0.0105 & 0.00233 & 0.0233 & 0.00541 & 0.0545 & 0.0129 & 0.199 & 0.0611 \\
AE&n=2 & 0.00256 & 0.000635 & 0.0127 & 0.00322 & 0.0306 & 0.00856& 0.0786 & 0.0262 & 0.286 & 0.111 \\
&n=3 & 0.00310 & 0.000774 & 0.0155 & 0.00415 & 0.0387 & 0.0119 & 0.105 & 0.0389 & 0.345 & 0.149 \\
\bottomrule
\end{tabular}
\vspace{0.3cm}
\caption{Mean of  1000 experiments of the quantiles for Relative Error (RE)  and Absolute Error (AE) for prediction of the standard and normalised robustness on  single trajectories  for dimensionality of signal $n=1,2,3$. 
}
  		\label{tab:result_traj_trainsize}
  	\end{center}
  \end{table}
  
 \begin{figure}
    \centering
   \includegraphics[width=.48\linewidth]{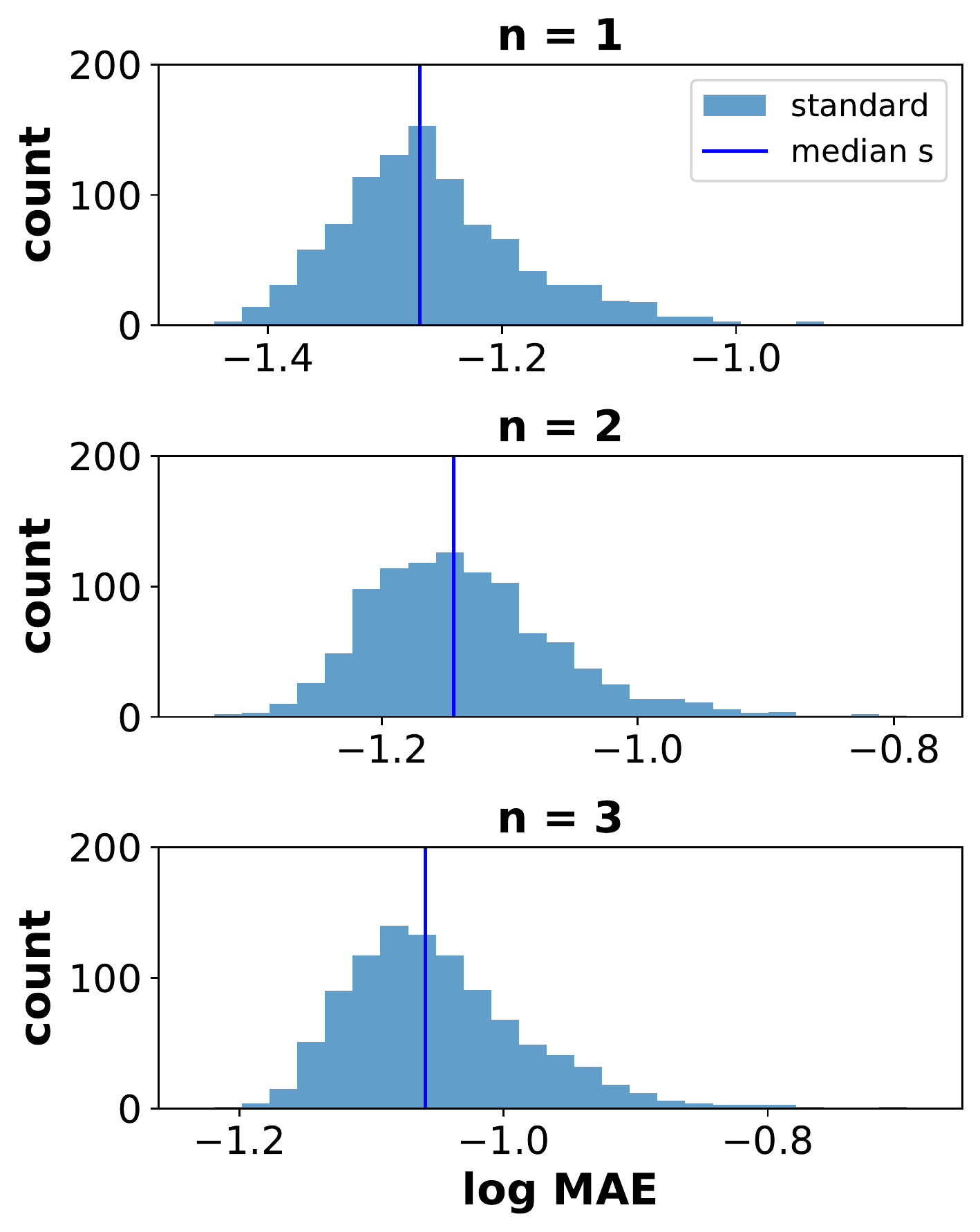}
   \includegraphics[width=.48\linewidth]{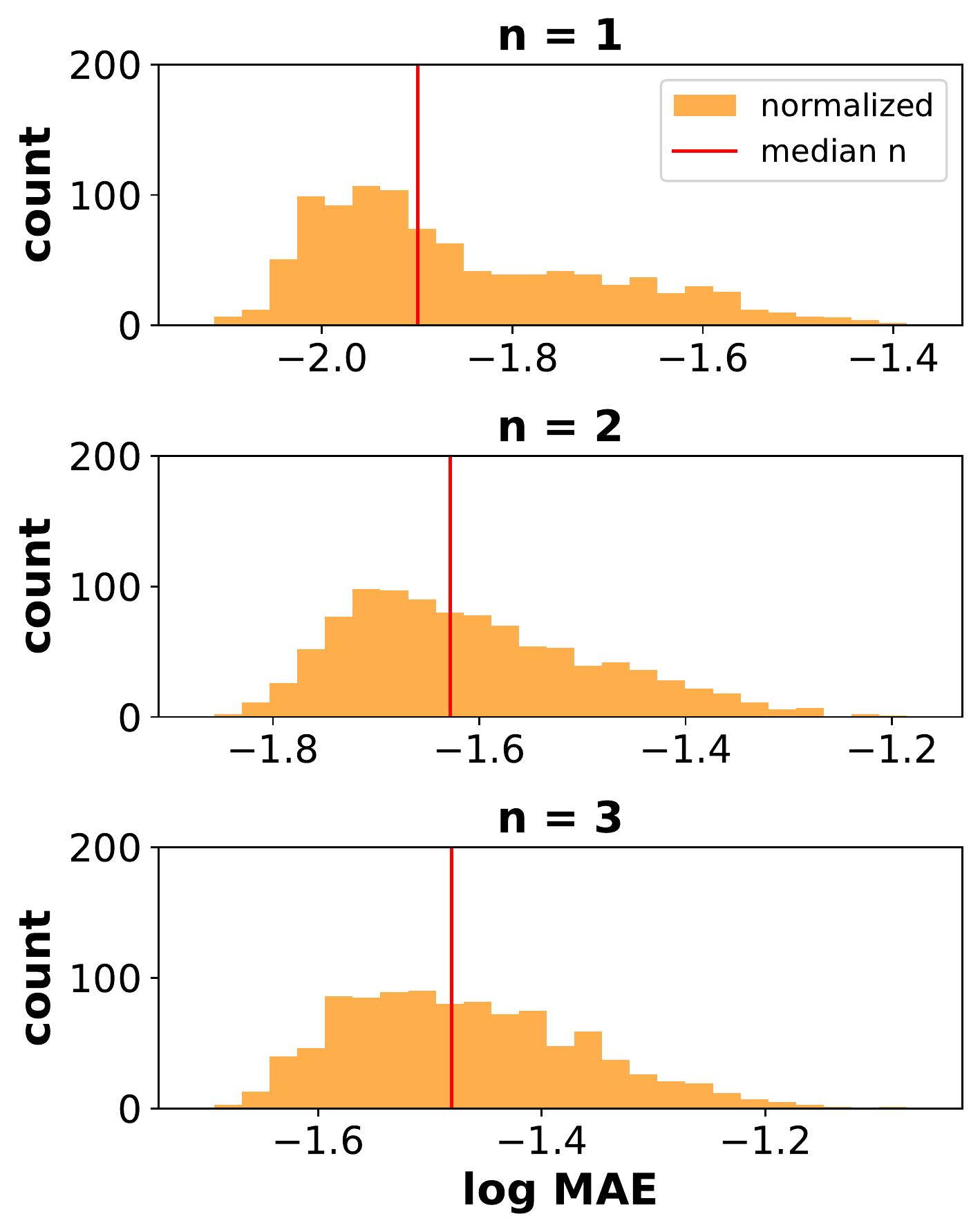}
    \caption{MAE over all 1000 experiments for standard and normalized robustness with trajectories sample from $\mu_0$ with dimensionality of signals $n=1,2,3$}
    \label{fig:mae_single}
\end{figure}

\begin{figure}
    \centering
    \includegraphics[width=.48\linewidth]{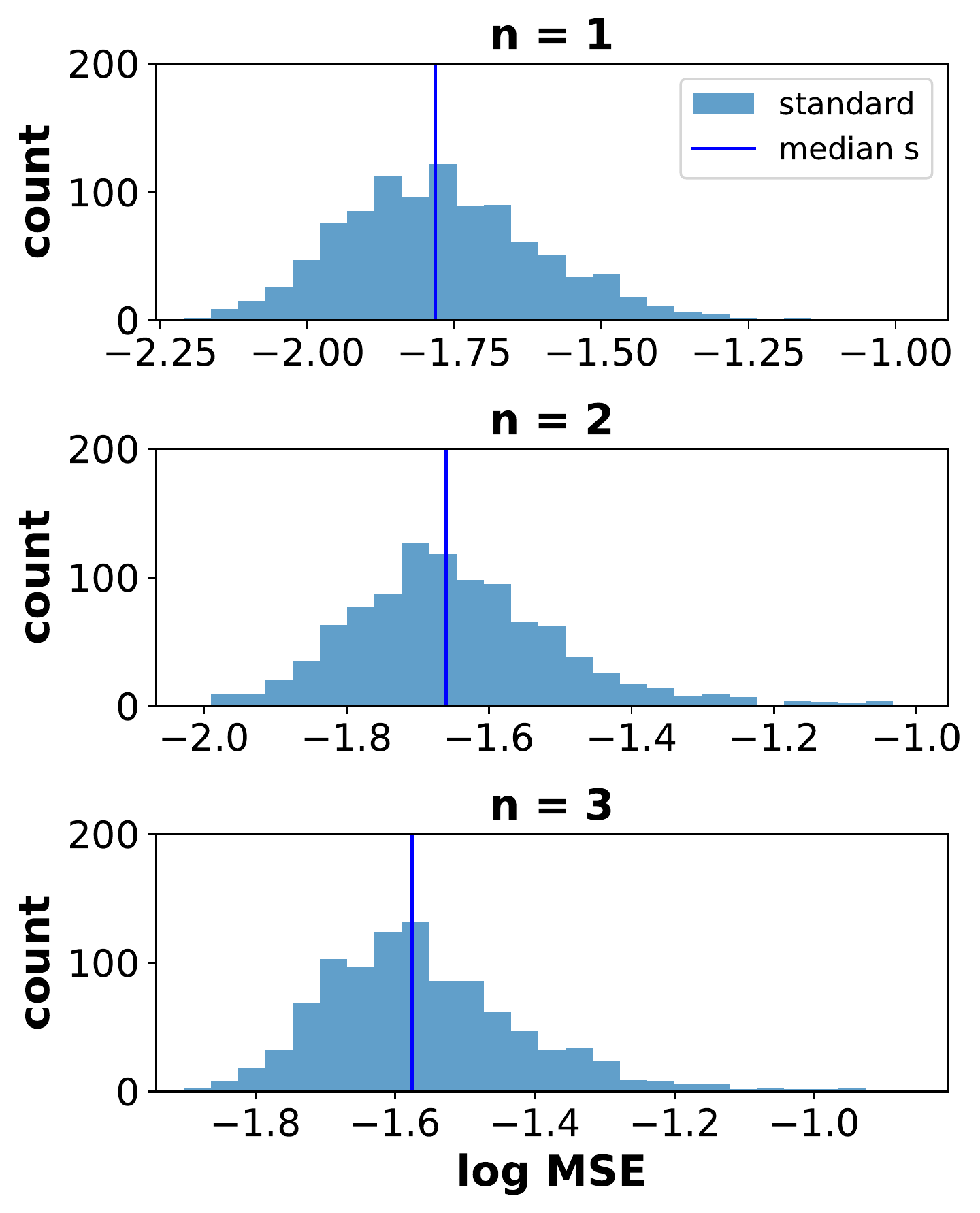}
     \includegraphics[width=.48\linewidth]{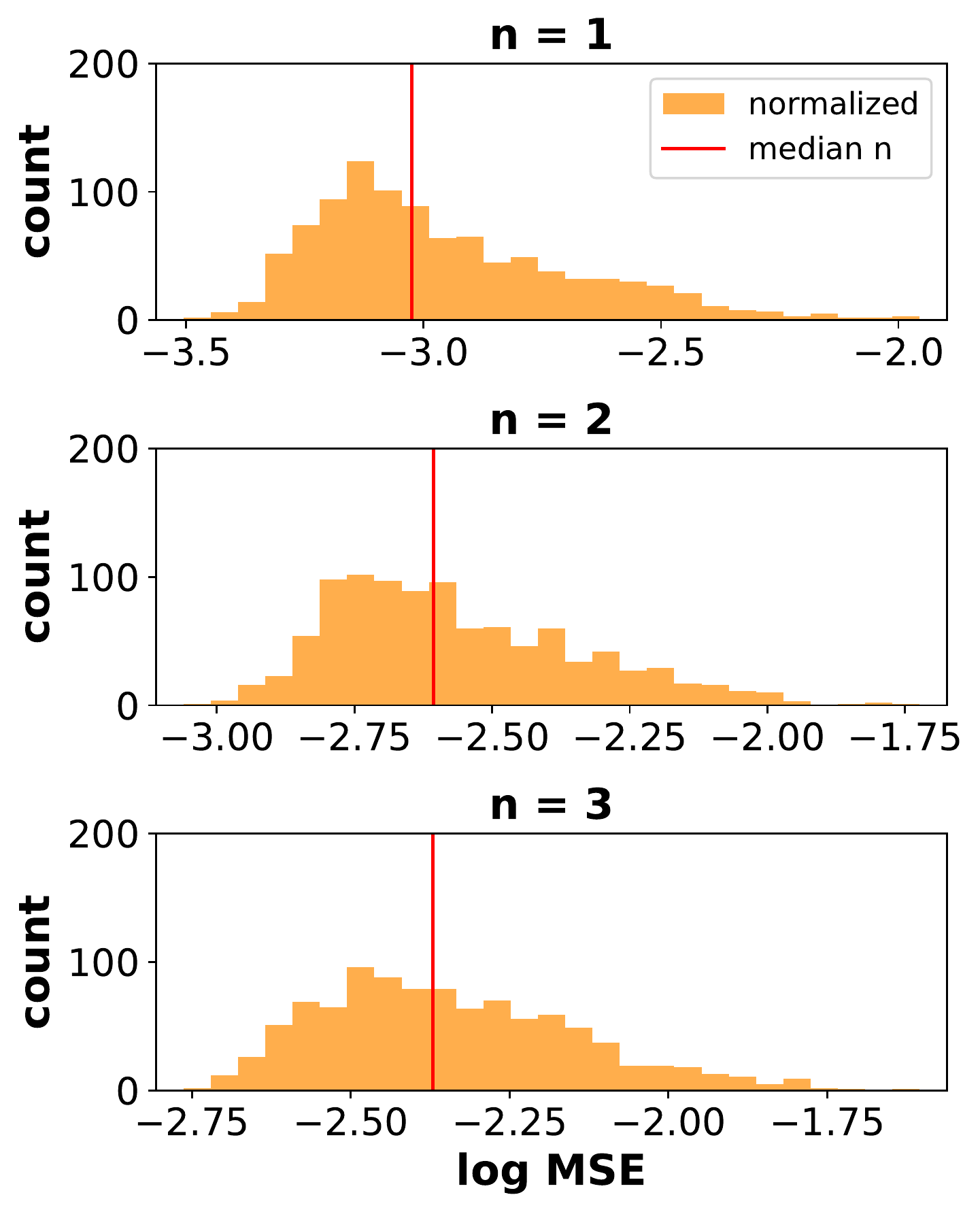}
    \caption{MSE over all 1000 experiments standard and normalized robustness for samples from $\mu_0$ with dimensionality of signals $n=1,2,3$}
    \label{fig:mse_single}
\end{figure}

\begin{figure}
    \centering
     \includegraphics[width=0.48\linewidth]{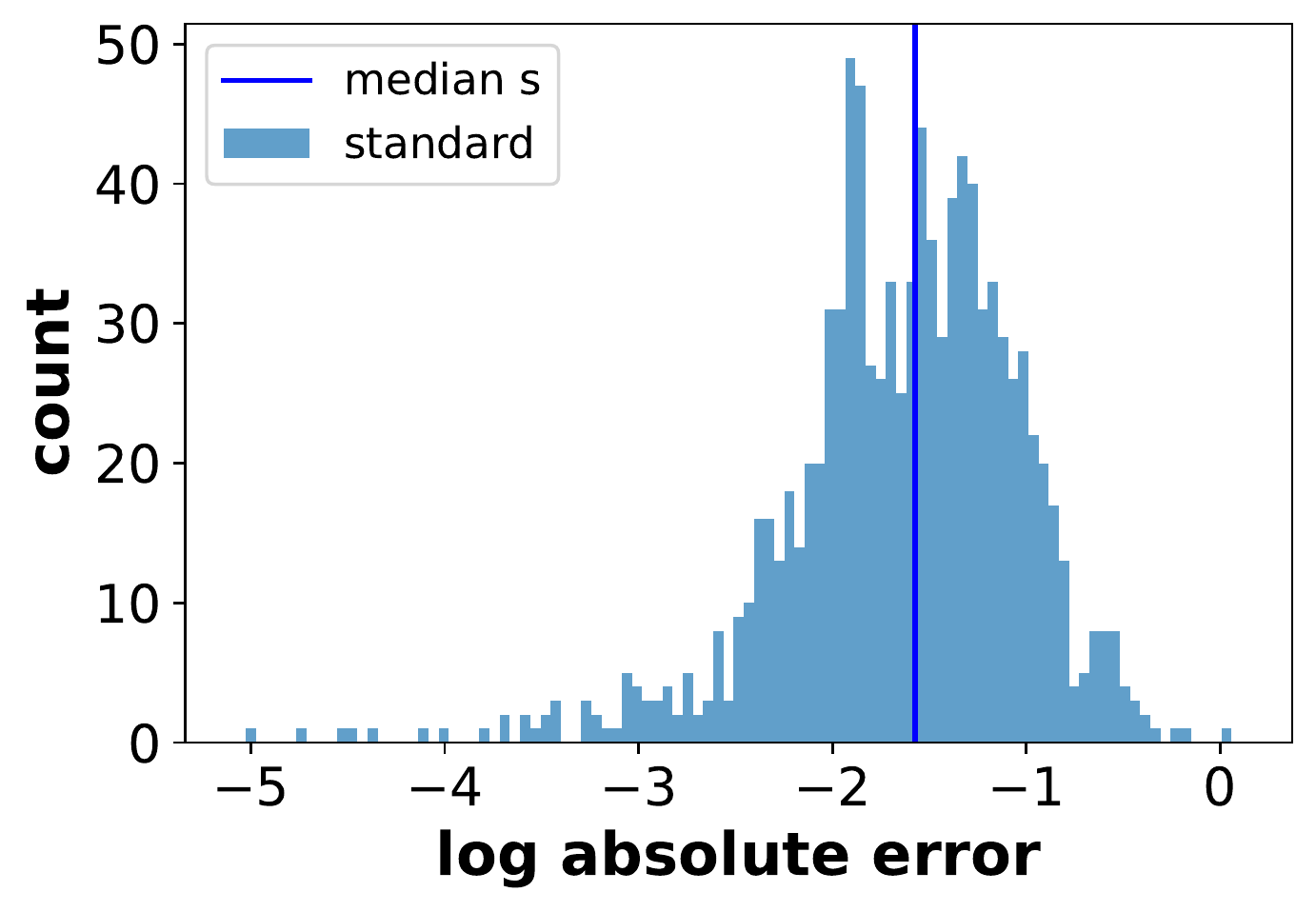}
    \includegraphics[width=0.48\linewidth]{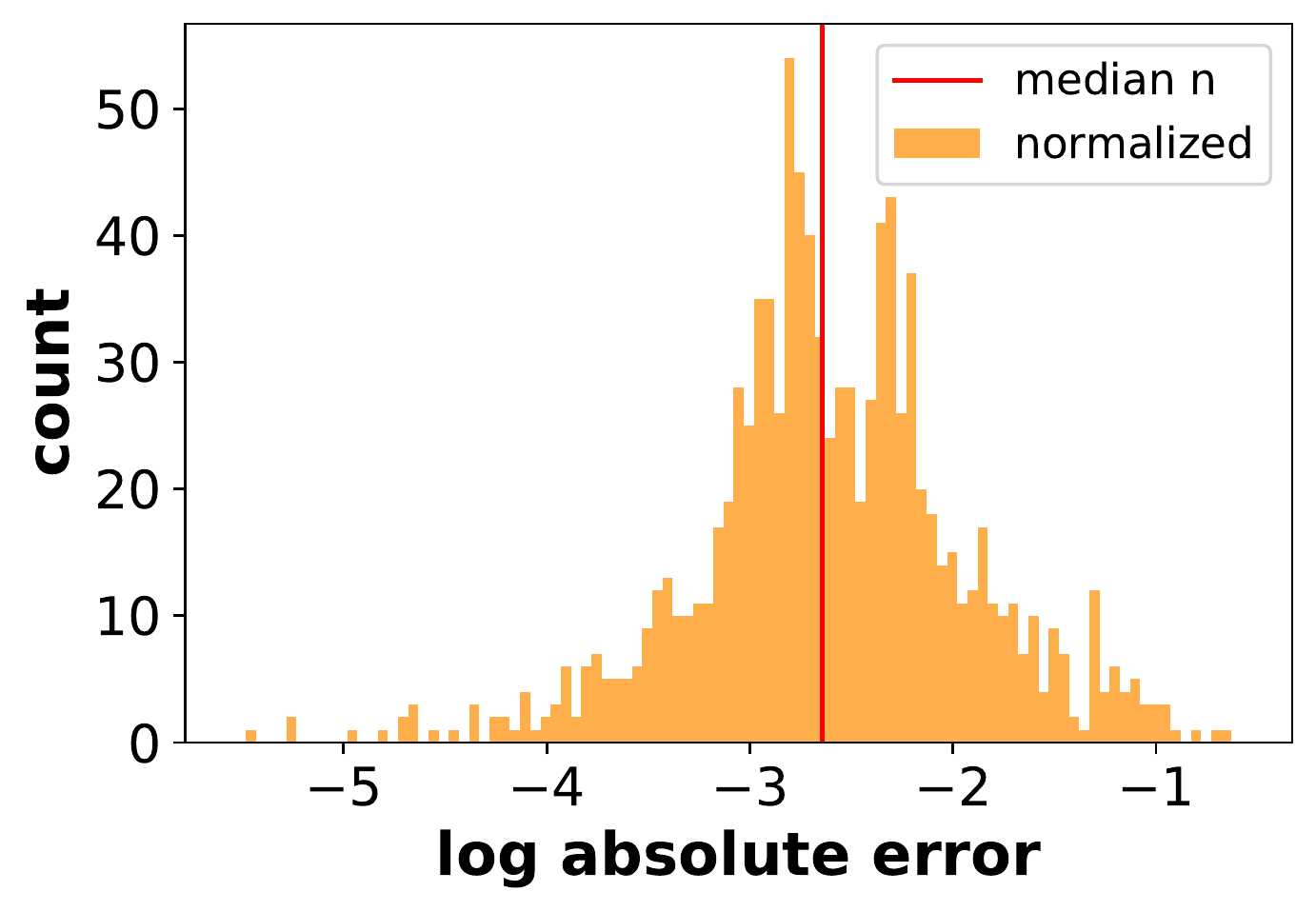}
    \caption{AE in predicting  standard (left) and normalized (right) robustness for a random experiment, with trajectories sample from $\mu_0$ with dimensionality of signals $n=1$
    }
    \label{fig:absrel_single}
\end{figure}

\begin{figure}
  	\vspace*{-1.1cm}
    \centering
    \includegraphics[width=0.49\linewidth]{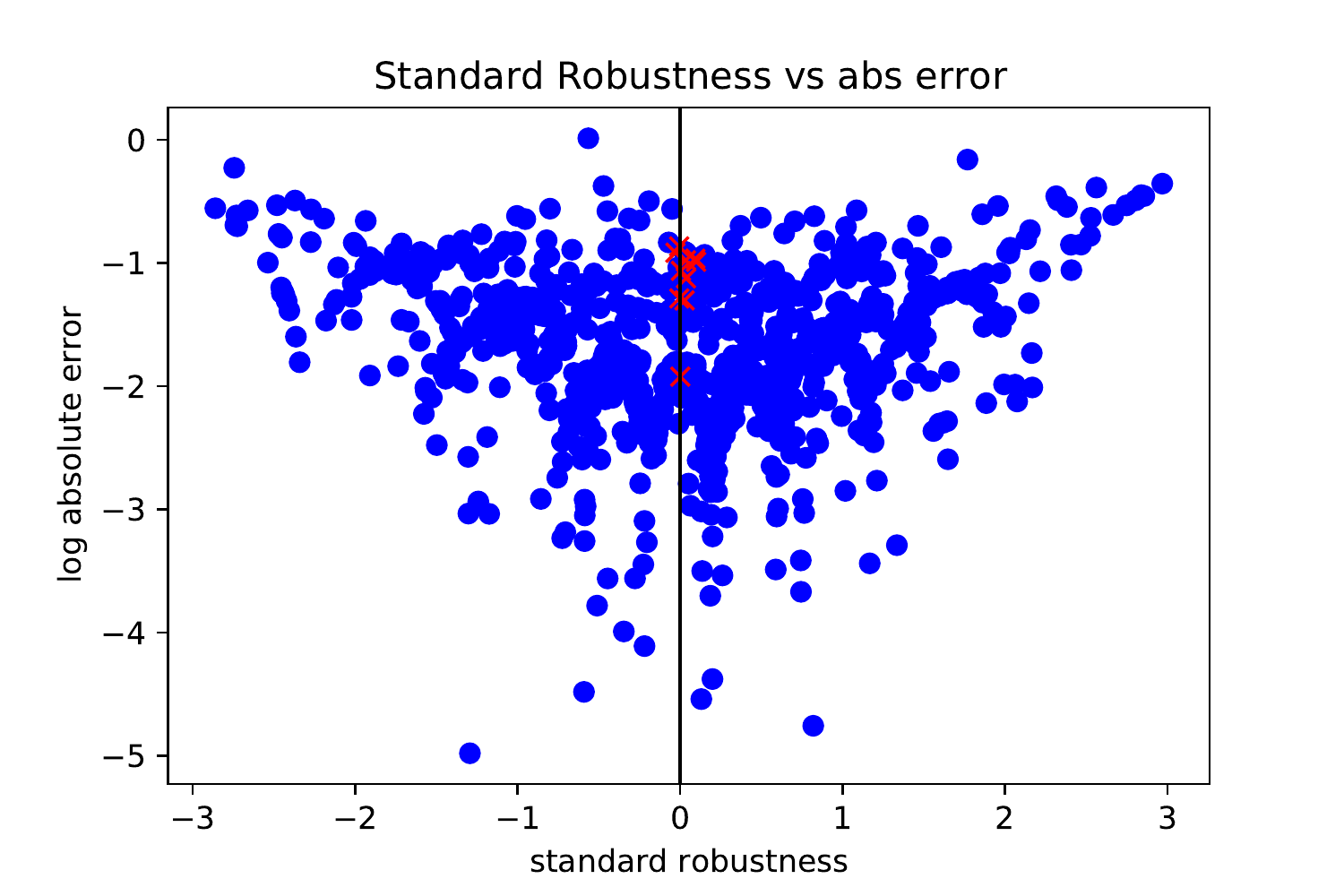}
    \includegraphics[width=0.49\linewidth]{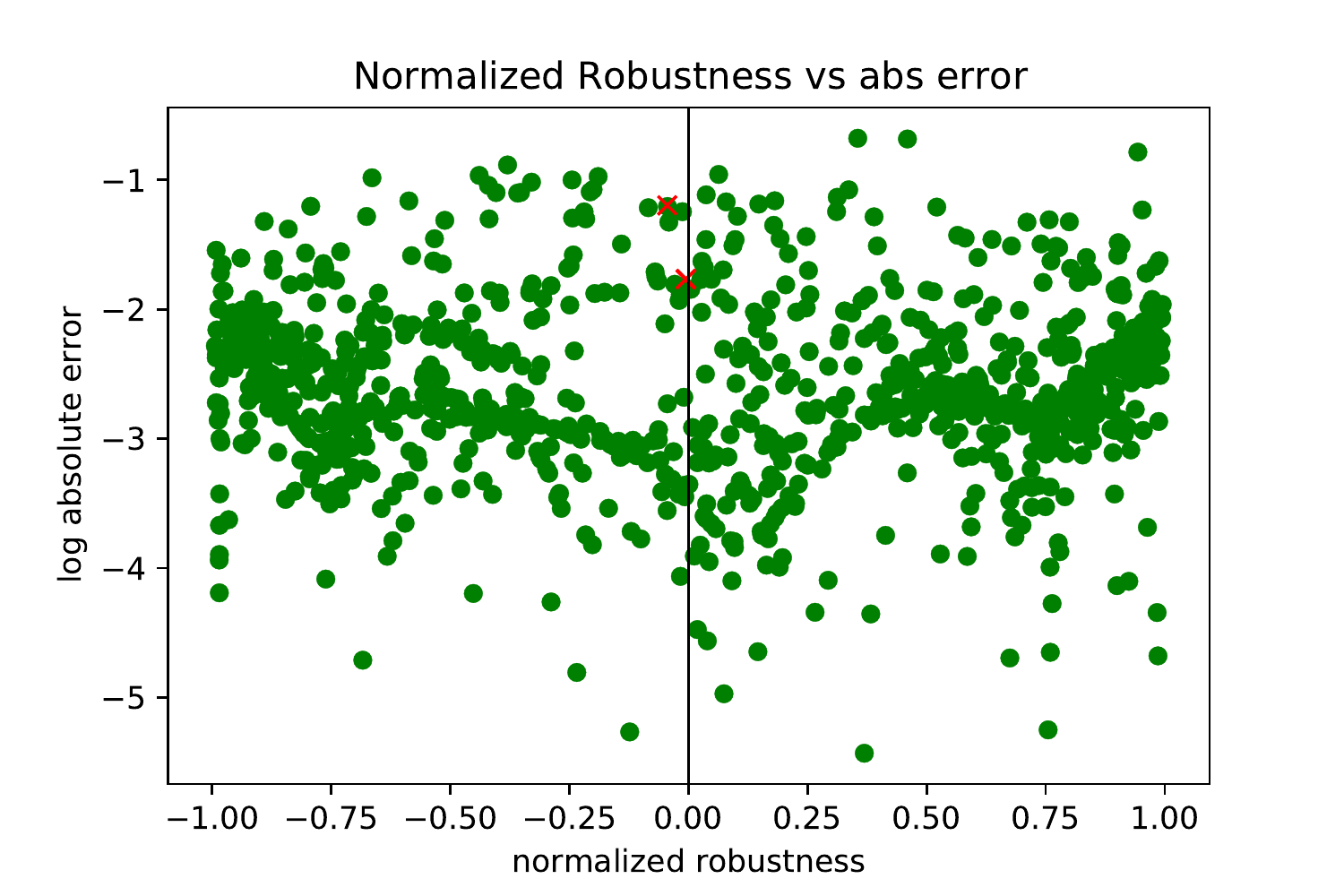}
    \caption{Robustness vs log AE for a single experiment  for prediction of the standard (left) and normalized (right) robustness 
     on  single trajectories sampled from $\mu_0$. The misclassified formulae are the red crossed.}
    \label{fig:robvsabs}
\end{figure}
We also test some properties of the ARCH-COMP 2020 \cite{ARCH20} to show that it works well even on real formulae. In particular, we consider the  properties AT1 and AT51 of the Automatic Transmission (AT) Benchmark, and  the  properties AFC27 of the Fuel Control of an Automotive Powertrain (AFC). We obtained an accuracy always equal to 1, median AE = $0.0229$, median  RE = $0.0316$ for AT1, median AE = $0.00814$, median RE = $0.00859$ for AT51, and median AE = $0.0146$, median RE = $0.0160$  for AFC27. 

\newpage
\subsection{Expected Robustness}
Further results on experiment for predicting the expected robustness.
In terms of error on the robustness itself, we plot in Fig. \ref{fig:avrob_absrel_all} and \ref{fig:avrob_absrel_all_normtrue} the distribution of MAE and MRE for the standard and normalize expected robustness respectively. 
In table \ref{tab:quantile_average_rob_mu0} we report the mean of the quantiles for AE and RE of 500 experiments.
  Distributions of these quantities over the test set for a randomly picked run are shown in Fig.  \ref{fig:avrob_absrel_normtrue}. 
In Fig. \ref{fig:avgrobvsrelerr}, instead, we plot relative error versus expected robustness for a random run.

\begin{figure}
    \centering
    \includegraphics[width=0.49\linewidth]{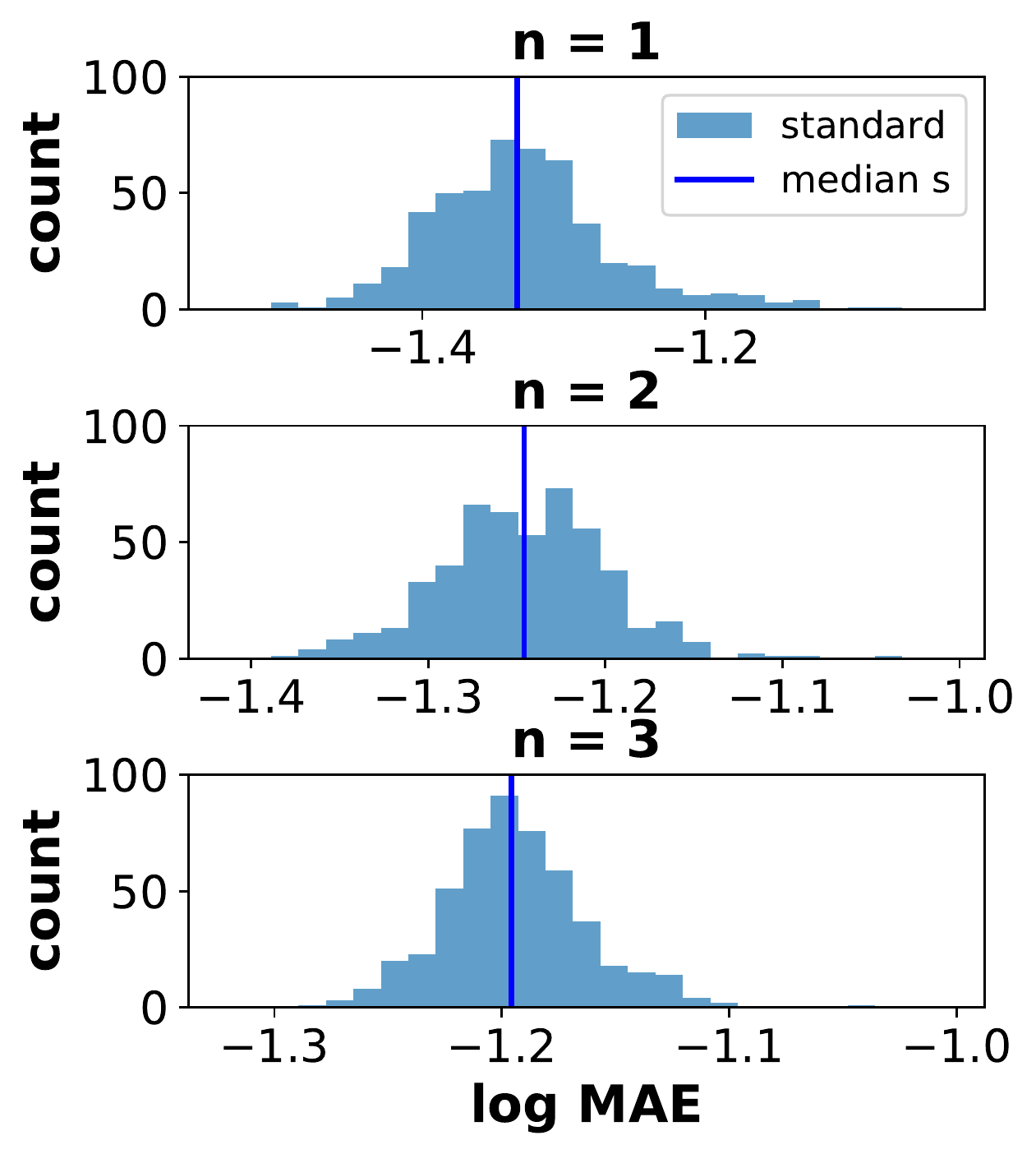}
    \includegraphics[width=0.48\linewidth]{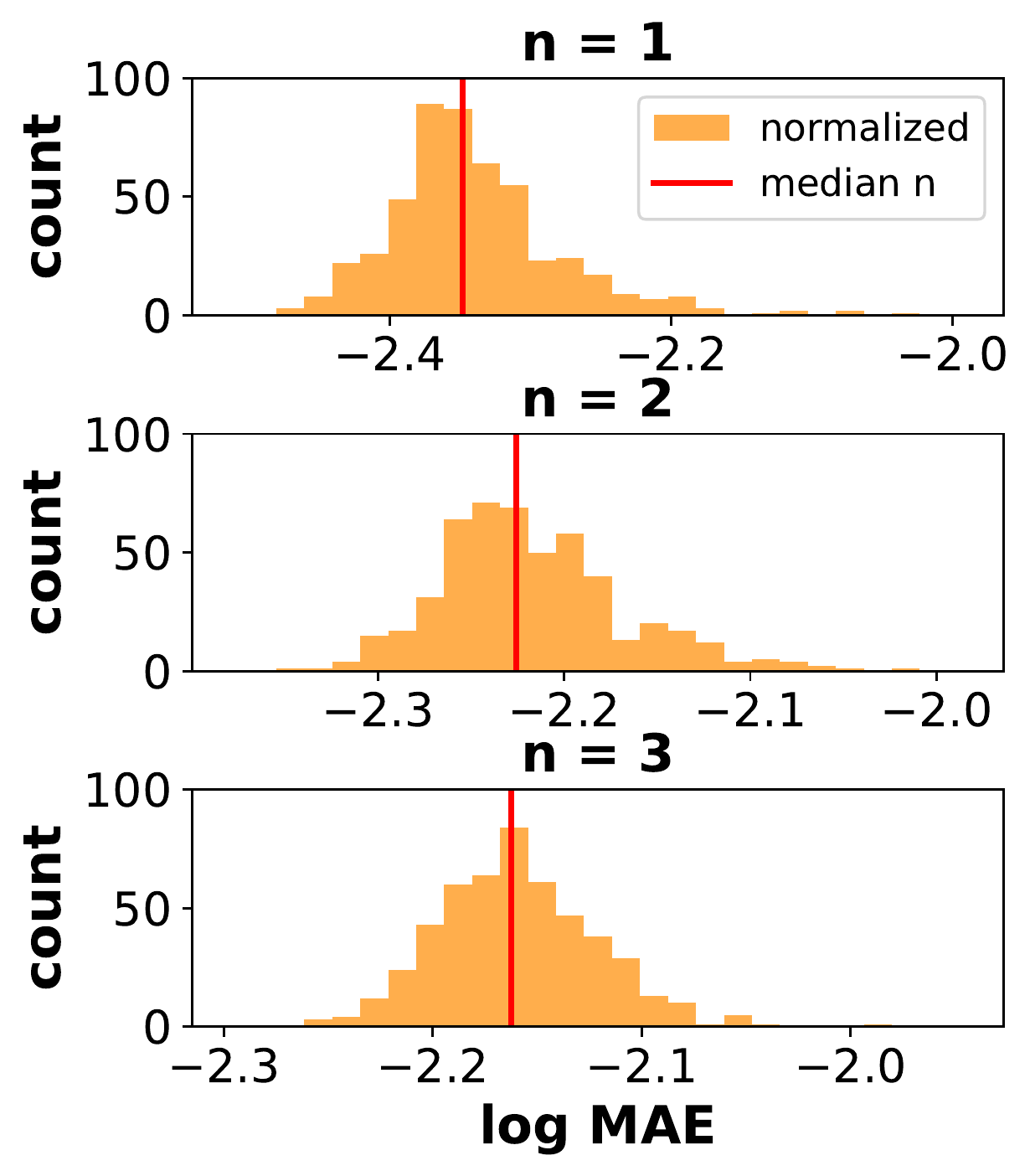}
    \caption{Absolute mean errors for robustness over all experiments for standard and normalized expected robustness with trajectories sample from $\mu_0$}
    \label{fig:avrob_absrel_all}
\end{figure}
\begin{figure}
    \centering
    \includegraphics[width=0.48\linewidth]{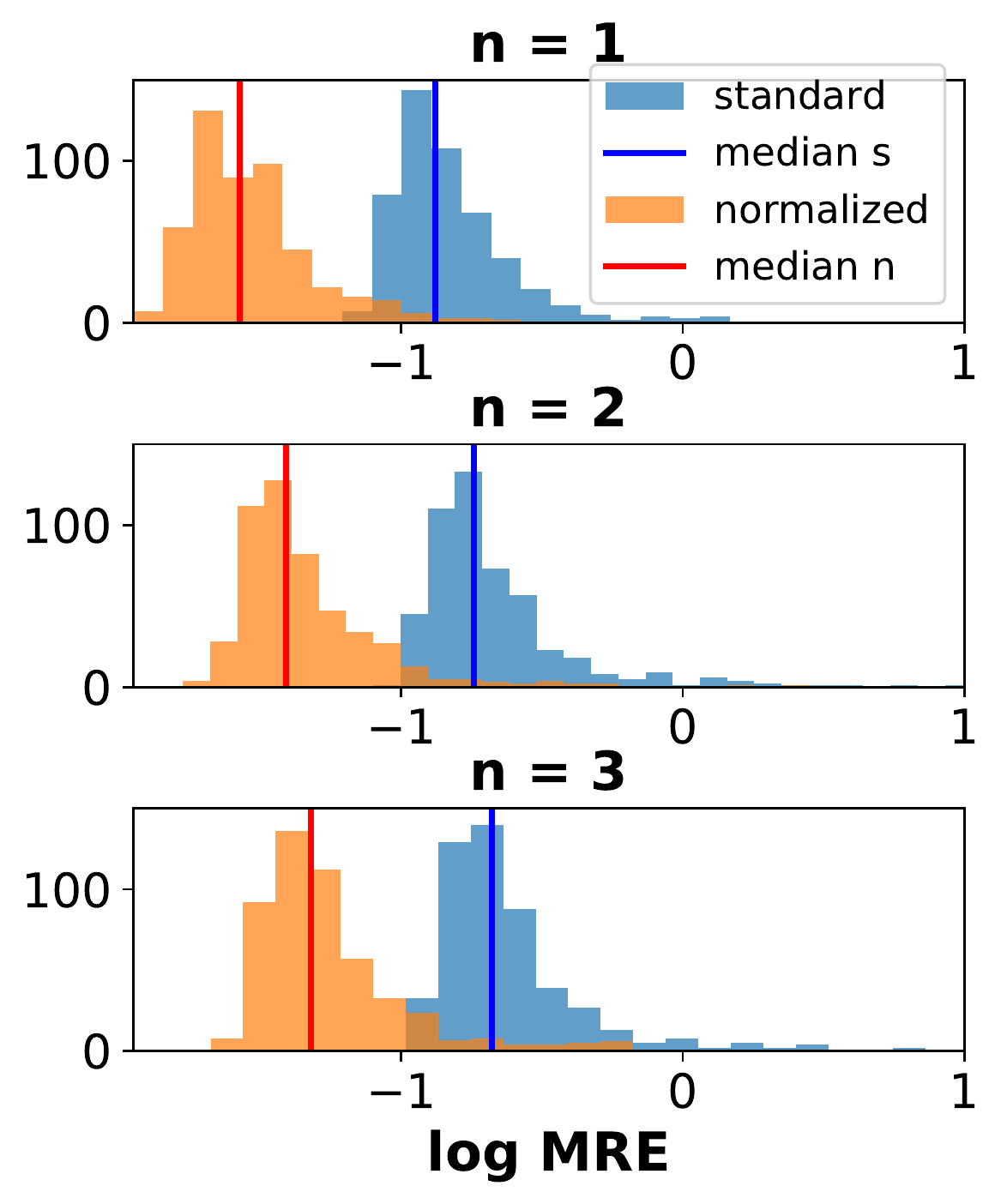}
    \includegraphics[width=0.48\linewidth]{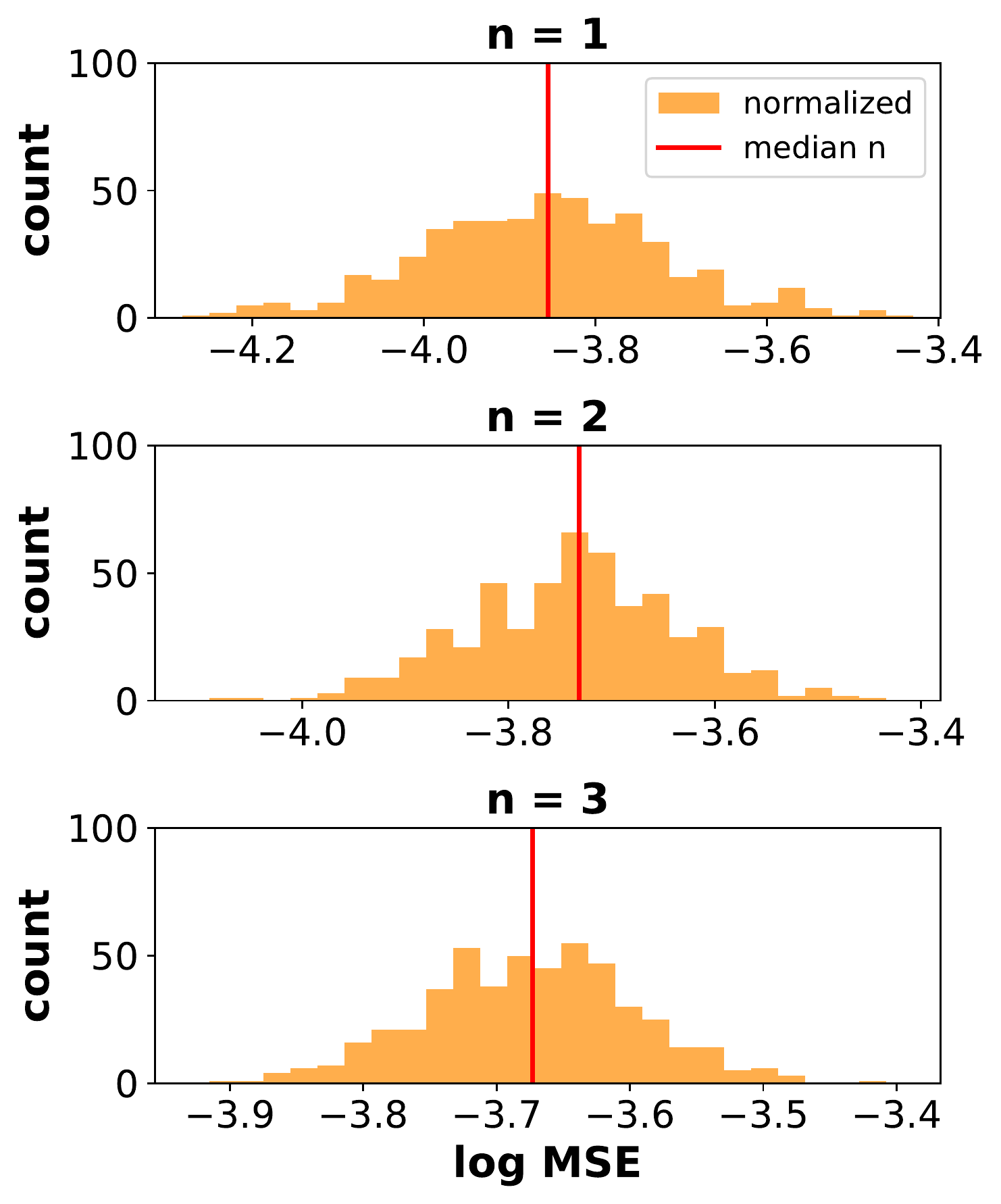}
    \caption{(left) MRE over all 500 experiments for standard and normalized expected robustness for samples from $\mu_0$ (right) MSE for the normalized robustness}
    \label{fig:avrob_absrel_all_normtrue}
\end{figure}

 \begin{table}[]
  	\begin{center}
  		\hspace*{0cm}
\begin{tabular}{lllllll}
\toprule
{} & \multicolumn{2}{l}{MSE} & \multicolumn{2}{l}{MAE} & \multicolumn{2}{l}{MRE} \\
{} &  $\rho$& $\rhon$ & $\rho$ & $\rhon$ &  $\rho$ & $\rhon$ \\
\midrule
n=1 &  0.0126 &   0.000140 &  0.0464 &   0.00448 &  0.132 &   0.0267 \\
n=2 &  0.0143 &   0.000185 &  0.0568 &   0.00593 &  0.181 &   0.0392 \\
n=3 &  0.0150 &   0.000212 &  0.0637 &   0.00688 &  0.210 &   0.0478 \\
\bottomrule
\end{tabular}
    \vspace{0.3cm}
\caption{Median for MSE, MAE and MRE of 500 experiments for prediction of the standard $\rho$ and normalised $\rhon$ expected robustness on trajectories sampled according to $\mu_0$}
  		\label{tab:result_average_rob}
  	\end{center}
  \end{table}


\begin{figure}
    \centering
       \includegraphics[width=0.48\linewidth]{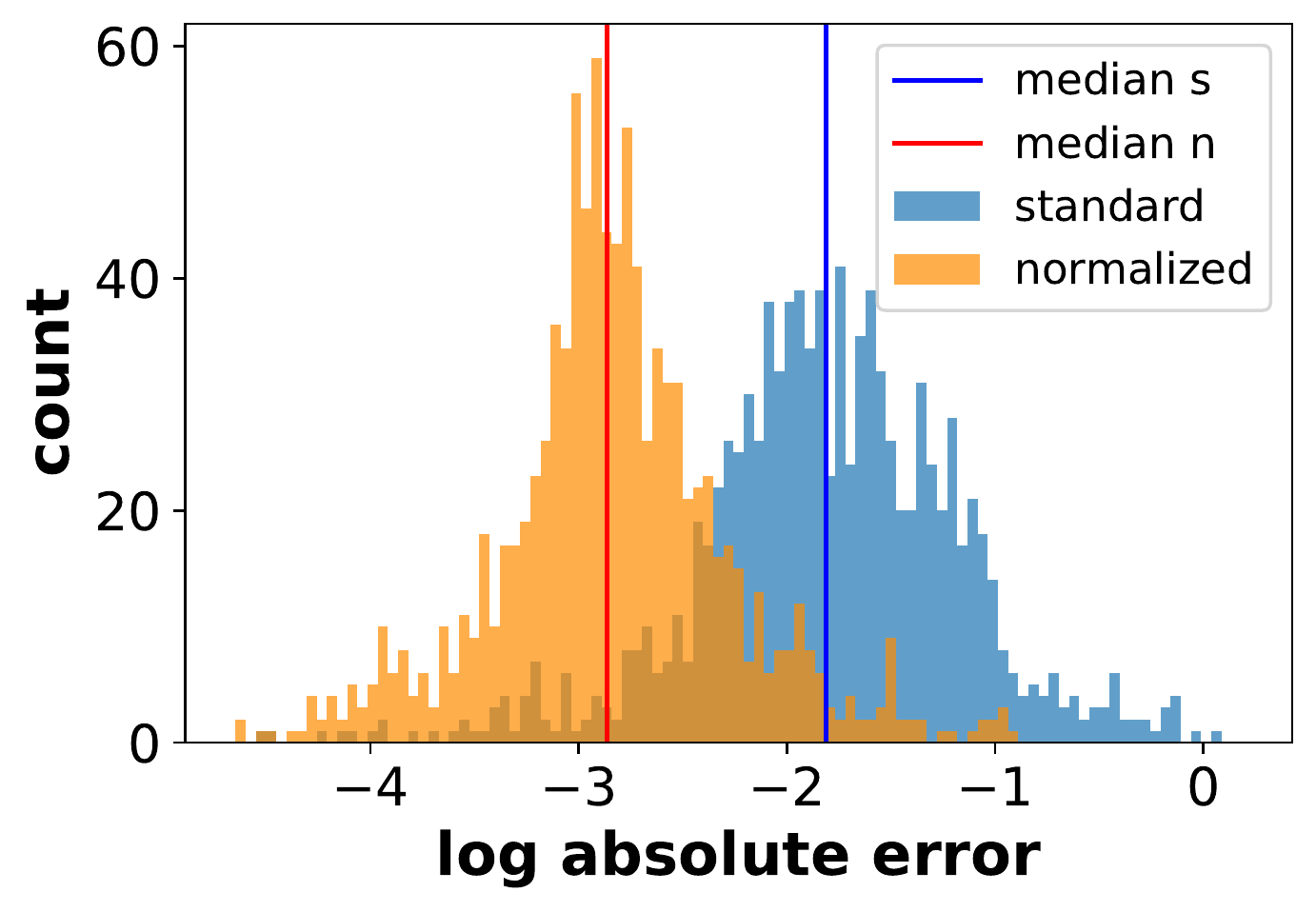}
    \includegraphics[width=0.48\linewidth]{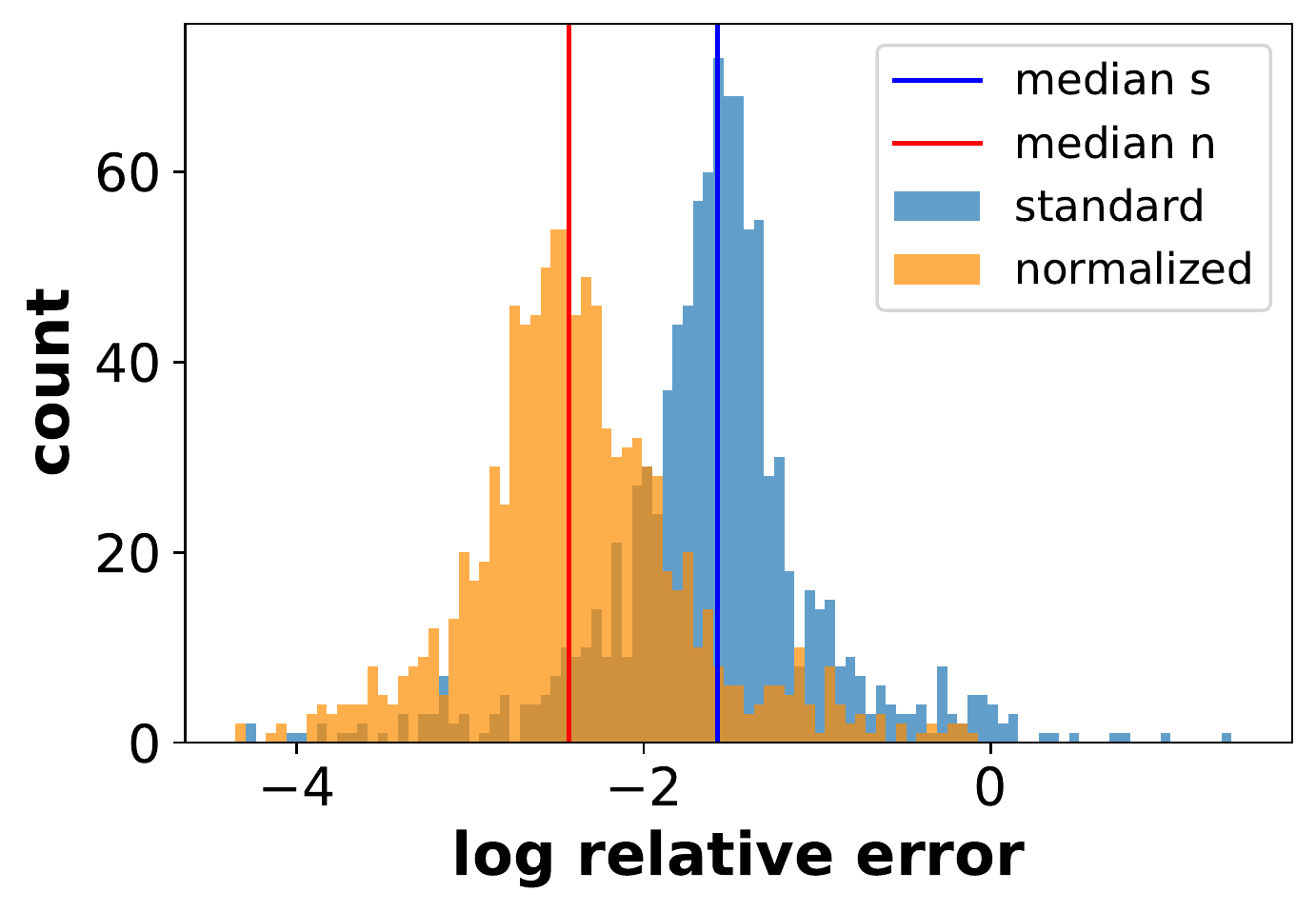}
    \vspace{0.3cm}
    \caption{Absolute (left) and relative (right) errors in predicting average standard and normalized robustness for a random experiment}
    \label{fig:avrob_absrel_normtrue}
\end{figure}

\begin{figure}
    \centering
    \includegraphics[width=1.\linewidth]{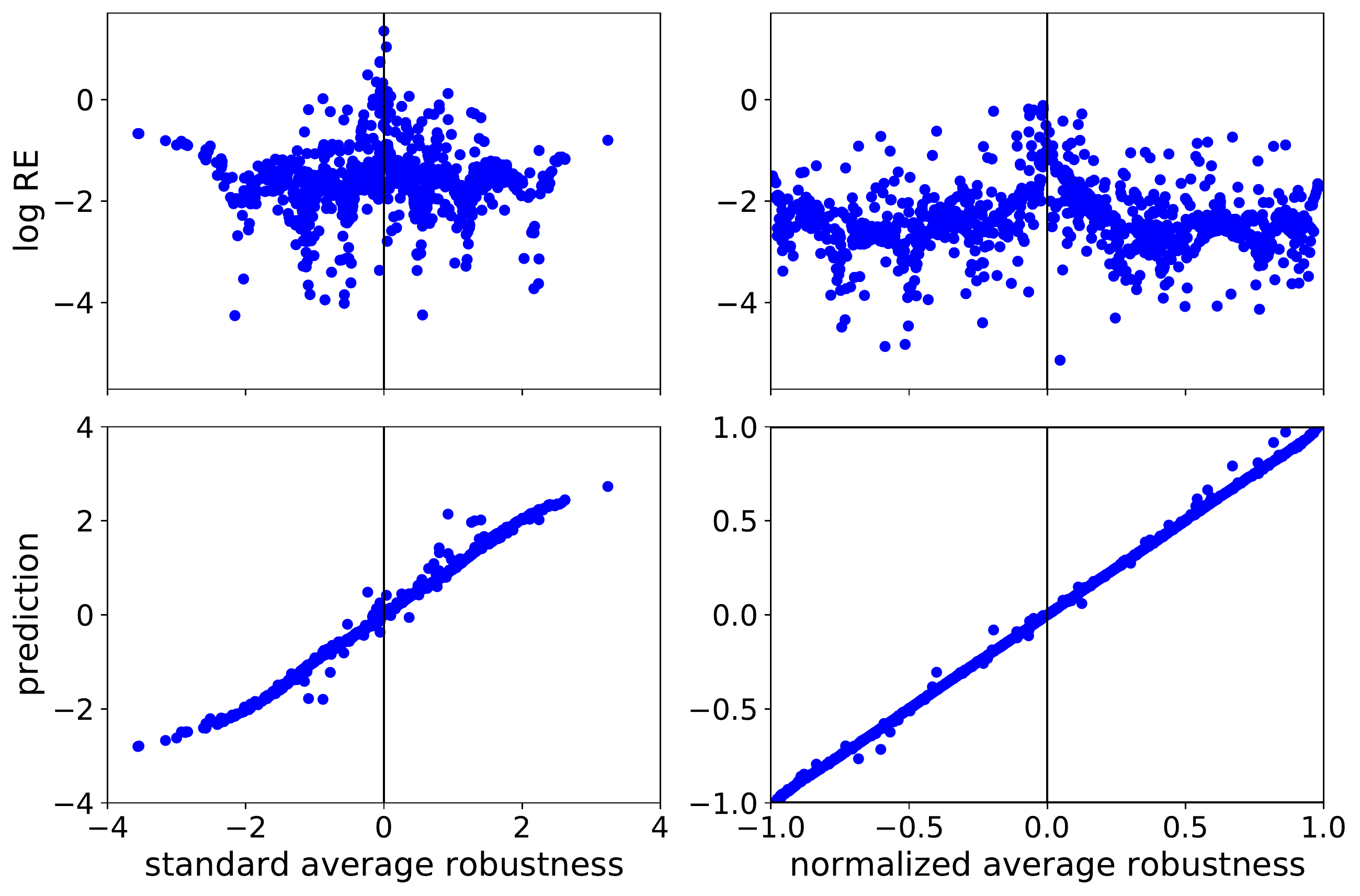}
    \caption{expected robustness vs relative error and predicted value for a single experiment for standard (left) and normalized (right) expected robustness}
    \label{fig:avgrobvsrelerr}
\end{figure}

   \begin{table}[]
  	\begin{center}
  		\hspace*{0cm}
\begin{tabular}{c|cc|cc|cc|cc|cc}
\toprule
{} & \multicolumn{2}{c}{5perc} & \multicolumn{2}{c}{1quart} & \multicolumn{2}{c}{median} & \multicolumn{2}{c}{3quart} & \multicolumn{2}{c}{95perc} \\
{} &   $\rho$ &  $\rhon$ &   $\rho$ &  $\rhon$ &   $\rho$ &  $\rhon$ &   $\rho$ &  $\rhon$ &   $\rho$ &  $\rhon$ \\
\midrule
n=1 & 0.00341 & 0.000521 & 0.0164 & 0.00249 & 0.0319 & 0.00522 & 0.0638 & 0.0111 & 0.344 & 0.0740 \\
n=2 & 0.00392 & 0.000543 & 0.0186 & 0.00269 & 0.0377 & 0.00607 & 0.0839 & 0.0159 & 0.484 & 0.115 \\
n=3 & 0.00445 & 0.000570 & 0.0211 & 0.00287 & 0.0439 & 0.00669 & 0.103 & 0.0196 & 0.548 & 0.133 \\
\midrule
n=1 & 0.00189 & 0.000188 & 0.00916 & 0.000890 & 0.0196 & 0.00189 & 0.0463 & 0.00413 & 0.159 & 0.0161 \\
n=2 & 0.00219 & 0.000196 & 0.011 & 0.000984 & 0.0247 & 0.00223 & 0.0585 & 0.00583 & 0.210 & 0.0249 \\
n=3 & 0.00250 & 0.000208 & 0.0125 & 0.00105 & 0.0293 & 0.00254 & 0.0704 & 0.00708 & 0.240 & 0.0286 \\
\bottomrule
\end{tabular}
    \vspace{0.3cm}
\caption{Mean of quantiles for RE and AE over 500 experiments for prediction of the standard $\rho$ and normalised $\rhon$ expected robustness on  5000 trajectories sampled according to $\mu_0$}
  		\label{tab:quantile_average_rob_mu0}
  	\end{center}
  \end{table}

 \newpage 
\subsection{Satisfaction Probability}
Further results on experiment for predicting the satisfaction probability.
In terms of error on the probability itself, we plot in Fig. \ref{fig:satprob_absrel_all}  the distribution of the MAE and MRE.  The Median and $95\%$ Confidence interval for MAE and MRE of 500 experiments are also reported in table \ref{tab:result_satprob}. 
In Table \ref{tab:quantile_sat_prob_mu0} we report the estimated median and quantiles for AE and RE on 500 experiments.

\begin{figure}
    \centering
    \includegraphics[width=.5\linewidth]{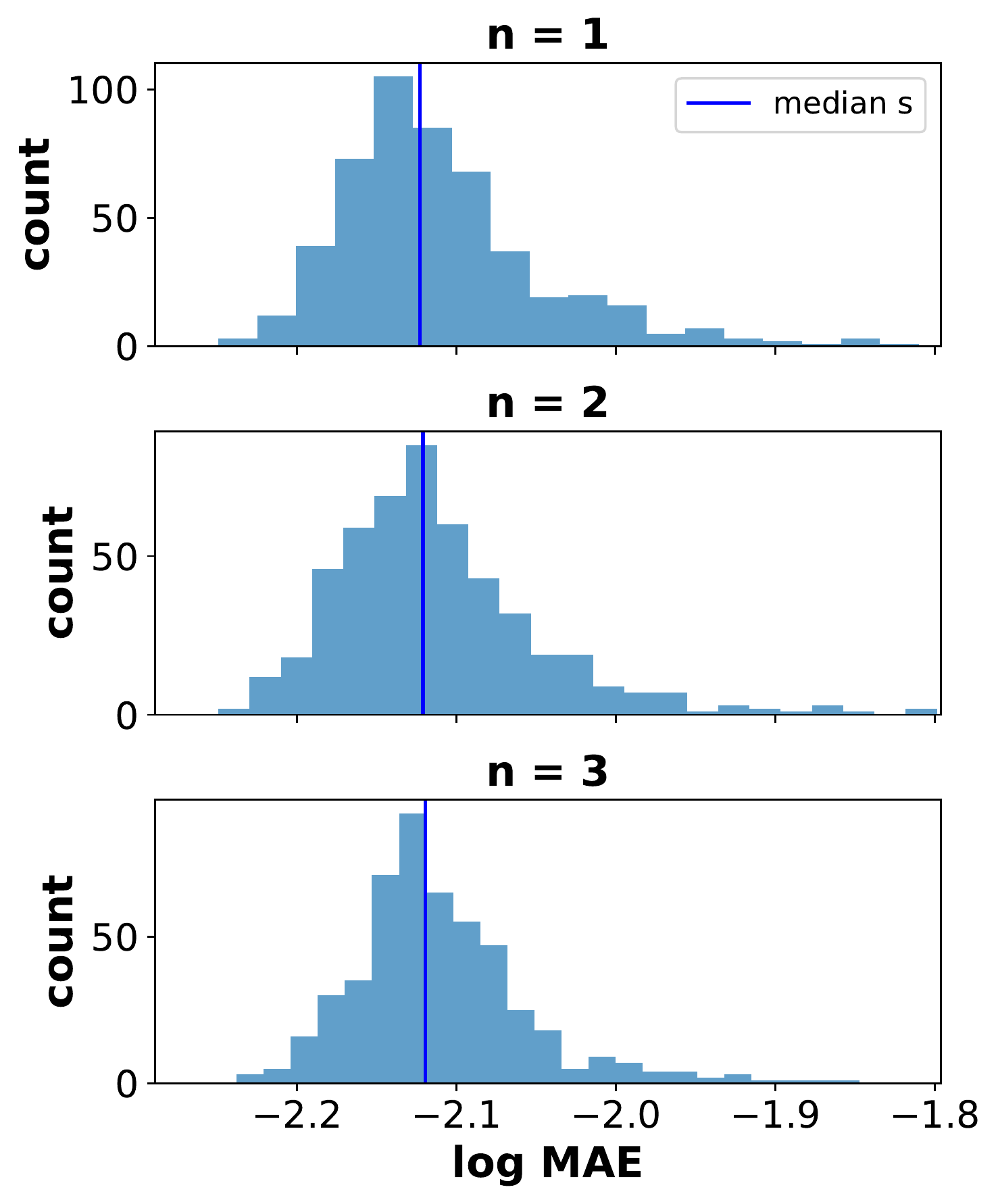}
        \includegraphics[width=.48\linewidth]{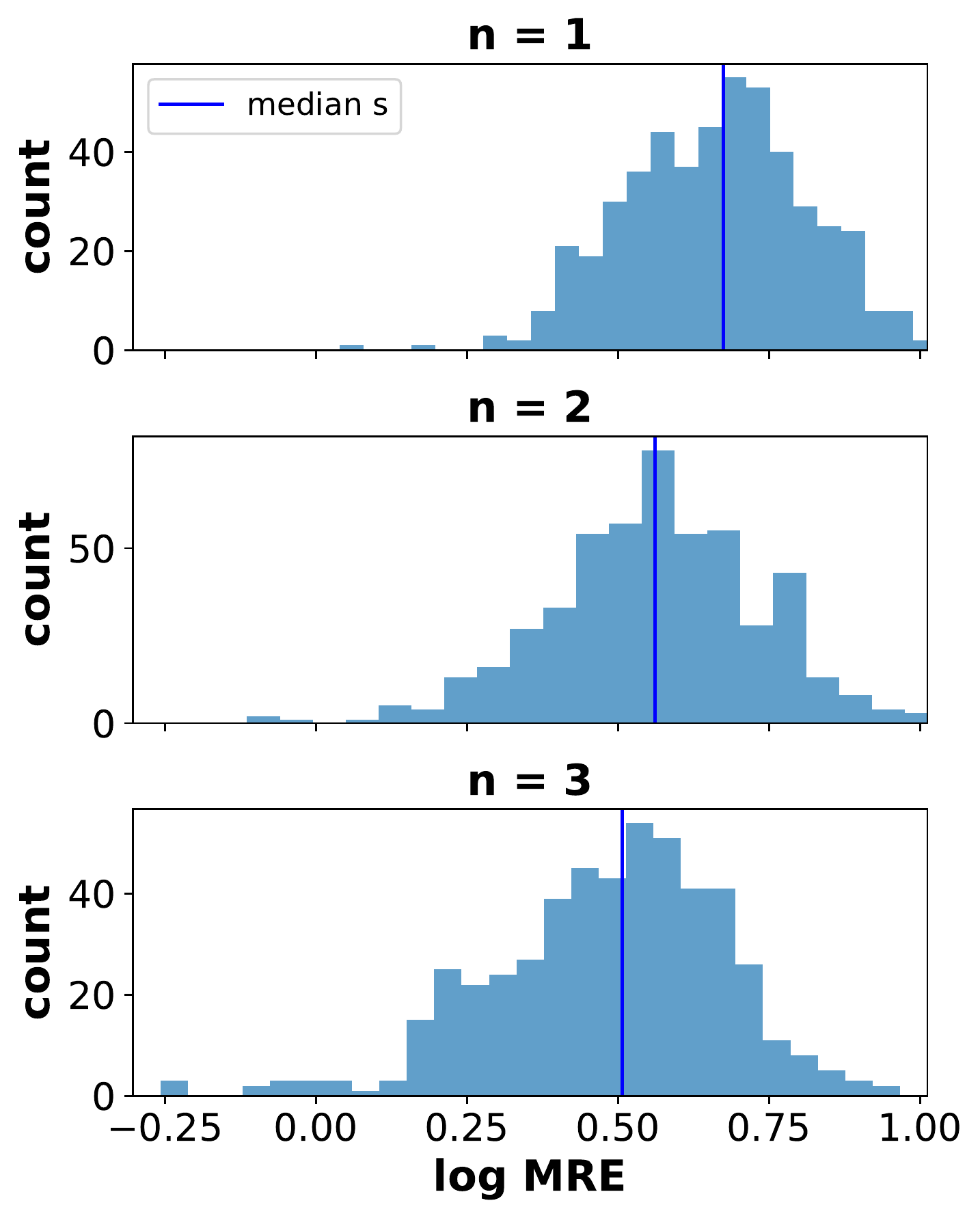}
    \caption{Absolute mean errors for satisfaction probability over all 500 experiments for standard robustness with trajectories sample from $\mu_0$ for trajectories with dimensionality $n=1,2,3$.}
    \label{fig:satprob_absrel_all}
\end{figure}

 \begin{table}[h!]
  	\begin{center}
  		\hspace*{0cm}
\begin{tabular}{lrrr}
\toprule
{} &      MSE &      MAE &       MRE \\
\midrule
n=1 &  0.000268 &  0.007534 &  4.722656 \\
n=2 &  0.000260 &  0.007568 &  3.640625 \\
n=3 &  0.000247 &  0.007591 &  3.214844 \\
\bottomrule
\end{tabular}
    \vspace{0.3cm}
\caption{Median for MSE MAE and MRE of 500 experiments for prediction of the satisfaction probability on trajectories sampled according to $\mu_0$}
  		\label{tab:result_satprob}
  	\end{center}
  \end{table}


   \begin{table}[]
  	\begin{center}
  		\hspace*{0cm}
\begin{tabular}{lrrrrrr|rrrrr}
\toprule
  		{} & \multicolumn{6}{c}{RE} & \multicolumn{4}{c}{AE} \\
  		\midrule
{} &    5perc &   1quart &   median &   3quart &   95perc &     99perc&    1quart &   median &   3quart &       99perc \\
\midrule
n=1 & 0.000189 & 0.00297 & 0.00755 & 0.0326 & 1.51 & 176 & 0.00112 &  0.00295 & 0.00647 &  0.0725 \\
n=2 & 0.000355 & 0.00289 & 0.00745 & 0.0299 & 0.876 & 130&  0.00123 &  0.00299 & 0.00669 & 0.0739 \\
n=3 & 0.000449 & 0.00309 & 0.00795 & 0.0309 & 0.586 &  81.8 & 0.00135 & 0.00322 & 0.00725 & 0.0722\\
\bottomrule
\end{tabular}
    \vspace{0.3cm}
\caption{Mean of quantiles for RE and AE over 500 experiments for prediction of the probability satisfaction with trajectories sample from $\mu_0$}
  		\label{tab:quantile_sat_prob_mu0}
  	\end{center}
  \end{table}

\subsection{Kernel Regression on  other stochastic processes}\label{app:exp:models}
Figure~\ref{fig:stoch_trajectories} reports the mean and standard deviation of 50 trajectories for the stochastic models: \emph{Immigration} (1 dim), \emph{Polymerase} (1 dim), \emph{Isomerization} (2 dim) and \emph{Transcription Intermediate} (3 dim), simulated using the Python library StochPy. \cite{maarleveld2013stochpy}.    
\begin{figure}[h!]
  	\begin{center}
  		\hspace*{-0.4cm}
  		\includegraphics[scale=0.35]{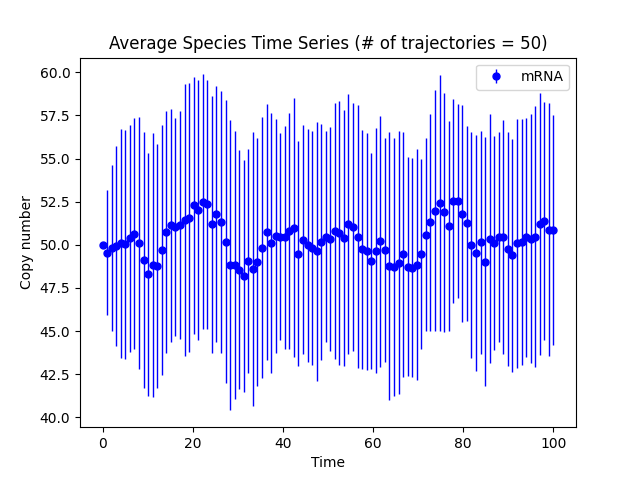}
  		\includegraphics[scale=0.35]{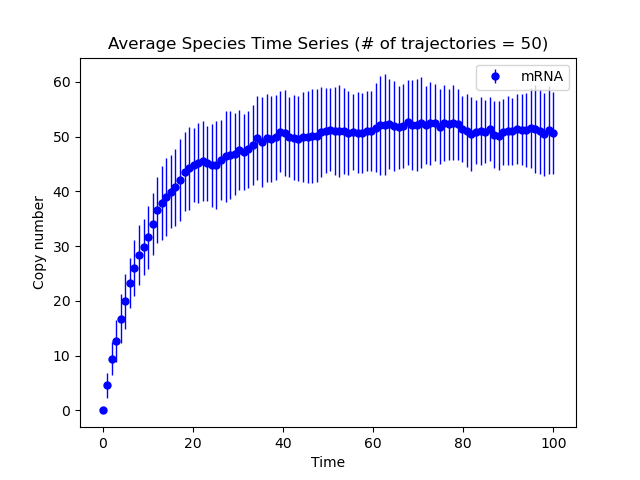}
  		\includegraphics[scale=0.35]{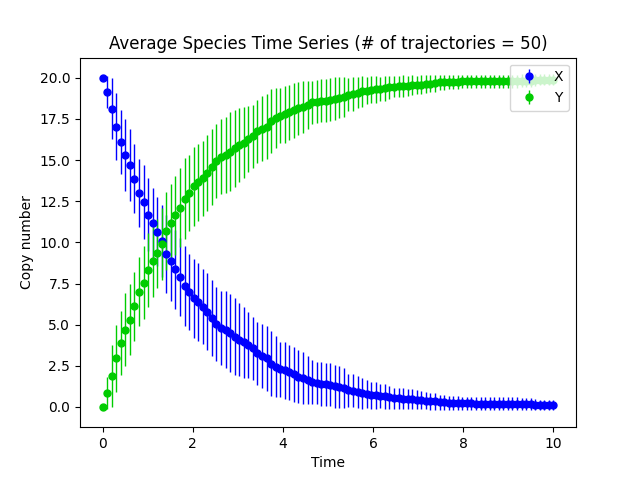}
  		\includegraphics[scale=0.35]{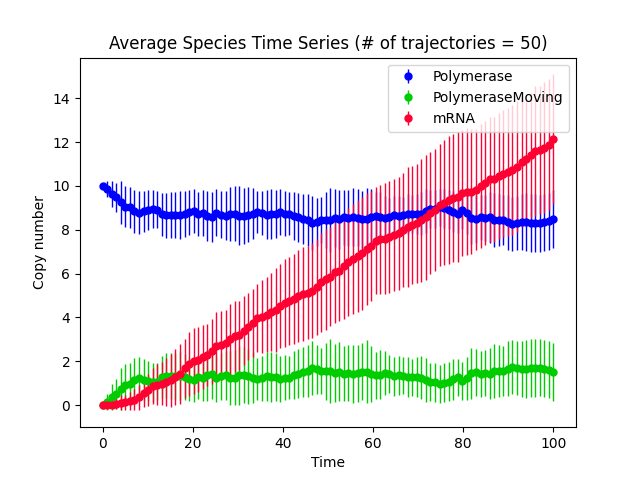}
  	\end{center}
  	\vspace*{-0cm}\caption{\label{fig:stoch_trajectories}From left to right, expected and standard deviation of 50 trajectories generated by the \emph{Immigration}, \emph{Polymerase},\emph{Isomerization} and \emph{Trancription Intermediate} model.}
  \end{figure}

\subsubsection{Robustness on single trajectories}
\label{app:subsec_stoch_single}
Further results on experiment for prediction of Boolean satisfiability of a formula using as a discriminator the sign of the robustness sampling trajectories on different stochastic models. 
Figures~\ref{fig:accuracy_rel_othermodels}, \ref{fig:mse_single_othermodels}, and \ref{fig:mae_single_othermodels} report the accuracy, MRE, MAE and MSE over all experiments for standard and normalized robustness for sample from \emph{Immigration} (1 dim), \emph{Isomerization} (2 dim) and \emph{Transcription} (3 dim). 
In table \ref{tab:quantile_rob_othermodels} we report the mean of the quantiles for AE and RE of 500 experiments.

\begin{figure}
    \centering
            \includegraphics[width=0.51\linewidth]{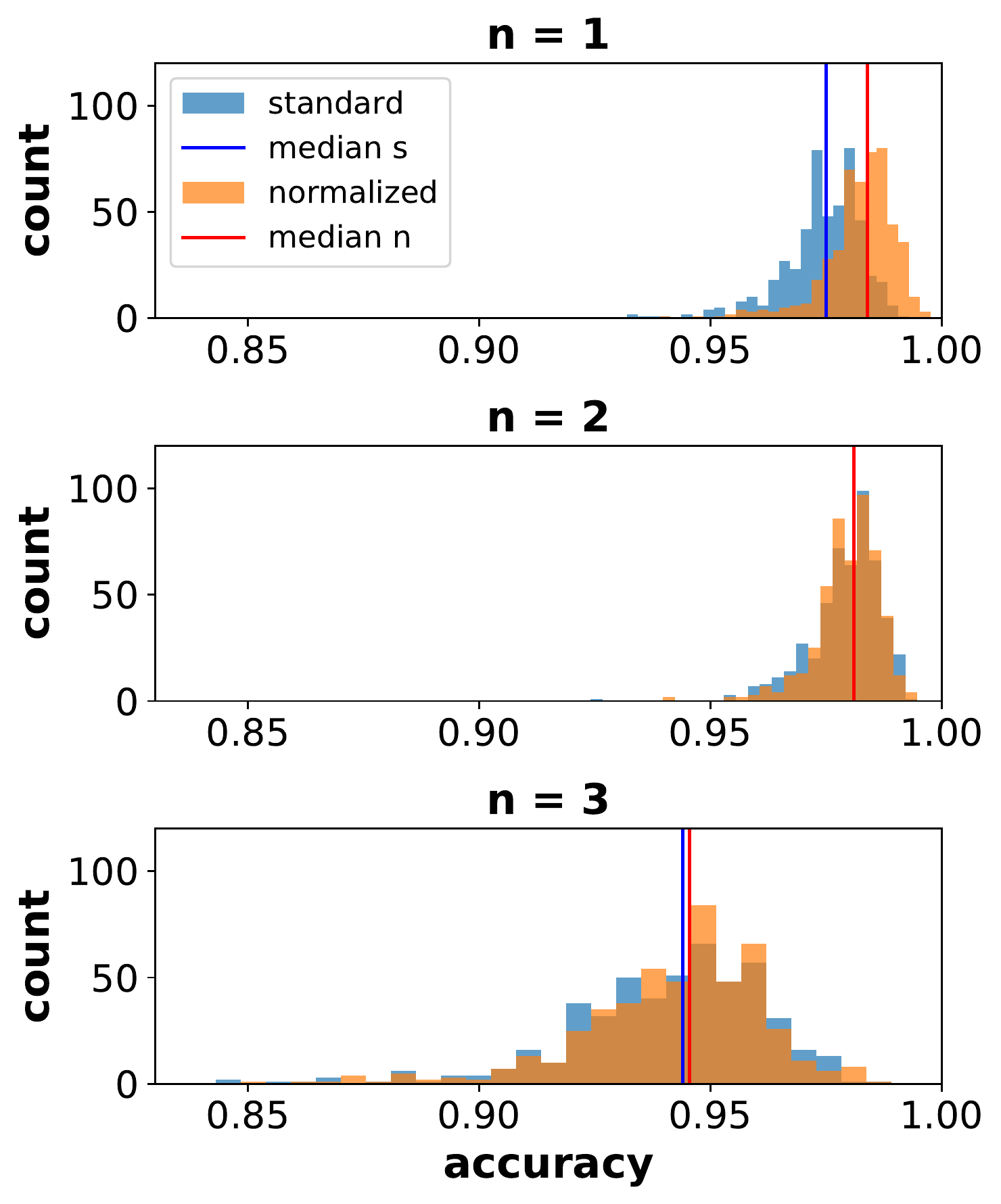}
        \includegraphics[width=0.48\linewidth]{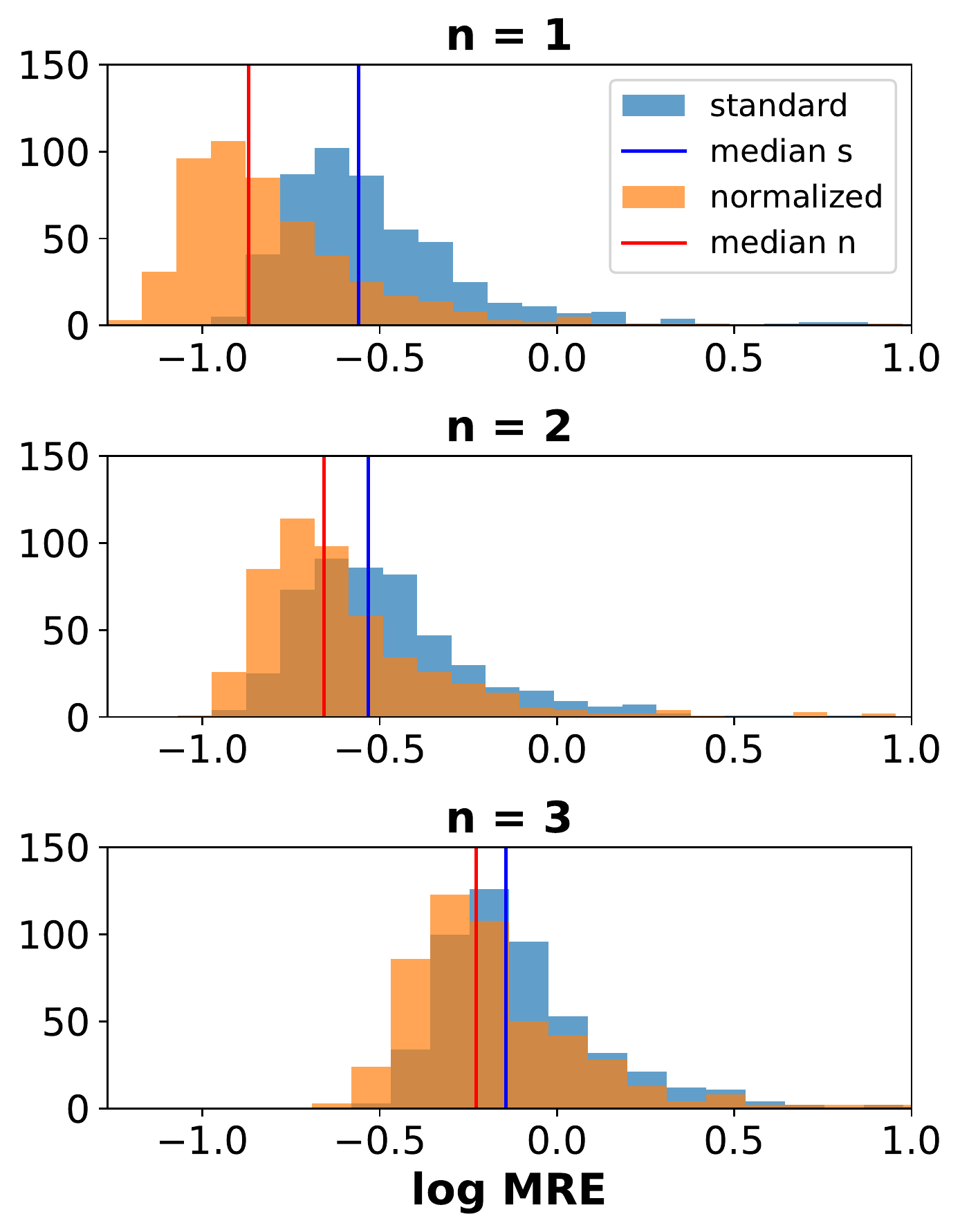}   
    \caption{Accuracy of satisfiability prediction (left)  and  $log_{10}$ of MRE   over all experiments for standard and normalized robustness for sample from \emph{Immigration} (1 dim), \emph{Isomerization} (2 dim) and \emph{Transcription} (3 dim)}
    \label{fig:accuracy_rel_othermodels}
\end{figure}

\begin{figure}[h!]
    \centering
   \includegraphics[width=.50\linewidth]{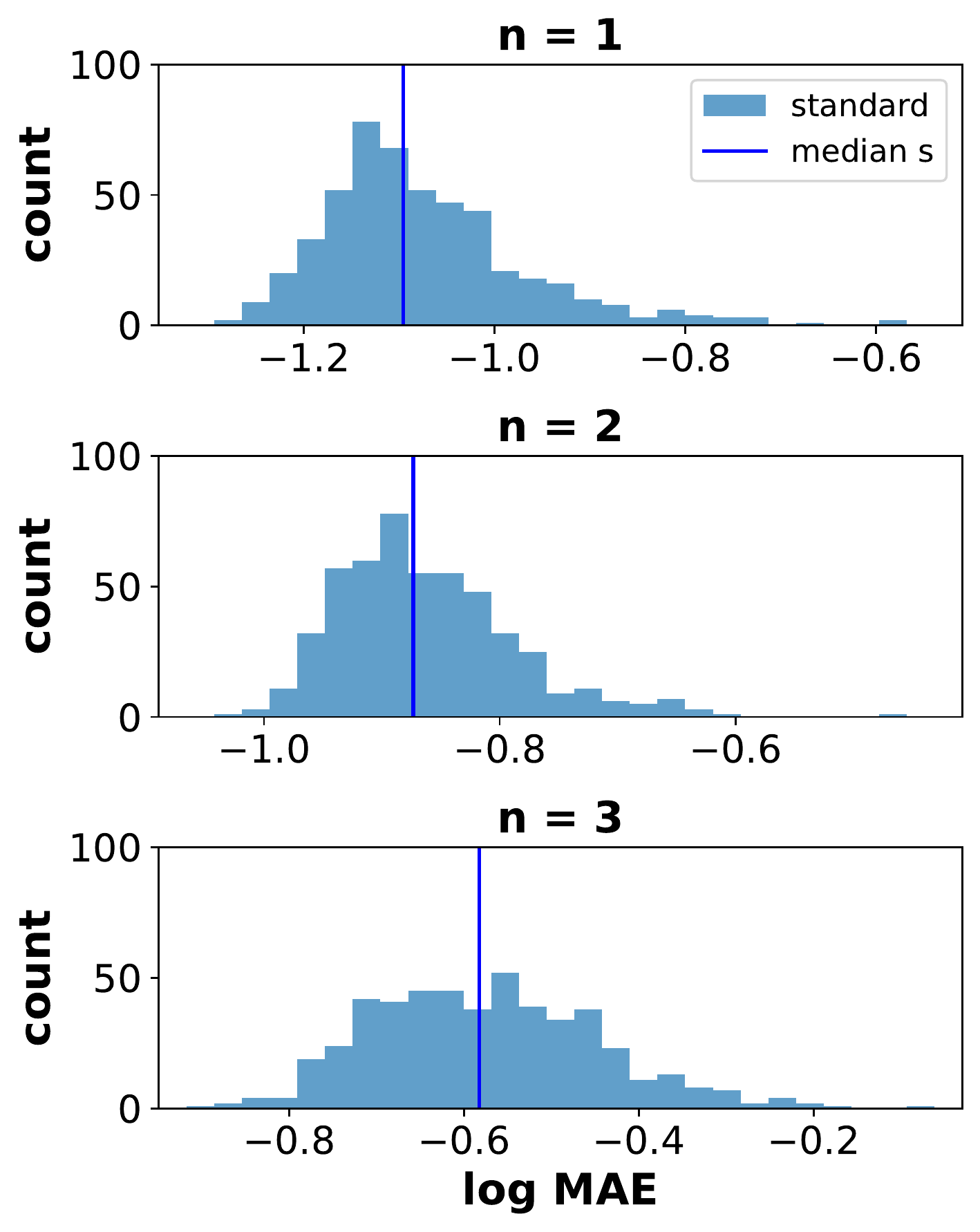}
   \includegraphics[width=.48\linewidth]{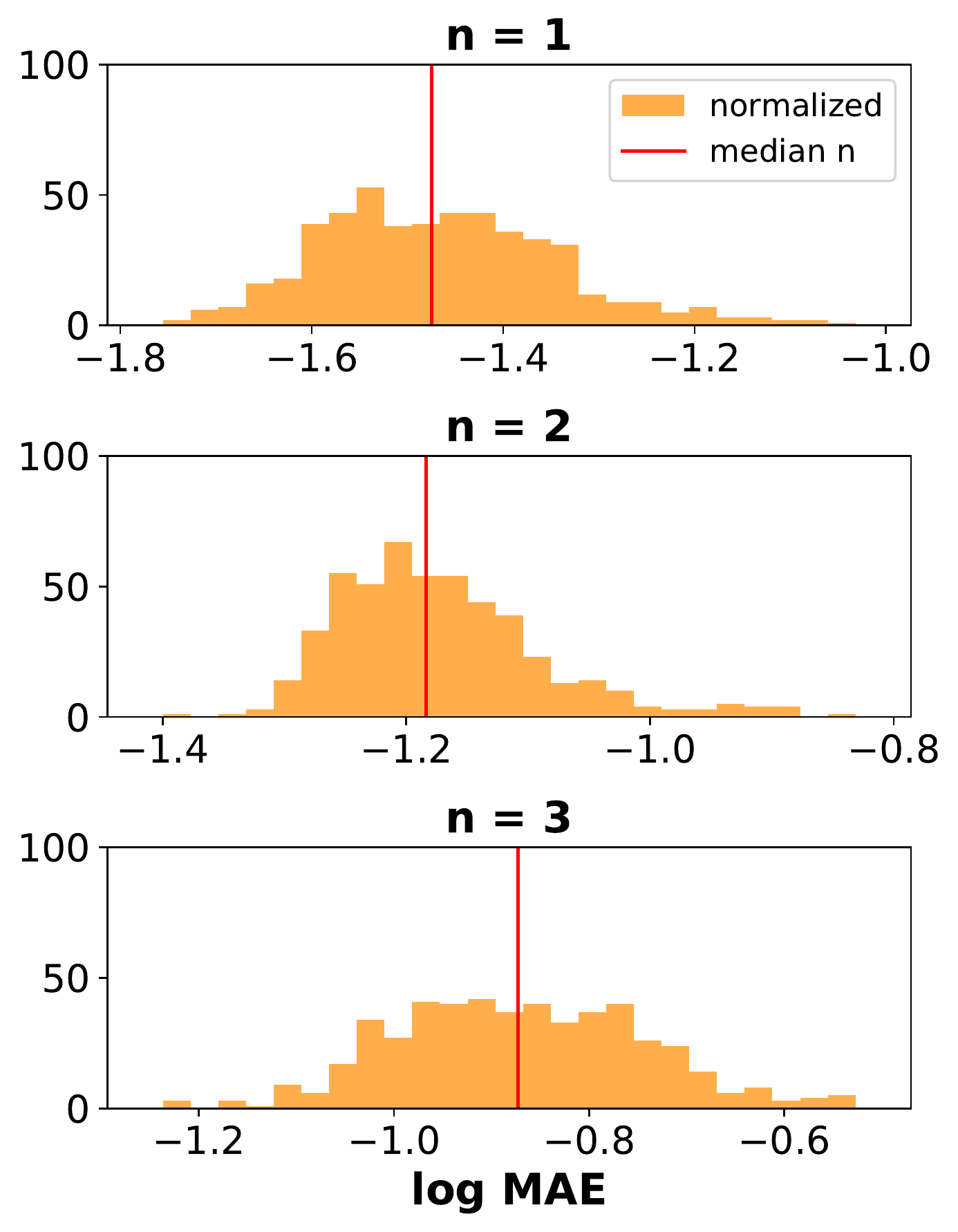}
    \caption{MAE over all experiments for standard and normalized robustness with trajectories sample from \emph{Immigration} (1 dim), \emph{Isomerization} (2 dim) and \emph{Transcription} (3 dim)}
    \label{fig:mae_single_othermodels}
\end{figure}

\begin{figure}[H]
    \centering
       \includegraphics[width=.50\linewidth]{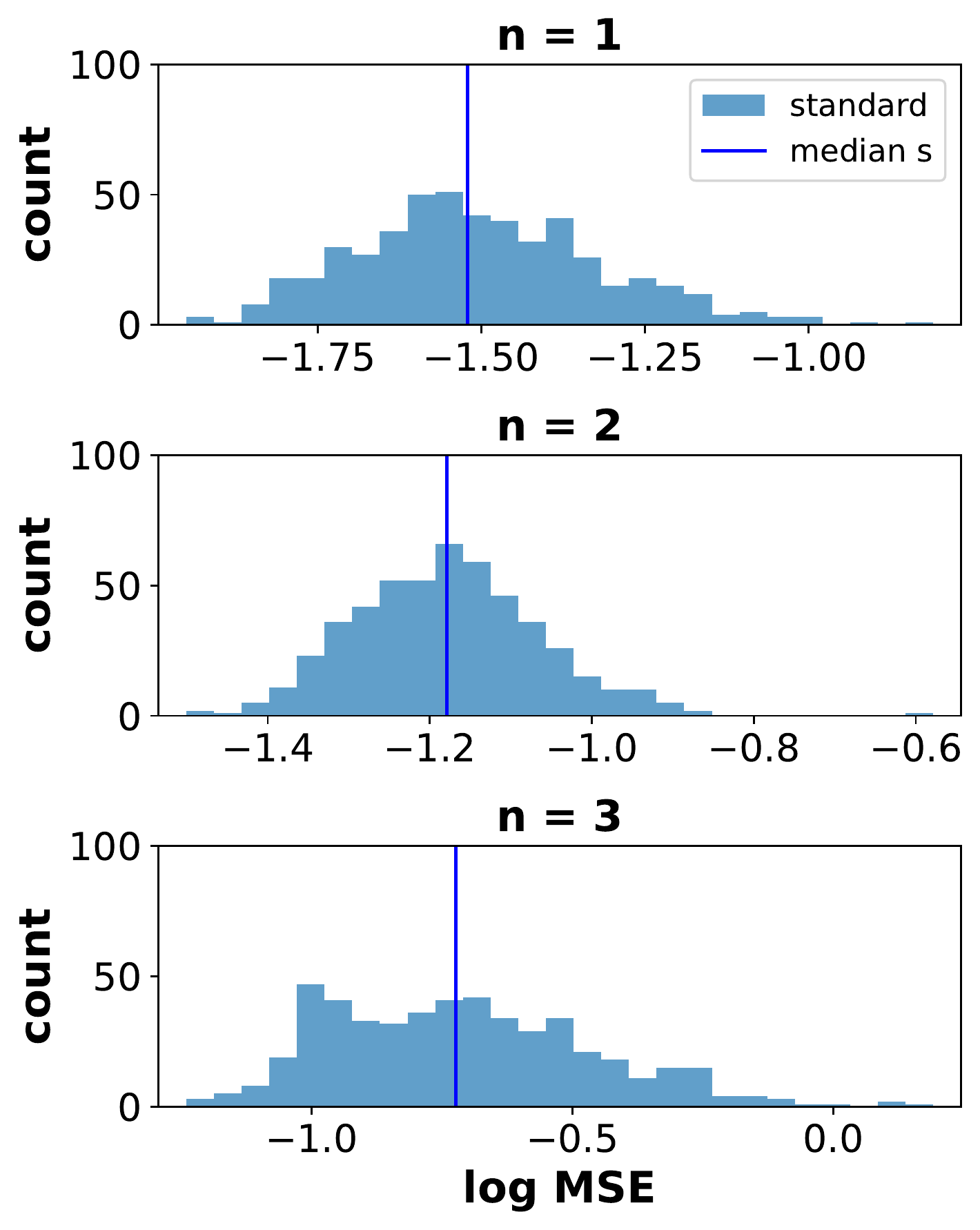}
   \includegraphics[width=.48\linewidth]{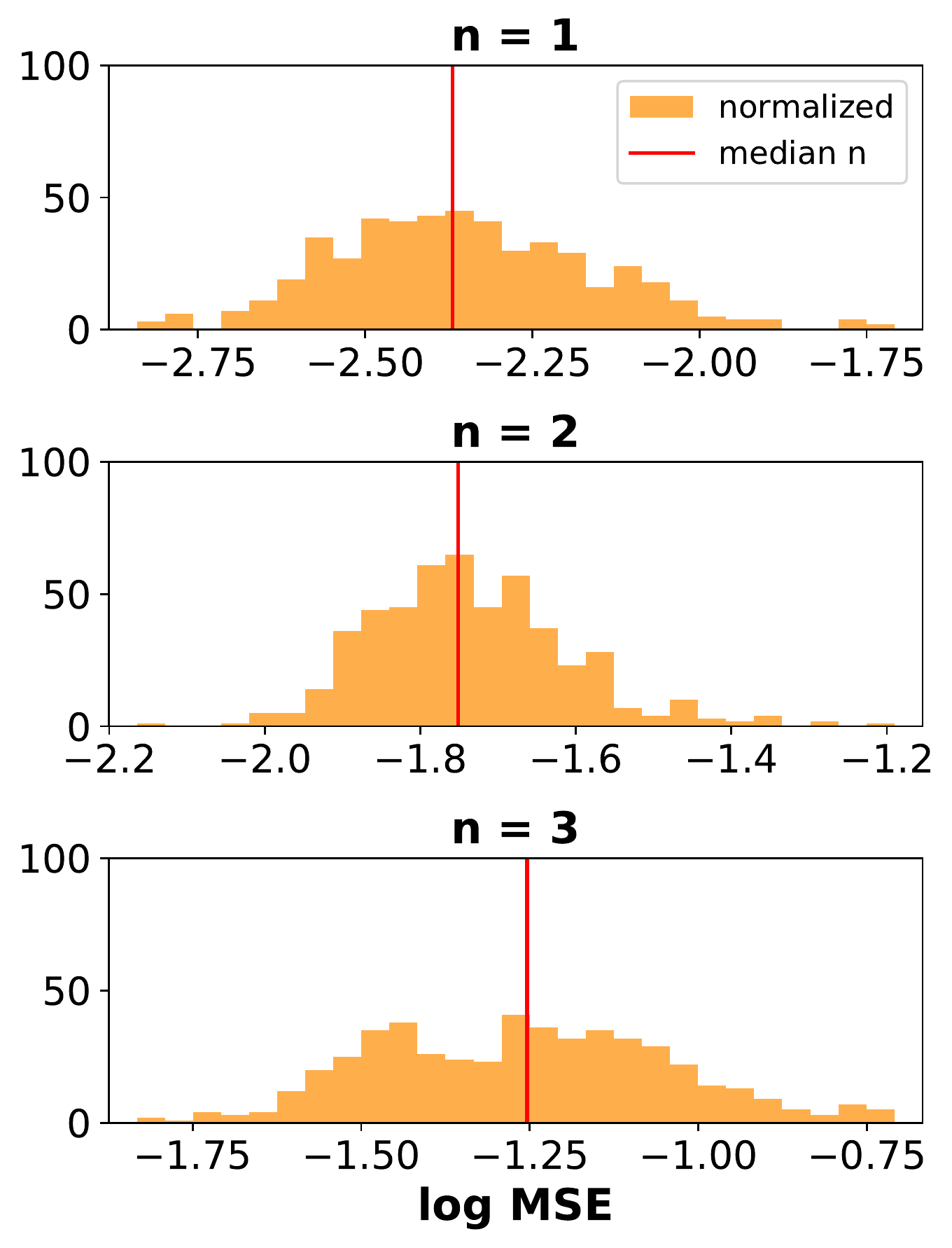}
    \caption{MSE over all 500 experiments standard and normalized robustness on trajectories sample from  \emph{Immigration} (1 dim), \emph{Isomerization} (2 dim) and \emph{Transcription} (3 dim)}
    \label{fig:mse_single_othermodels}
     	\vspace*{-0.2cm}
\end{figure}

  	\vspace*{-0.5cm}
  	\begin{table}[h!]
  	\begin{center}
  		\hspace*{0cm}
\begin{tabular}{l|cc|cc|cc|cc}
\toprule
{} & \multicolumn{2}{c}{MSE} & \multicolumn{2}{c}{MAE} & \multicolumn{2}{c}{MRE} & \multicolumn{2}{c}{ACC} \\
{} & $\rho$ &  $\rhon$ &  $\rho$ &  $\rhon$ &  $\rho$ &  $\rhon$ &  $\rho$ &  $\rhon$ \\
\midrule
immigration   &  0.0302 &   0.00427 &  0.0803 &   0.0335 &  0.404 &   0.278 &  0.975 &   0.984 \\
isomerization &  0.0663 &   0.0177 &  0.134 &   0.0654 &  0.333 &   0.228 &  0.981 &   0.981 \\
transcription &  0.189 &   0.056 &  0.261 &   0.134 &  0.957 &   0.635 &  0.944 &   0.9453 \\
\bottomrule
\end{tabular}
    \vspace{0.3cm}
\caption{Median for MSE, MAE, MRE and ACC of 500 experiments for standard and normalized robustness on trajectories sample from \emph{Immigration} (1 dim), \emph{Isomerization} (2 dim) and \emph{Transcription} (3 dim)}
  		\label{tab:result_single_othermodels}
  	\end{center}
  \end{table}

   \begin{table}[H]
        	\vspace*{-0.3cm}
  	\begin{center}
  		\hspace*{0cm}
\begin{tabular}{l|cc|cc|cc|cc|cc}
\toprule
{} & \multicolumn{2}{c}{5perc} & \multicolumn{2}{c}{1quart} & \multicolumn{2}{c}{median} & \multicolumn{2}{c}{3quart} & \multicolumn{2}{c}{95perc} \\
{} &   $\rho$ &  $\rhon$ &   $\rho$ &  $\rhon$ &   $\rho$ &  $\rhon$ &   $\rho$ &  $\rhon$ &   $\rho$ &  $\rhon$ \\
\midrule
immigration   & 0.00531 & 0.00257& 0.0266 & 0.0124 & 0.0635 & 0.0317 & 0.171 & 0.0982 & 1.05 & 0.706 \\
isomerization & 0.00297 & 0.00207 & 0.0149 & 0.0109 & 0.0388 & 0.0305 & 0.118 & 0.106 & 0.831 & 0.632 \\
transcription & 0.00721 & 0.00536 & 0.0418 & 0.0313 & 0.127 & 0.0965 & 0.400 & 0.304 & 2.37 & 1.62 \\
\midrule
immigration   & 0.00386 & 0.00162 & 0.0186 & 0.00758 & 0.0422 & 0.0166 & 0.0940 & 0.0373 & 0.300 & 0.137 \\
isomerization & 0.00546 & 0.00192 & 0.0273 & 0.00999 & 0.065 & 0.0266 & 0.157 & 0.0753 & 0.522 & 0.275 \\
transcription & 0.0105 & 0.00445 & 0.0571 & 0.0251 & 0.150 & 0.069 & 0.354 & 0.172 & 0.973 & 0.530 \\
\bottomrule
\end{tabular}
    \vspace{0.3cm}
\caption{Mean of quantiles for RE and AE of 500 experiments for prediction of the standard and normalised robustness on trajectories sample from \emph{Immigration} (1 dim), \emph{Isomerization} (2 dim) and \emph{Transcription} (3 dim)}
  		\label{tab:quantile_rob_othermodels}
  	\end{center}
  \end{table}

\begin{figure}[H]
    \centering
        \includegraphics[width=0.48\linewidth]{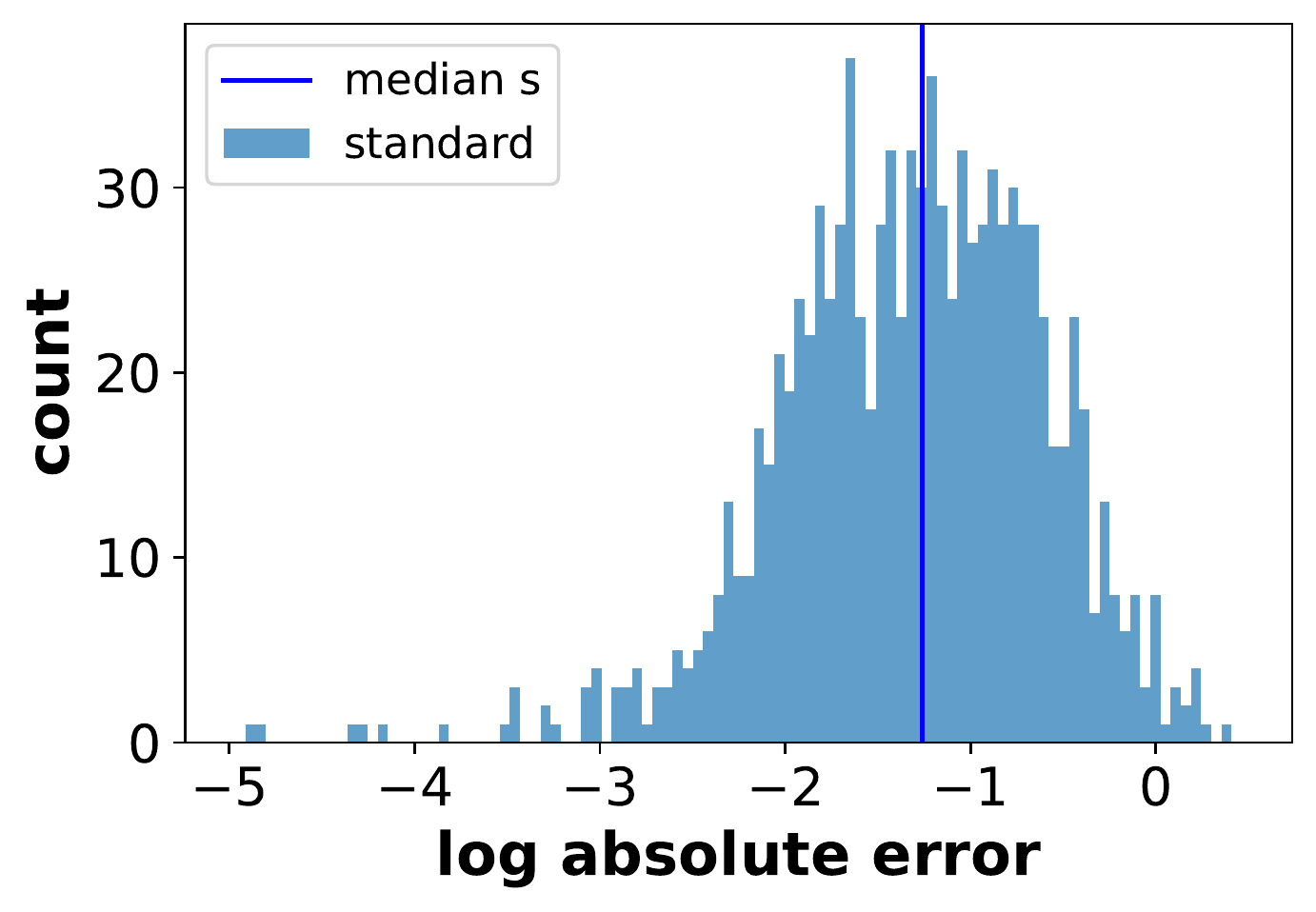}
    \includegraphics[width=0.48\linewidth]{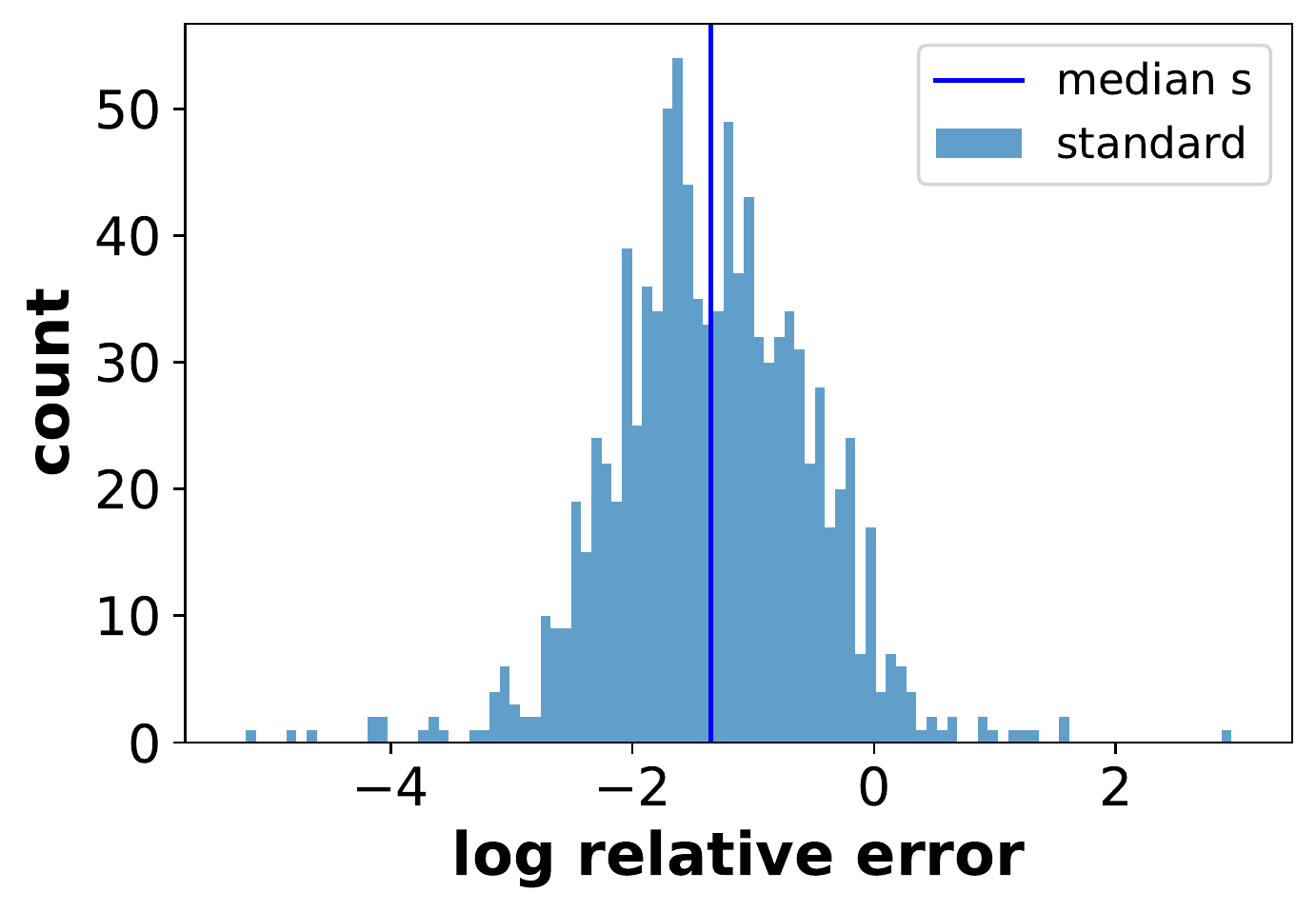}
    \caption{Absolute (left) and relative (right) errors in predicting robustness for a random experiment, with trajectories sample from \emph{Transcription}}
    \label{fig:absrel_single_s}
\end{figure}

\begin{figure}[H]
    \centering
        \includegraphics[width=0.48\linewidth]{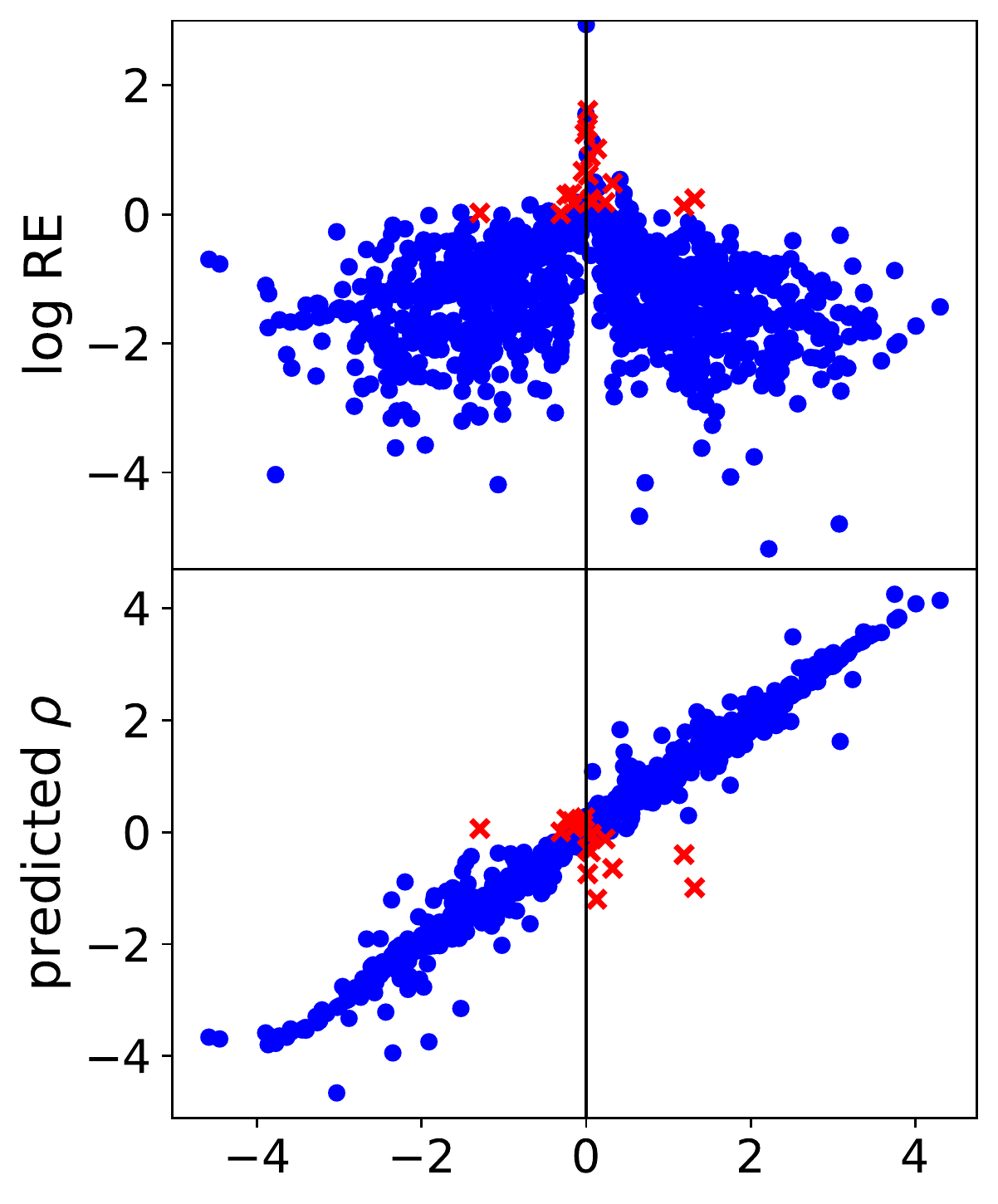}
    \includegraphics[width=0.48\linewidth]{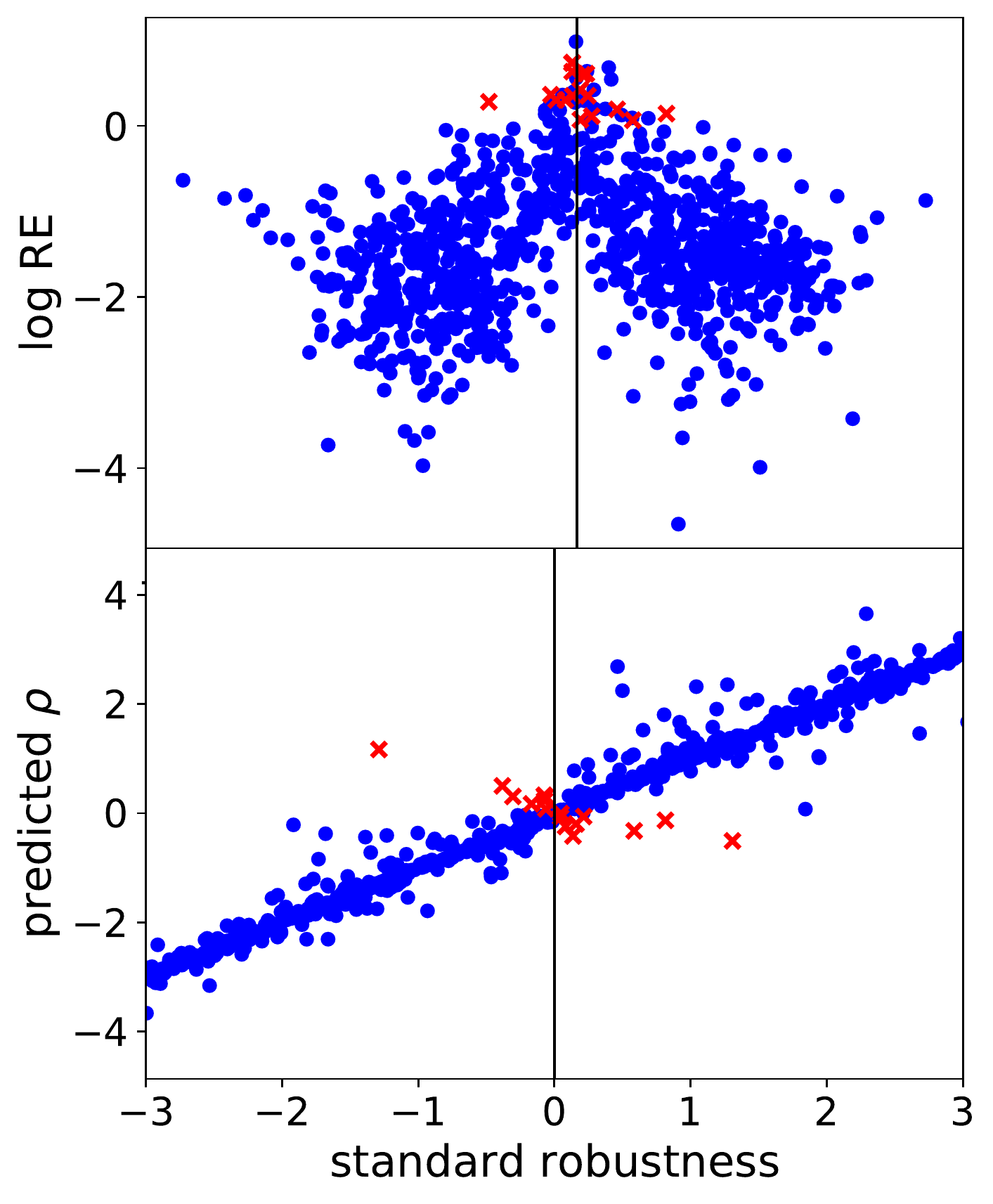}
    \caption{true robustness vs predicted values and RE for a random experiment, with trajectories sample from \emph{Transcription} (left) and \emph{Isomeration} (right)}
    \label{fig:rob_vs_pred_single_s}
\end{figure}

\begin{figure}[H]
    \centering
    \includegraphics[width=0.49\linewidth]{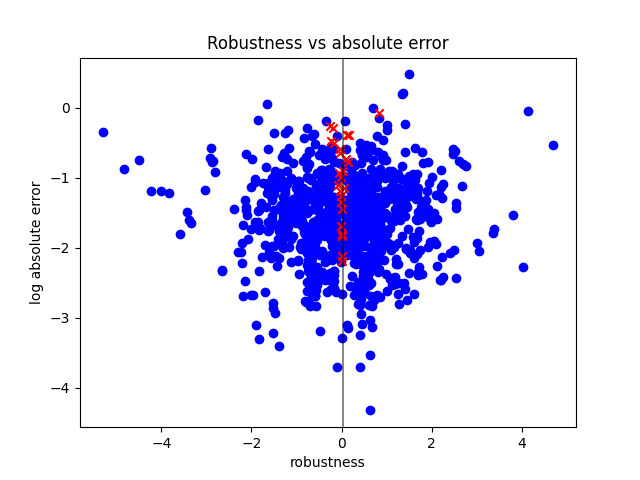}
    \includegraphics[width=0.49\linewidth]{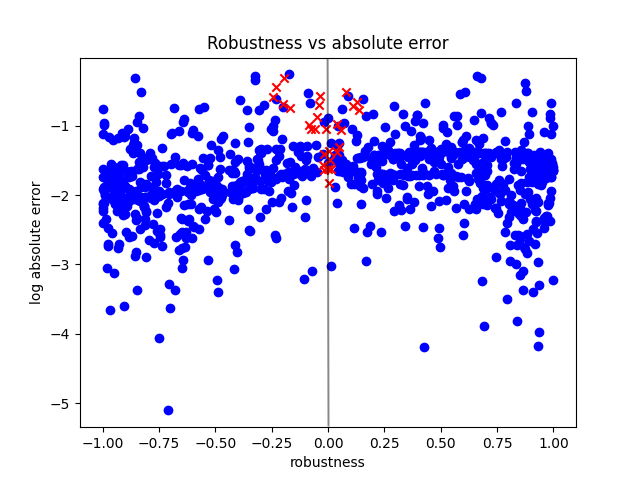}
    \caption{Robustness vs absolute error for a single experiment     for prediction of the standard (left) and normalized (right) robustness 
     on  single trajectories sample from \emph{Immigration}}
    \label{fig:robvsrelerr}
\end{figure}

   \begin{table}[]
  	\begin{center}
  		\hspace*{0cm}
\begin{tabular}{c|cc|cc|cc|cc|cc}
\toprule
{} & \multicolumn{2}{c}{5perc} & \multicolumn{2}{c}{1quart} & \multicolumn{2}{c}{median} & \multicolumn{2}{c}{3quart} & \multicolumn{2}{c}{95perc} \\
{} &   $\rho$ &  $\rhon$ &   $\rho$ &  $\rhon$ &   $\rho$ &  $\rhon$ &   $\rho$ &  $\rhon$ &   $\rho$ &  $\rhon$ \\
\midrule
immigration   & 0.00414 & 0.00149 & 0.0197 & 0.00666 & 0.0447 & 0.0162 & 0.112 & 0.0486 & 0.737 & 0.360 \\
isomerization & 0.00248 & 0.00179 & 0.0124 & 0.00922 & 0.0310 & 0.0257 & 0.103 & 0.0906 & 0.745 & 0.569 \\
transcription & 0.00588 & 0.00412 & 0.0321 & 0.0229 & 0.095 & 0.0712 & 0.305 & 0.240 & 1.82& 1.49 \\
\midrule
immigration   & 0.00270 & 0.000800 & 0.0130 & 0.00370 & 0.0290 & 0.00843 & 0.0615 & 0.0187 & 0.214 & 0.0683 \\
isomerization & 0.00388 & 0.00158 & 0.0193 & 0.00806 & 0.0454 & 0.0208 & 0.117 & 0.0573 & 0.432 & 0.211 \\
transcription & 0.00811 & 0.00339 & 0.0425 & 0.0183 & 0.106 & 0.0487 & 0.252 & 0.122 & 0.704 & 0.376 \\
\bottomrule
\end{tabular}
\caption{Mean of quantiles for RE and AE of 500 experiments for prediction of the standard $\rho$ and normalised $\rhon$ expected robustness on trajectories sample from\emph{Immigration} (1 dim), \emph{Isomerization} (2 dim) and \emph{Transcription} (3 dim)}
  		\label{tab:quantile_rob_othermodels_exp}
  	\end{center}
  \end{table}
\subsubsection{Expected Robustness}
Further results on experiment for prediction of the expected robustness  sampling trajectories on different stochastic models. In terms of error on the robustness itself, we plot in Fig. \ref{fig:mre_exp_othermodels}, \ref{fig:mae_exp_othermodels}, and \ref{fig:mse_exp_othermodels} the distribution of MAE, MRE and MSE for standard and normalized expected robustness on trajectories sample from\emph{Immigration} (1 dim), \emph{Isomerization} (2 dim) and \emph{Transcription} (3 dim) (right) expected robustness vs the predicted one and RE on trajectories sample from  \emph{Isomerization}.
In table \ref{tab:quantile_rob_othermodels_exp} we report the mean of the quantiles for AE and RE of 500 experiments.

\begin{figure}
    \centering
    \includegraphics[width=.48\linewidth]{fig/Histogram_logMRE_othermodels.pdf}
    \includegraphics[width=.48\linewidth]{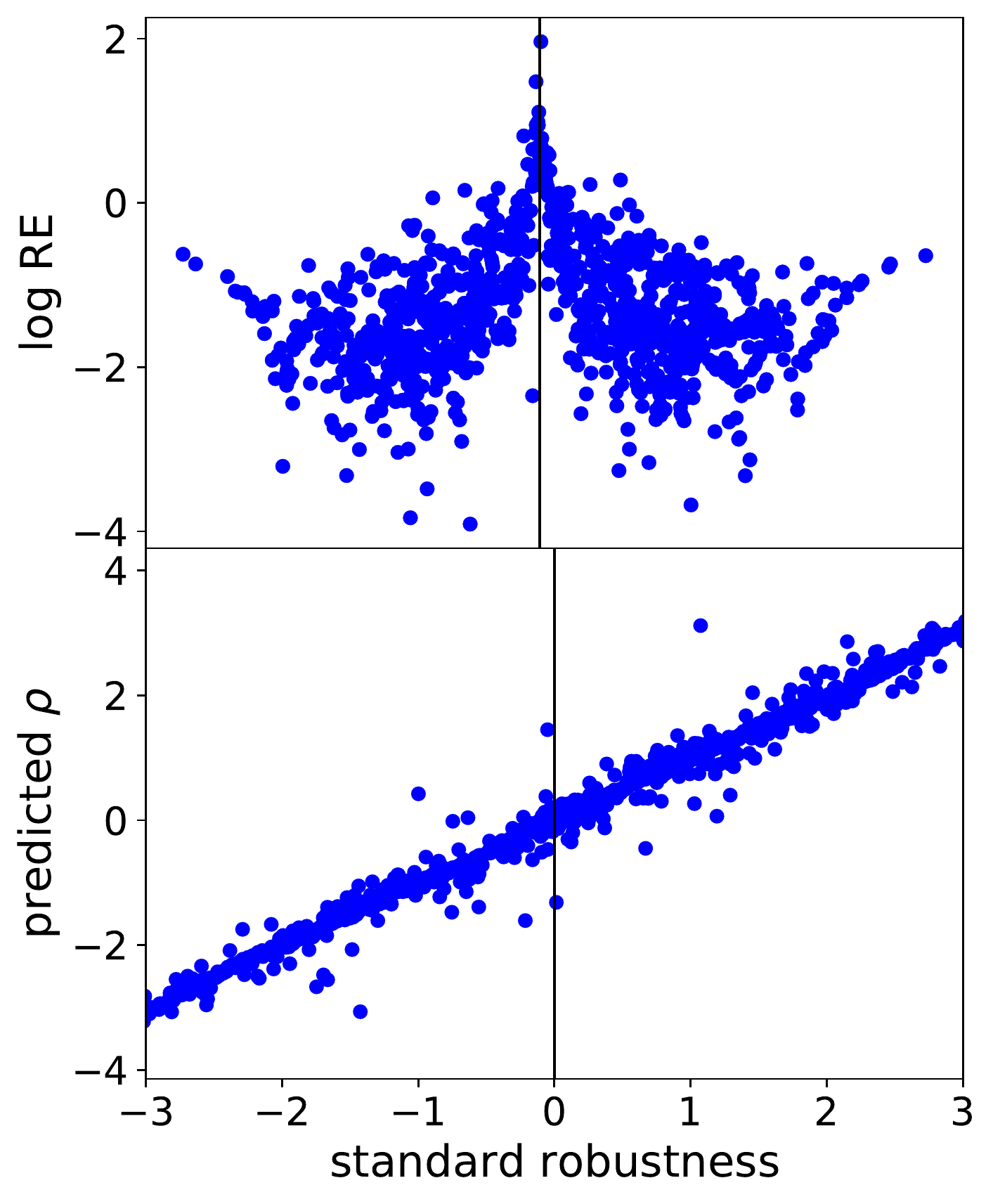}  
    \caption{(left) MRE over all experiments for standard and normalized expected robustness on trajectories sample from\emph{Immigration} (1 dim), \emph{Isomerization} (2 dim) and \emph{Transcription} (3 dim) (right) expected robustness vs the predicted one and RE on trajectories sample from  \emph{Isomerization}}
    \label{fig:mre_exp_othermodels}
\end{figure}

\begin{figure}[H]
    \centering
   \includegraphics[width=.50\linewidth]{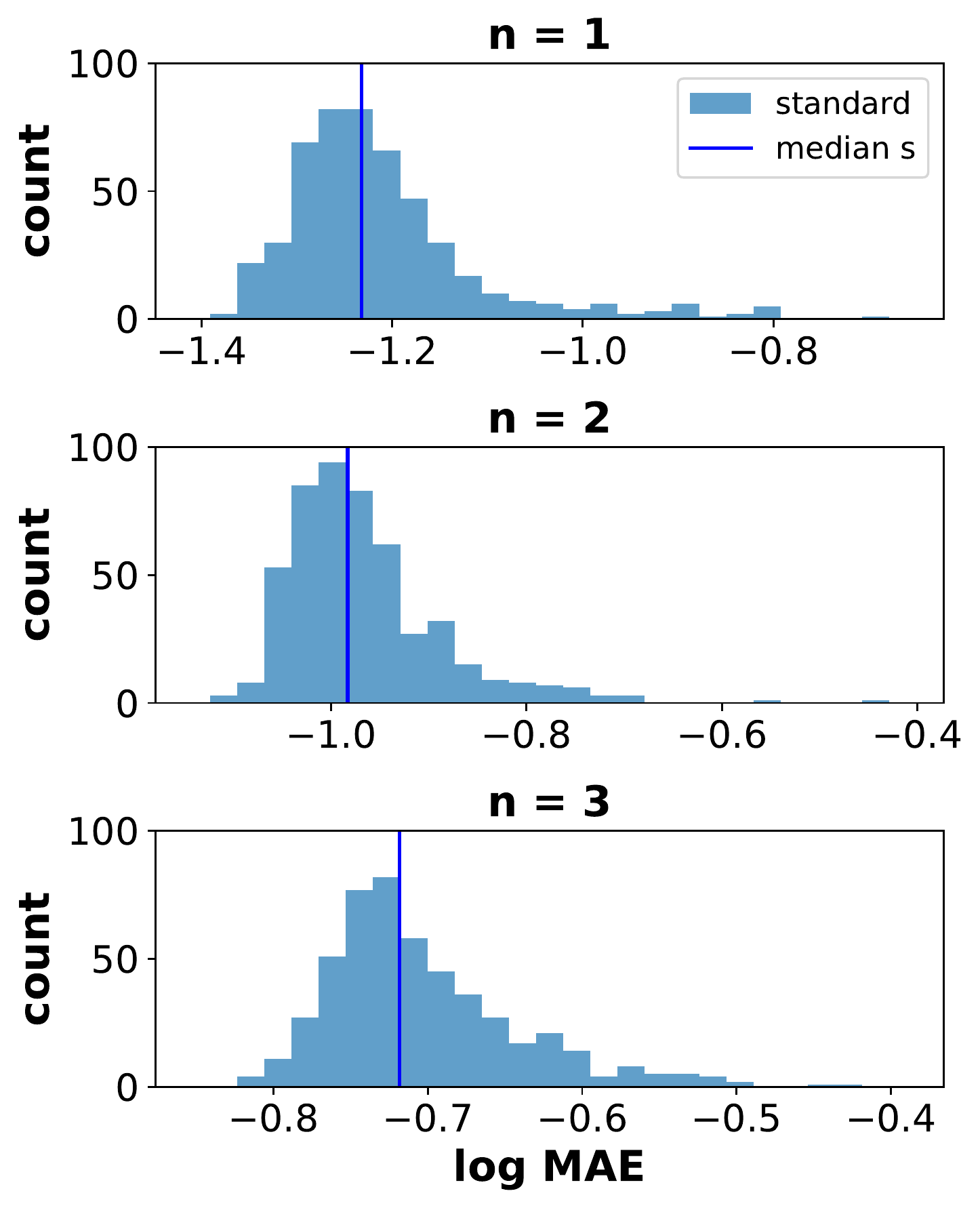}
   \includegraphics[width=.48\linewidth]{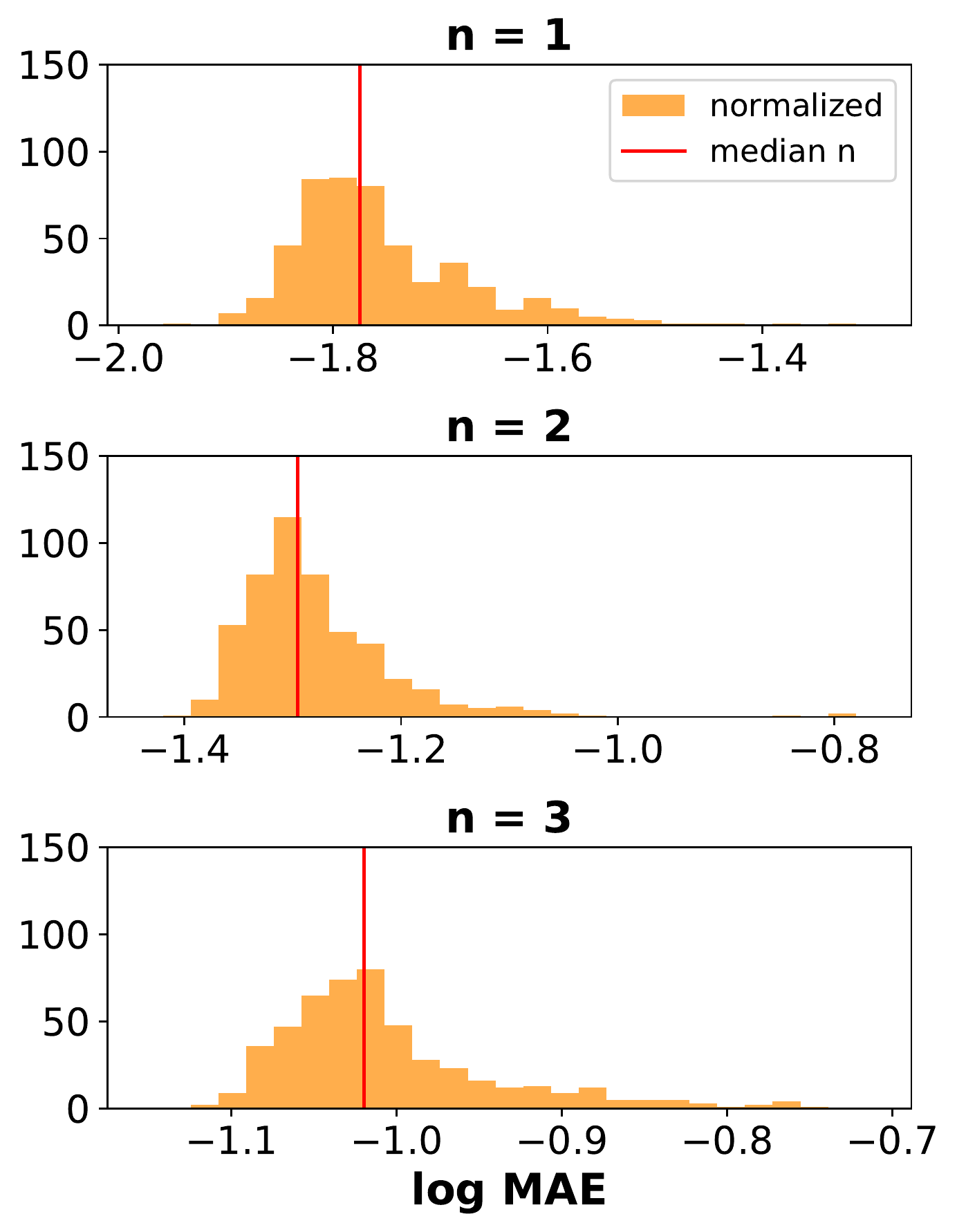}
    \caption{MAE over all 500 experiments for standard and normalized expected robustness with trajectories sample from\emph{Immigration} (1 dim), \emph{Isomerization} (2 dim) and \emph{Transcription} (3 dim)}
    \label{fig:mae_exp_othermodels}
\end{figure}

\begin{figure}[H]
    \centering
   \includegraphics[width=.48\linewidth]{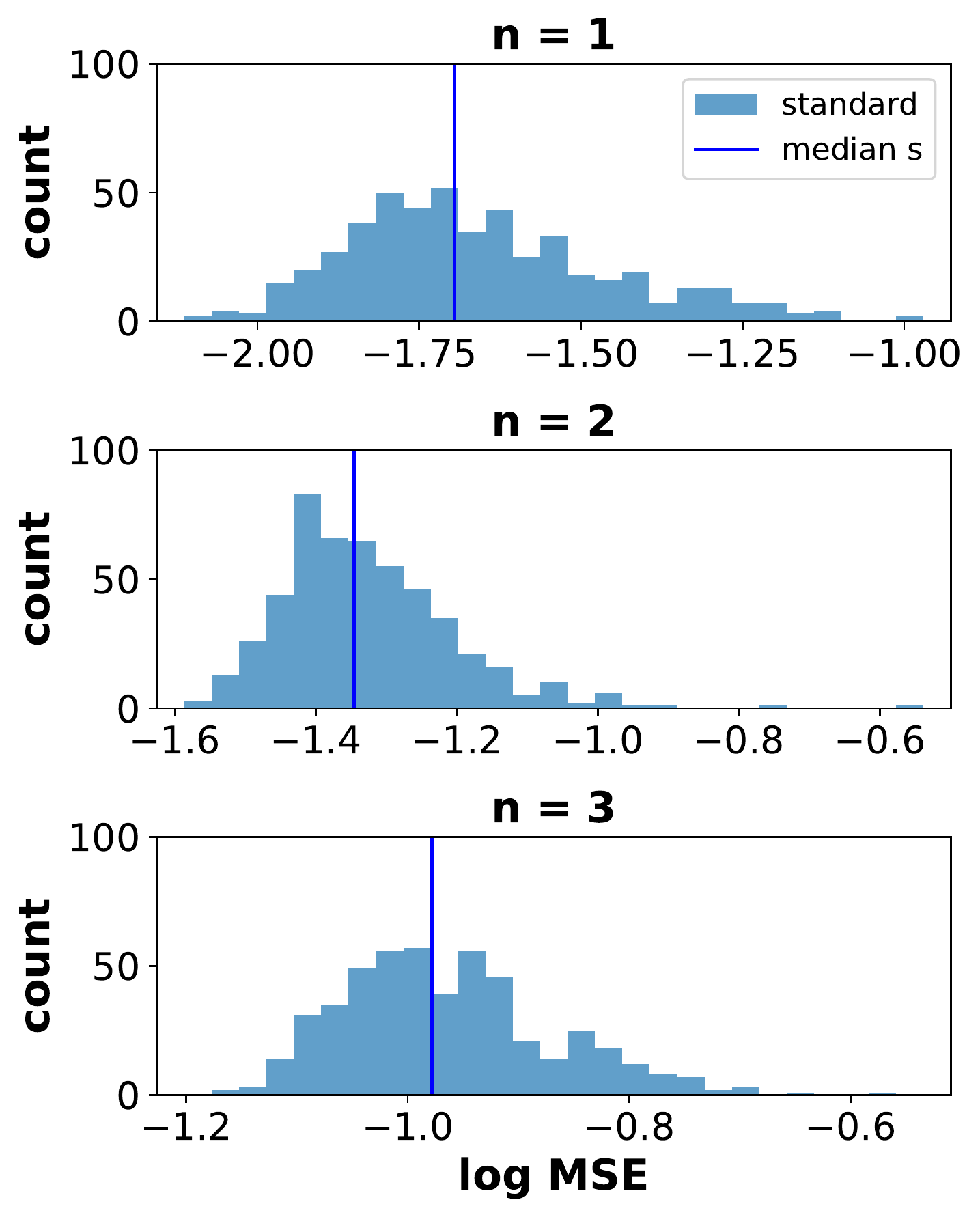}
   \includegraphics[width=.48\linewidth]{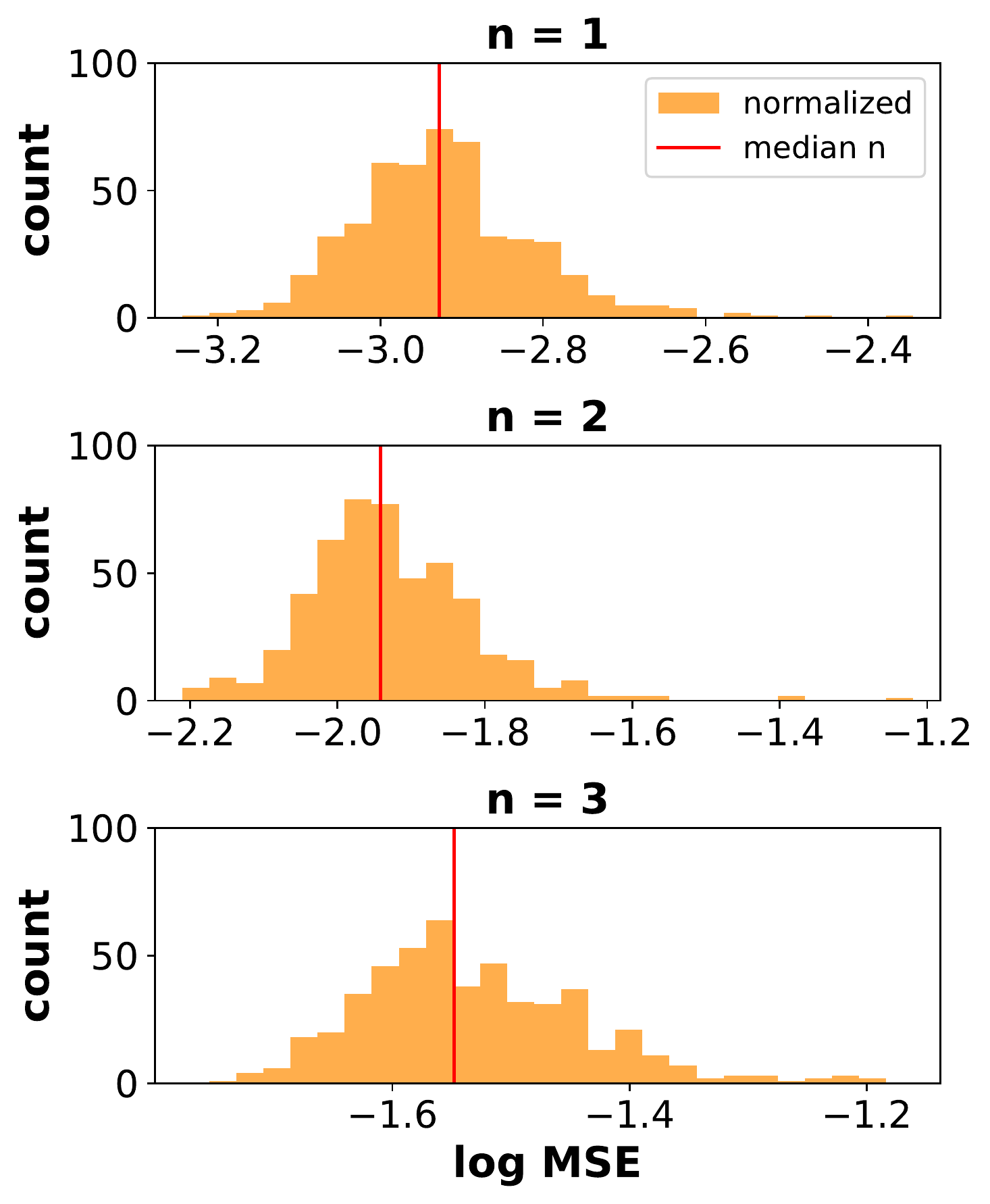}
    \caption{MSE over all 500 experiments for standard and normalized expected robustness with trajectories sample from \emph{Immigration} (1 dim), \emph{Isomerization} (2 dim) and \emph{Transcription} (3 dim)}
    \label{fig:mse_exp_othermodels}
\end{figure}

 \begin{table}[]
  	\begin{center}
  		\hspace*{0cm}
\begin{tabular}{lrrrrrr}
\toprule
{} & \multicolumn{2}{l}{MSE} & \multicolumn{2}{l}{MAE} & \multicolumn{2}{l}{MRE}  \\
{} &  $\rho$ &$\rhon$ & $\rho$ &$\rhon$ &  $\rho$& $\rhon$ \\
\midrule
immigration   &  0.020172 &   0.001182 &  0.058594 &    0.01680 &  0.275879 &   0.134888  \\
isomerization &  0.045105 &   0.011444 &  0.104004 &    0.05069 &  0.293457 &   0.220337  \\
transcription &  0.105103 &   0.028336 &  0.191284 &    0.09552 &  0.718750 &   0.591309    \\
\bottomrule
\end{tabular}
\vspace{0.3cm}
\caption{Median of MSE, MAE and MRE of 500 experiments standard $\rho$ and normalized $\rhon$ expected robustness on trajectories sample from \emph{Immigration} (1 dim), \emph{Isomerization} (2 dim) and \emph{Transcription} (3 dim)}
  		\label{tab:result_avrob_othermodels}
  	\end{center}
  \end{table}


\end{document}